\documentclass[12pt]{article}

\usepackage{amssymb} 
\usepackage{amsthm} 
\usepackage[english]{babel} 
\usepackage{color} 
\usepackage{dsfont}
\usepackage{enumitem} 
\usepackage{float}
\usepackage{geometry}
\usepackage[style=super]{glossaries}
\usepackage{hyperref} 
\usepackage{mathtools} 
\usepackage{mathrsfs} 
\usepackage{longtable}

\usepackage{tikz}

\numberwithin{equation}{section}

\DeclareMathOperator{\sgn}{sgn}
\DeclareMathOperator{\arcosh}{arcosh}
\DeclareMathOperator{\Tr}{Tr}
\DeclareMathOperator{\im}{Im}
\DeclareMathOperator{\re}{Re}
\DeclareMathOperator{\dist}{dist}
\DeclareMathOperator{\Id}{Id}
\DeclareMathOperator{\Str}{Str}
\DeclareMathOperator{\Sdet}{Sdet}
\newcommand*\diff{\;\!\mathrm{d}}
\newcommand*\eto{\;\!\mathrm{e}}

\newtheorem{theo}{Theorem}[section]
\newtheorem{lemma}[theo]{Lemma}

\theoremstyle{definition}
\newtheorem{defi}[theo]{Definition}
\theoremstyle{remark}
\newtheorem*{rem}{Remark}
\newtheorem*{notation}{Notation}

\makeatletter
\renewenvironment{proof}[1][\proofname]{%
  \par\pushQED{\qed}\normalfont%
  \topsep6\p@\@plus6\p@\relax
  \trivlist\item[\hskip\labelsep\bfseries#1\@addpunct{.}]%
  \ignorespaces
}{
  \popQED\endtrivlist\@endpefalse
}
\makeatother

\title{Density of States for Random Band Matrices \\
 in two dimensions
}
\makeglossary

\newglossaryentry{lambda}{name=\ensuremath{\Lambda}, description={$\subset \mathbb{Z}^2$, discrete cube},sort=1a}
\newglossaryentry{W}{name=\ensuremath{W}, description={band width},sort=1c}
\newglossaryentry{H}{name=\ensuremath{H}, description={$:\Lambda\times\Lambda\to\mathbb{C}$ random band matrix}, sort=1b}
\newglossaryentry{rho}{name=\ensuremath{\bar\rho_\Lambda(E) }, description={averaged density of states in finite volume $\Lambda$},sort=1d}
\newglossaryentry{G}{name=\ensuremath{G^+_{\Lambda}(z)}, description={Green's function, $z\in\mathbb{C}$},sort=1e}
\newglossaryentry{Eeps}{name=\ensuremath{E_\varepsilon}, description={$=E+i\varepsilon$ energy with imaginary part},sort=1g}
\newglossaryentry{sc}{name=\ensuremath{\rho_{SC}(E) }, description={Wigner's semicircle law},sort=1f}
\newglossaryentry{J}{name=\ensuremath{J}, description={initial covariance},sort=1h}
\newglossaryentry{I}{name=\ensuremath{\mathcal{I}}, description={energy interval},sort=1i}
\newglossaryentry{alpha}{name=\ensuremath{\alpha}, description={$\in (0,1)$, parameter entering in the definition of the reference volume in the cluster expansion},sort=1j}

\newglossaryentry{a,b}{name=\ensuremath{a,b}, description={$\in \mathbb{R}^\Lambda$ integration variables},sort=2a}
\newglossaryentry{mu}{name=\ensuremath{\diff\mu_J(a,b)}, description={Gaussian measure with covariance $J$},sort=2f}
\newglossaryentry{as,bs}{name=\ensuremath{a_s^\pm,b_s^\pm}, description={saddle points},sort=2b}
\newglossaryentry{mathcalE}{name=\ensuremath{\mathcal{E}}, description={$=\mathcal{E}_r-\mathcal{E}_i = \tfrac{E}{2}-i\sqrt{1-\tfrac{E^2}{4}}$, value of saddle point $a_s^+$},sort=2c}
\newglossaryentry{B}{name=\ensuremath{B}, description={new complex covariance, obtained after contour deformation},sort=2d}
\newglossaryentry{R(a,b)}{name=\ensuremath{\mathcal{R}(a,b)}, description={remainder in the functional integral after contour deformation},sort=2g}
\newglossaryentry{D}{name=\ensuremath{D}, description={diagonal matrix depending on $a,b$},sort=2h}
\newglossaryentry{V(a,b)}{name=\ensuremath{\mathcal{V}(a,b)}, description={effective potential after contour deformation},sort=2i}
\newglossaryentry{V(x)}{name=\ensuremath{V(x)}, description={cubic Taylor remainder},sort=2j}
\newglossaryentry{O(a,b)}{name=\ensuremath{\mathcal{O}(a,b)}, description={local observable, later $\mathcal{O}_{m,n}(a,b)$},sort=2k}
\newglossaryentry{C}{name=\ensuremath{C}, description={new real covariance},sort=2e}
\newglossaryentry{mr,mi}{name=\ensuremath{m_r,m_i}, description={real and imaginary part of complex mass term $1-\mathcal{E}^2$ of $C$},sort=2l}
\newglossaryentry{Is}{name=\ensuremath{I^s}, description={$\subset \mathbb{R}^\Lambda \times \mathbb{R}^\Lambda$, partition of integration domain, $s=1,...5$},sort=2m}
\newglossaryentry{Fmns}{name=\ensuremath{F^{m,n}_s}, description={functional integral with local observable $\mathcal{O}_{m.m}$ restricted to $I^s$},sort=2n}

\newglossaryentry{M}{name=\ensuremath{M}, description={$=(M_j)_{j\in\Lambda}$ set of $2\times 2$ supermatrices},sort=3a}
\newglossaryentry{rhoj}{name=\ensuremath{(\bar\rho_j,\rho_j)_{j\in\Lambda}}, description={set of Grassmann variables}, sort=3b}
\newglossaryentry{muM}{name=\ensuremath{\diff\mu_B(M)}, description={Gaussian measure in both complex and Grassmann variables}, sort=3c}
\newglossaryentry{V(M)}{name=\ensuremath{\mathcal{V}(M)}, description={effective potential depending on the supermatrix $M$}, sort=3d}

\newglossaryentry{Ytilde}{name=\ensuremath{\tilde Y}, description={$(\triangle_0,\tilde\triangle_1,\dots,\tilde\triangle_r)$ generalized polymer}, sort=3i}
\newglossaryentry{T}{name=\ensuremath{T}, description={ordered tree on generalized polymer $\tilde Y$}, sort=3j}
\newglossaryentry{triangle}{name=\ensuremath{\triangle}, description={cube in $\mathbb{Z}^2$ of size $W^2(\ln W)^\alpha$}, sort=3k}
\newglossaryentry{triangle0}{name=\ensuremath{\triangle_0}, description={root cube containing $0$}, sort=3l}
\newglossaryentry{tildetriangle}{name=\ensuremath{\tilde\triangle}, description={$=(\triangle,\triangle',\triangle'')$ generalized cube}, sort=3m}
\newglossaryentry{ijkk'}{name=\ensuremath{i,k,k',j}, description={$\in\mathbb{Z}^2$ indices summed over  $i\in\triangle',j\in\triangle'',k'\in\triangle,k\in$ ``old'' cubes}, sort=3n}

\newglossaryentry{Gq}{name=\ensuremath{G_q(s)}, description={propagator depending only on $s_1,\dots,s_q$}, sort=3h}
\newglossaryentry{C(s)}{name=\ensuremath{C(s)}, description={interpolated real covariance $C(s)_{ij}=s_{ij}C_{ij}$}, sort=3f}
\newglossaryentry{B(s)}{name=\ensuremath{B(s)}, description={interpolated complex covariance $(C(s)^{-1}+i\sigma_Em_i^2)^{-1}$}, sort=3g}
\newglossaryentry{sp}{name=\ensuremath{s_p}, description={inductively introduced interpolation parameters}, sort=3e}

\begin{document}
\author{Margherita Disertori\footnote{Institute for Applied Mathematics
\& Hausdorff Center for Mathematics, 
University of Bonn, Endenicher Allee 60, 53115 Bonn, Germany
\newline
E-mails: disertori@iam.uni-bonn.de,  lager@iam.uni-bonn.de}
 \ and \ 
Mareike Lager\footnotemark[2]}

\maketitle

\begin{abstract}
We consider a two dimensional random band matrix ensemble,
in the limit of infinite volume and fixed but large band width $W$. For this model we rigorously prove 
smoothness  of the averaged density of states. We also prove that the resulting expression
coincides with Wigner's semicircle law with a precision $W^{-2+\delta },$ where $\delta\to 0$
when $W\to \infty.$ The proof uses the
supersymmetric approach and extends  results by Disertori, Pinson and Spencer 
\cite{disertori-pinson-spencer} from three to two  dimensions. 
\end{abstract}

\section{Introduction and main result}
\paragraph{General setting.}
It is a well known fact  that conducting properties
of disordered materials can be related, in the context of quantum mechanics, 
to the statistics of eigenvalues and eigenvectors of certain random matrix ensembles 
\cite{anderson-1958}.
The most famous example are random Schr\"odinger operators, whose lattice version 
is characterized by the random matrix $H_{\Lambda }:\Lambda \times \Lambda \to \mathbb{R}$, 
on  a subset $\Lambda$ of $\mathbb{Z}^{d},$ 
defined by $H_{\Lambda } =-\Delta+\lambda V,$
where $-\Delta$ is the discrete Laplacian, $V$ is a diagonal matrix with random diagonal
entries  and $\lambda >0$ is a parameter encoding the strength of the disorder.
The entries $V_{j}$ are generally assumed to be independent identically distributed 
(for instance Gaussian). 
As $\Lambda\uparrow\mathbb{Z}^{d}$, this model exhibits a localized phase for all $\lambda $ 
in $d=1$ and at large disorder $\lambda\gg 1$
in $d\geq 2$. The localized phase is conjectured to hold also at weak disorder $\lambda\ll 1$
in $d=2$, while a phase transition 
is conjectured in $d\geq 3$.
Though the localized phase is well understood, the weak disorder regime in $d\geq 2$
remains an open problem. For a review of definitions and results see for instance 
\cite{kirsch-2008}.

\paragraph{}
Another relevant model in this context is  the random band matrix (RBM) ensemble, characterized
by   a self-adjoint matrix $\gls{H}_{\Lambda }:\Lambda \times \Lambda \to \mathbb{K}$, $\mathbb{K}=\mathbb{R},\mathbb{C}$,
whose entries are \emph{all}  independent (up to self-adjointness) random variables \emph{not} 
identically distributed with negligible entries outside a band of width $W\geq 0$, i.e. 
$|H_{ij}|\ll 1$, with large probability, when $|i-j|> W$.  
As in the case of random Schr\"odinger, when  $\Lambda\uparrow\mathbb{Z}^{d}$, band matrices are believed 
to exhibit a phase transition in $d\geq  3$ between a localized phase at small $W$ 
and an extended phase at large $W$, while the localized phase is conjectured to hold
for all $W$ in $d= 1,2$.

Two important examples of RBM are  the 'smooth Gaussian'  and the $N$-orbital model.
In the first case, the matrix elements are Gaussian: 
\begin{align*}
H_{ii}\sim \mathcal{N}_{\mathbb{R}} (0,J_{ii}),
\qquad H_{ij}\sim \mathcal{N}_{\mathbb{C}} (0,J_{ij}), \quad \text{for }
i<j,
\end{align*}
where $<$ denotes an order relation on $\mathbb{Z}^{d},$ and the band structure is
  encoded in  the covariance  
$J_{ij}=J_{ji}=f(|i-j|)$  decaying to zero when
$|i-j|\gg W$.  In the second case the covariance $J_{ij}$ is short range, for instance $J= \Id+a\Delta$ 
for some $a>0$, but each matrix element $H_{ij}$ is itself a $N\times N$ matrix with i.i.d. entries 
$(H_{ij})_{\alpha \beta }\sim  \mathcal{N}_{\mathbb{C}} (\frac{1}{N}J_{ij})$  $\forall i<j$, or $i=j$ and 
$\alpha<\beta$ and $(H_{ii})_{\alpha \alpha  }\sim  \mathcal{N}_{\mathbb{R}} (\frac{1}{N}J_{ii})$. 
The band width in this case is $W=2N$. 
These models are difficult to analyse with standard random matrix tools, since
the probability distribution is not invariant under unitary rotations.
At the moment, most  results available
deal  with the one dimensional case 
(cf. \cite{CMI-1990,CCGI-1993,Schenker2009,sodin-2010,georges-peche,Shcherbina2014,bao-erdoes-2015,
Pchelin2015}). Recently localization
at strong disorder (i.e. small band width) in any dimension was proved for a large class of $N-$orbital models by
Peled, Schenker, Shamis and Sodin \cite{SchenkerPeledShamisSodin2016}.

\paragraph{}
In this paper we consider the density of states
$\rho_\Lambda(E):= \frac{1}{|\Lambda|}\sum_{j} \delta_{\lambda_{j }} (E),$ where  $E\in\mathbb{R}$ is the energy
and  $\lambda_{j}$ are the (random) eigenvalues of $H.$
Since the probability distribution is translation invariant and the bandwidth $W$ is
fixed, by standard ergodicity arguments (see \cite{PasturFigotin1992}) this measure is non random in the
thermodynamic limit.  Similar results hold in $d=1$ also for the (non ergodic) case when $W$ diverges together with the matrix size
\cite{MPK-199}.
We will therefore concentrate on the  \emph{averaged density of states} (DOS) defined by
\begin{align}\label{eq:model:averagedos}
\gls{rho} \coloneqq \mathbb{E}[ \rho_\Lambda(E)] 
\coloneqq \frac{1}{|\Lambda|} \mathbb{E}\left [ \sum_{j} \delta_{\lambda_{j }} (E) \right ]
= -\frac{1}{\pi|\Lambda|}\lim_{\varepsilon\to 0}  \mathbb{E}[\im \Tr G^+_{ \Lambda } (E_{\varepsilon})],
\end{align}
where $E_{\varepsilon }:= E+i\varepsilon,$ with $\varepsilon>0,$
$\mathbb{E}$ denotes the average with respect 
to the probability distribution of $H$  and,
 for any $z\in \mathbb{C},$ the Green's function (or resolvent) is defined by
\begin{align*}
\gls{G}  \coloneqq  (z\cdot \mathds{1}-H)^{-1} =(z-H)^{-1}.
\end{align*}
By standard analyticity arguments, the limit $\varepsilon\to 0$ in \eqref{eq:model:averagedos} above
exists and is finite for Lebesgue a.e. $E\in \mathbb{R}$ (see for ex. \cite[App. B]{AizenmanWarzel2015}).

In dimension larger than one rigorous results on the density of states for RBM were obtained by
\cite{disertori-pinson-spencer,CFGK87}, and more recently by  \cite{SchenkerPeledShamisSodin2016}, 
based on seminal work by Wegner \cite{Wegner1981}.
In the case of the classical GUE ensemble, corresponding to $d=1,$ $\Lambda=1,\dots ,N$
and $J_{ij}=1/N$ $\forall i,j$,  the density of states, in the limit $N\uparrow\infty$
is given a.s. by Wigner's famous semicircle law
\begin{align}
\label{eq:intro:semicirclelaw} 
\gls{sc} =
\begin{cases}
\frac{1}{\pi} \sqrt{1-\frac{E^2}{4}} & \text{if } |E|\leq 2, \\ 
0 &  \text{if } |E| > 2.
\end{cases} 
\end{align}

\paragraph{The model.}
We consider a Gaussian complex RBM ensemble defined on a \emph{two dimensional} discrete cube 
$\gls{lambda}\subset\mathbb{Z}^2,$ centered at the origin,
with covariance 
\begin{align}
\label{eq:model:J}
\gls{J}_{ij}\coloneqq (-\gls{W}^2\Delta +\mathds{1})^{-1}_{ij},
\end{align}
where $-\Delta\in\mathbb{R}^{\Lambda\times\Lambda}$ is the discrete Laplace operator on $\Lambda$
with periodic boundary conditions, i.e.  $(x,-\Delta x) = \sum_{i \sim j} (x_i-x_j)^2$
for any $x\in\mathbb{R}^\Lambda$  and we write $i\sim j$ when $i$ and $j$ are nearest neighbors
in the torus $\mathbb{Z}^2/(\Lambda\mathbb{Z}^2)$.
The parameter $W\in\mathbb{R}$ is large but fixed.
Note that  $J_{ij}$ is exponentially small for distances 
$|i-j|\geq W.$ Hence all matrix elements outside a (two-dimensional) band of width $W$
centered around the diagonal are small with high probability and
this model describes a ``smoothed-out'' version of a band matrix ensemble with band width $W$.

For this model, the averaged density of states \eqref{eq:model:averagedos} exists
and takes a finite value for all $E\in \mathbb{R}.$
In the \emph{three dimensional} case,  Disertori, Pinson and Spencer 
\cite{disertori-pinson-spencer} derived an explicit representation of the function $\bar \rho_{\Lambda }(E),$
in the bulk of the spectrum, in terms of a convergent sum of certain integrals. 
Using this representation they obtained  detailed information 
on the function  $E\mapsto \bar \rho_{\Lambda }(E)$ and its derivatives, in the limit $\Lambda \to\mathbb{Z}^{3}$
for fixed but large band width $W$ (weak disorder regime). In particular they proved that the limit expression
coincides with Wigner's semicircle law with a precision
$1/W^{2}.$ Similar results were obtained for the $N-$orbital model  in \cite{CFGK87} 
in the case of dominant diagonal disorder.  In this paper we construct an extension of
the representation derived in \cite{disertori-pinson-spencer} to the two dimensional case, and
use it to derive precise information (such as smoothness) on the function
$E\mapsto \bar \rho_{\Lambda }(E)$ for energies in the bulk of the spectrum,
in the limit $\Lambda \to\mathbb{Z}^{2}$.

The proof in  \cite{disertori-pinson-spencer} used the so-called \emph{supersymmetric approach} (SUSY), 
pioneered by K. Efetov  \cite{efetov-adv}
based on seminal work by Sch\"afer and Wegner  \cite{schafer-wegner-1980,wegner-1979} 
and further developed (among others)  by Y. Fyodorov, A. Mirlin and M. Zirnbauer
\cite{mirlin1991,Mirlin-Fyodorov,littelmann-sommrs-zirnbauer}.
A good introduction to random matrix theory and SUSY can be found in \cite{haake}.
This is a duality transformation 
that allows one to write averages in $H$ as new integrals where a saddle approximation
may be justified: $\mathbb{E} [f (H)]= \int  \tilde{f} (M) \eto^{-F (M)} \prod_{j\in \Lambda } \diff M_{j},$
where $M_{j}$ is a small matrix  containing both complex
and Grassmann (odd) elements, $F$ can be seen as the free energy functional in some statistical 
mechanical model, and $\tilde{f}$ is the new observable.
The Grassmann variables can be always integrated out
exactly  (though the combinatorics involved may be quite difficult) and the resulting measure 
is complex but {\em normalized}. The measure $\exp(-F (M))\diff M$ depends on the probability measure $P$
but also on the observable $f$, and has internal (odd) symmetries, 
{\em inherited by the observable only}. 
Different rigorous versions of the SUSY approach have been tested on the 
standard matrix ensembles GUE and GOE (where other techniques also apply) 
\cite{disertori2004,shamis2013}. 
When  $\Lambda\uparrow \mathbb{Z}^{d}$, the integral is expected to concentrate
near the saddle manifold determined by $\partial_{M}F (M)=0$.
This reduces to   \emph{isolated points} in the case of the DOS, and the main difficulty 
is to obtain estimates on the fluctuations that are uniform in the volume.
In   \cite{disertori-pinson-spencer}  the dual representation is 
studied via a complex translation coupled with a  cluster expansion, which
effectively factorizes the integral over regions of volume $W^{3}.$
A key feature of the dual representation, in three dimensions, is the presence of 
a \emph{double well} structure, with one well  suppressed  
by an exponential factor $\exp(-W),$ inside each region. 
In the two dimensional case, the second well is only weakly suppressed, and
the arguments used in   \cite{disertori-pinson-spencer} do not apply.

\paragraph{Main result.}
This article is devoted to prove  the following result which extends 
\cite[Theorem 1]{disertori-pinson-spencer} to the two dimensional case.

Remember that in the case of GUE the spectrum of $H,$ in the thermodynamic limit,
is concentrated on $\{|E|\leq 2\}$  (see \eqref{eq:intro:semicirclelaw}).
In the case of a band matrix, non rigorous arguments
suggest that most of the spectrum remains in the same interval. Here,
we restrict to energies $E$ in the bulk and avoid $E=0$ for technical reasons.
Precisely, for $\eta > 0$ small but fixed, we consider the interval
\begin{align}
\label{eq:model:interval}
\gls{I}= \{E: \eta < |E| \leq 1.8 \}.
\end{align}

\begin{theo}
\label{theo:model:mainresult}
For $d=2$ and each fixed $\gls{alpha}\in(0,1)$, there exists a value $W_0(\alpha)$ such that
for all $W\geq W_0(\alpha)$ and $E\in\mathcal{I}$
\begin{align}
\label{eq:model:semicircle}
|\bar\rho_\Lambda(E) - \rho_{SC}(E) | &\leq W^{-2}\eto^{K(\ln W)^\alpha},\\
\label{eq:model:derivative}
|\partial_E^n \bar\rho_\Lambda(E)|&\leq  C_n W^{-1}  (\ln W)^{n (\alpha+1) } \eto^{K  (\ln W)^{\alpha }}
 \qquad\forall n\geq 0,
\end{align}
where $\rho_{SC}$ is Wigner's semicircle law defined in \eqref{eq:intro:semicirclelaw}.
The constants $C_n$  and $K$ are independent of $\Lambda$ and $W.$
Both estimates hold uniformly in $\Lambda$ and hence also in the infinite volume limit $\Lambda\to\mathbb{Z}^2$.
\end{theo}

\begin{rem} Note that we obtain the semicircle law 
with a precision 
\begin{align*}
\bar\rho_\Lambda(E)=\rho_{SC}(E)  +
 \mathcal{O}\left(W^{-2+\delta} \right) 
\end{align*}
for small $\delta > 0$ depending on $W_0(\alpha)$,
while in $d=3$ one obtains $ \mathcal{O}\left(W^{-2} \right)$  \cite[eq. (2.7)]{disertori-pinson-spencer}.
Moreover, \eqref{eq:model:derivative} implies a $W$-independent estimate on the derivatives up to a certain order $n_0(W)$
\begin{align*}
|\partial_E^n \bar\rho_\Lambda(E)|&\leq C_{n}\quad  \forall n\leq n_{0} (W),
\end{align*}
with $\lim_{W\to\infty} n_0(W) = \infty$.
\end{rem}

\paragraph{Strategy.} The strategy is similar to the one in  \cite{disertori-pinson-spencer}. 
We establish a dual representation for the averaged DOS via the supersymmetric approach and
apply a complex translation into the saddle points.
To overcome the second well problem, we
modify the factorization procedure, using slightly larger blocks (of size $W^{2} (\ln W)^{\alpha}$ instead
of the natural $W^{2}$) for our cluster expansion.
This yields better estimates for configurations near the second saddle, but creates new problems
for the 'good' configurations, near the main saddle. 
An equilibrium between these two conflicting effects is possible 
because $d=2$ is a 'limit' case. 
Finally, as in  \cite{disertori-pinson-spencer}, we apply a non standard cluster expansion,
extracting at each step a 'multi-link' consisting of three instead of one connection.
Here, in contrast to \cite{disertori-pinson-spencer}, we use the (super-) symmetric structure of the dual representation to reformulate 
the cluster expansion in a more compact and transparent way.
We also use the symmetry to simplify the dual representation and a number of
other equalities.

\begin{rem}
Note that Wegner-type estimates on the integrated density of states can be obtained, in any dimension,
by softer methods (see  \cite{SchenkerPeledShamisSodin2016}). Here, the dual representation plus
cluster expansion give an explicit representation of the function $\frac{1}{|\Lambda|}\mathbb{E}[ \Tr (z-H)^{-1}]$
where the limit $\im z\to 0$ can be taken explicitely. This representation 
 remains valid in the thermodynamic limit  and allows to study detailed properties
(such as smoothness and main contributions) of the limit function. 
\end{rem}

\paragraph{Organization of the paper.}
In Section \ref{sec:strategy}, the dual representation  and the complex  contour
deformation are introduced. This allows to reformulate the problem 
as Theorem \ref{theo:strategy:reducedresult}.
Finally a sketch of the proof's strategy is given.
In Section \ref{sec:preliminary}, we summarize some properties and preliminary estimates
that will be needed in the proof. We also prove the main result in a finite volume.
These results build a foundation for the infinite volume case.
In Section \ref{sec:cluster}, the cluster expansion is introduced and the limit 
$\Lambda \uparrow \mathbb{Z}^2$ is analyzed.
A short introduction to the supersymmetric formalism is given in App. \ref{sec:appsusy},
and a proof of the dual representation is given in App. \ref{sec:appproofsusy}.
Finally, App. \ref{sec:appcov} collects some results on the discrete Laplace
operator in $d=2$ and some matrix inequalities, together with their proof.
A list of symbols can be found at the end.
 
\begin{notation}
Since we apply many estimates in the paper, we denote by $K$ any large positive constant 
independent of $W$ and $\Lambda$.
\end{notation}

\paragraph{Acknowledgements}
We are grateful to  Tom Spencer for encouraging us to pursue the  $d=2$ case and for numerous discussions.
We thank Sasha Sodin for  useful discussions on  representations of the  discrete Laplacian.
We are also grateful to Martin Lohmann, Susanne Hilger, Richard H\"ofer and Anna Kraut for 
helpful discussions and suggestions related to this paper. 
Very special thanks go to David Brydges for sharing his many insights 
on the model and inspiring discussions on cluster expansions.
Finally, we acknowledge the Deutsche Forschungsgemeinschaft for support through CRC 1060
``The  Mathematics  of  Emergent  Effects'' and the Hausdorff Center for Mathematics.

\section{Reformulating the problem}
\label{sec:strategy}
We perform a duality transformation and 
a complex contour deformation to rewrite the average of the Green's function as a 
Gaussian integral with some remainder (cf. Lemma \ref{lem:strategy:dualrep}).
This reduces the proof of Theorem \ref{theo:model:mainresult} to bound this functional integral 
appropriately (cf. Theorem \ref{theo:strategy:reducedresult}). 
In the end we give ideas for the proof.

\subsection{Duality transformation}
\label{subsec:dualtrafo}
The first step in the proof is to represent, via the supersymmetric formalism  
(cf. App. \ref{sec:appsusy}), 
the trace in \eqref{eq:model:averagedos}
as a functional integral where a saddle point analysis can be justified.
Recall the definition of $J$ given by \eqref{eq:model:J}. A normalized Gaussian measure with 
covariance $J$ is defined by
\begin{align}
\label{eq:strategy:normalizedGaussianmeasure}
\gls{mu}
\coloneqq \det\left[ \tfrac{J^{-1}}{2\pi}\right] \eto^{-\frac{1}{2} ((a, J^{-1} a) +(b, J^{-1} b))}
\prod_{j \in \Lambda} \diff  a_j  \diff  b_j ,
\end{align}
with $\gls{a,b}\in\mathbb{R}^\Lambda$, $\prod_{j \in \Lambda} \diff  a_j  \diff  b_j$ is the Lebesgue product measure
and 
$(a, J^{-1} a) = \sum_{i, j \in \Lambda} a_i J_{ij}^{-1} a_j$. 
With this definition we can state the following lemma which is a variant of \cite[Lemma 1]{disertori-pinson-spencer}.

\begin{lemma}
\label{lem:strategy:susy}
For any space dimension $d\geq 1$, the following identities hold:
\begin{align}
\label{eq:strategy:dualrepsusy1}
\frac{1}{|\Lambda|}\mathbb{E}[\Tr G^+_\Lambda(E_\varepsilon)] 
&= \int \diff\mu_{J} (a,b)\prod_{j \in \Lambda}\tfrac{E_\varepsilon- i b_j}{E_\varepsilon - a_j}\ 
\det \left[1-FJ  \right] \ a_{0},\\
\label{eq:strategy:dualrepsusy2}
\partial^n_E\frac{1}{|\Lambda|}\mathbb{E}[\Tr G^+_\Lambda(E_\varepsilon)] 
&= \int \diff\mu_{J} (a,b)\prod_{j \in \Lambda}\tfrac{E_\varepsilon- i b_j}{E_\varepsilon - a_j}\ 
\det \left[1-FJ  \right] \ a_{0} \left[\sum_{j\in\Lambda}a_j-ib_j\right]^n,
\end{align}
where $\diff\mu_{J} (a,b)$
is the normalized Gaussian measure defined above and $F = F(a,b)$ is a diagonal matrix with entries
\begin{align*} 
F(a,b)_{ij} 
&= \delta_{ij} \tfrac{1}{(E_\varepsilon - a_j)(E_\varepsilon - i b_j)}.
\end{align*}
\end{lemma}

\begin{proof}
The idea of the proof is to write $(E_\varepsilon-H)^{-1}_{ii}$ as a complex Gaussian integral and
represent the normalization as a Fermionic Gaussian integral. Then the average over $H$ can be computed
easily and one can integrate the Fermionic variables again. This 
is analog to the procedure in \cite{disertori-pinson-spencer}, but we apply an additional step 
of integration by parts to simplify our result. The second identity is proven similarly.
For convenience of the reader, a sketch of the proof is given in App. \ref{sec:appproofsusy}.
\end{proof}

The integrals above are well-defined only for $\varepsilon > 0$ since for each $a_j$
there is a pole at $a_j = E_\varepsilon$.
Note that there are no singularities in $b_j$.

\begin{rem}
In this dual representation, the  only contribution from the observable  is the term $a_0$. 
By the same techniques, we obtain 
\begin{align}
\label{eq:strategy:avarage1}
1 = \mathbb{E} [1] 
= \int \diff\mu_{J} (a,b)\prod_{j \in \Lambda}\tfrac{E_\varepsilon - i b_j}{E_\varepsilon - a_j}\ 
\det \left[1-JF (a,b)  \right] .
\end{align}
\end{rem}

\subsection{Contour deformation}
\label{subsec:countourdeform}
By the same saddle analysis performed in \cite[Section 4]{disertori-pinson-spencer}, we expect 
the complex normalized measure 
\begin{align}\label{eq:strategy:measure}
\diff\mu_{J} (a,b)\prod_{j \in \Lambda}\tfrac{E_\varepsilon- i b_j}{E_\varepsilon - a_j}\ 
\det \left[1-F (a,b) J  \right] 
\end{align}
to be concentrated near the constant configurations given by 
$a_{j}=a^{\pm}_{s},$ $b_{j}=b^{\pm}_{s}$ for all $j\in \Lambda$, 
where $\gls{as,bs}$ are the saddle points
$a_s^{\pm} = \mathcal{E}_r \mp i \mathcal{E}_i,$ 
$b_s^{\pm} = - i \mathcal{E}_r \mp \mathcal{E}_i,$ and 
\begin{align*}
\gls{mathcalE} = \mathcal{E}_r - i \mathcal{E}_i = \tfrac{E}{2} - i \sqrt{1 - \tfrac{E^2}{4}}
\end{align*}
has the useful properties $E - \mathcal{E}= \bar{\mathcal{E}}$
and $\mathcal{E} \bar{\mathcal{E}} = 1$ for all $|E|< 2$.

We perform a translation of the real axis in the complex plane
in order to pass through a saddle point.
For the variables $a$, we translate to the saddle $a_s^{+} = \mathcal{E}$ 
to avoid crossing the pole in $a = E_\varepsilon$.
The variables $b$ have no pole and both saddle points have the same
imaginary part. Hence a complex translation allows to pass through both saddles. 
We will prove later that $b_s^{+} = - i \mathcal{E}$ is the dominant one.
In the next lemma we show that, after the deformation, we can take the limit
$\varepsilon\to 0$, and  the translated measure can be reorganized as
\[
\diff\mu_{J} (a+\mathcal{E},b-i\mathcal{E})\prod_{j \in \Lambda}
\tfrac{E_\varepsilon- i (b_j-i\mathcal{E})}{E_\varepsilon - (a_j+\mathcal{E})}
\ \det \left[1-F (a+\mathcal{E},b-i\mathcal{E}) J  \right] =
\diff\mu_{B} (a,b) \mathcal{R} (a,b),
\]
where the Gaussian measure $\diff\mu_{B} (a,b)$ has now a complex covariance.

\begin{lemma}
\label{lem:strategy:dualrep}
By a complex deformation the functional integrals \eqref{eq:strategy:dualrepsusy1} and 
\eqref{eq:strategy:dualrepsusy2}, in the limit $\varepsilon \to 0$, can be written as
\begin{align}
\label{eq:strategy:dualrepspa1}
\lim_{\varepsilon\to 0 }\frac{1}{|\Lambda|}\mathbb{E}[\Tr G^+_\Lambda(E_\varepsilon)] 
=& a_{s}^{+}+ \int\diff  \mu_B(a,b)  \mathcal{R} (a,b) a_{0},\\
\label{eq:strategy:dualrepspa2}
\lim_{\varepsilon\to 0 }\partial_E^n\frac{1}{|\Lambda|}\mathbb{E}[\Tr G^+_\Lambda(E_\varepsilon)]
=&\int\diff  \mu_B(a,b)  \mathcal{R} (a,b) a_{0}\left(\sum_{j\in\Lambda}a_j-ib_j\right)^n\\
\label{eq:strategy:dualrepspa3}
=& \sum_{j_{1},\dots ,j_{n}} \int\diff \mu_B(a,b) \mathcal{R} (a,b) a_{0}
\prod_{k=1}^{n}\left(a_{j_{k}}-ib_{j_k}\right),
\end{align}
where $a,b\in \mathbb{R}^\Lambda$, and $\diff\mu_B(a,b)$ is the normalized Gaussian measure 
as defined in \eqref{eq:strategy:normalizedGaussianmeasure}
with  complex covariance 
\begin{align}
\label{eq:strategy:defB}
\gls{B} \coloneqq (- W^2 \Delta + (1 - \mathcal{E}^2))^{-1}.
\end{align}
The remainder $ \mathcal{R} (a,b)$ is defined by
\begin{align}
\label{eq:strategy:remainder}
\gls{R(a,b)}\coloneqq   \det[1 + DB]\   \eto^{\mathcal{V} (a,b)},
\end{align}
where $D_{ij}=D_{ij} (a,b) = \delta_{ij} D_j (a,b)$ is a diagonal matrix, and we defined
\begin{align}
\label{eq:strategy:defD}
\begin{split}
\gls{D}_j (a,b) 
& = \mathcal{E}^2 - F(a + a_{s}^{+}, b +b_{s}^{+})_{jj} 
= \mathcal{E}^2 - \frac{1}{(\bar{\mathcal{E}} - a_j) (\bar{\mathcal{E}} - i b_j)} \\
& = -\int_0^1  \left(\frac{a_j}{(\bar{\mathcal{E}} - t a_j)^2(\bar{\mathcal{E}} - i t b_j)} 
+ \frac{i b_j}{(\bar{\mathcal{E}} - t a_j) (\bar{\mathcal{E}} - i t b_j)^2}\right) \diff  t,
\end{split}
\\
\label{eq:strategy:defV}
\gls{V(a,b)}&= \sum_{j\in\Lambda} \mathcal{V}_j(a,b)
= \sum_{j\in \Lambda } V(a_{j})-V (ib_{j}), \quad 
\gls{V(x)} = \int_0^1  \frac{x^3(1 - t)^2}{(\bar{\mathcal{E}} - t x)^3} \diff  t.
\end{align}
\end{lemma}

\begin{proof}
By Cauchy's theorem, we can perform the translations
$a_j \mapsto a_j +a_s^+$ and $b_j \mapsto b_j +b_s^+$
for all $j \in \Lambda$ and take the limit $\varepsilon \to 0$
inside the functional integral \eqref{eq:strategy:dualrepsusy1}.
Note that translating to $a_{s}^{+}$ ensures that there is no
additional contribution from the pole $E+i\varepsilon$. 
Using \eqref{eq:strategy:avarage1}, the integral with constant $a_s^+$ gives $1$.
The measure \eqref{eq:strategy:measure} after the translation is reorganized as follows.
Expanding around $a=b=0$ we can write
\begin{align*}
\diff\mu_J(a+\mathcal{E},b-i\mathcal{E}) 
\prod_{j \in \Lambda}\tfrac{E_\varepsilon- i (b_j-i\mathcal{E})}{E_\varepsilon - (a_j+\mathcal{E})}=
\diff\mu_B(a,b)\tfrac{\det B}{\det J} e^{\gls{V(a,b)}}.
\end{align*}
where $\gls{V(a,b)}=O (|a|^{3}+|b^{3}|)$ since linear contributions vanish (we are expanding around the saddle)
and constant terms cancel.
Finally
\begin{align*}
\det[1-F(a+\mathcal{E},b-i\mathcal{E})J]\tfrac{\det B}{\det J}= \det[1 + D B].
\end{align*}

To obtain the second identity, note that $(a_j+a_s^+) - i(b_j+b_s^+) = a_j-ib_j$,
and the integral with constant $a_s^+$ vanishes 
since it corresponds to the derivative of a constant.
\end{proof}

\begin{rem}
Note that now there is no pole in $a_j$ if $|E|<2$ since $|\bar{\mathcal{E}}-a_j|\geq |\mathcal{E}_i|>0$
for all $a_j\in\mathbb{R}$. 
For $b_j$, a singularity appears from the determinant for the special case $E=0$.
As the same factor appears outside the determinant, this is a removable singularity. 
Nevertheless we avoid $E=0$ in the definition of the interval $\mathcal{I}$ \eqref{eq:model:interval}.
\end{rem}

With these representations, the proof of Theorem \ref{theo:model:mainresult}
is reduced  to prove the following theorem since $\im (a_{s}^{+})=\sqrt{1-E^{2}/4}$ yields the semicircle law
\eqref{eq:intro:semicirclelaw}.

\begin{theo}
\label{theo:strategy:reducedresult}
Under the same assumptions as in Theorem \ref{theo:model:mainresult}, we have
\begin{align}
\label{eq:strategy:semicircle}
\left |\int\diff  \mu_B(a,b)  \mathcal{R} (a,b) a_{0}\right |
&\leq W^{-2} K^{(\ln W)^\alpha}\\
\label{eq:strategy:derivative}
\left|\sum_{j_1,\dots,j_n\in\Lambda}\int\diff\mu_B(a,b)\mathcal{R} (a,b) a_{0}\prod_{k=1}^{n}\left(a_{j_{k}}-ib_{j_k}\right)\right|
&\leq C_n W^{-1}  (\ln W)^{n (\alpha+1) } \eto^{K  (\ln W)^{\alpha }}.
\end{align}
\end{theo}

\subsection{Strategy of the proof: finite and infinite volume}
To prove the results above, we will need to estimate integrals of the following
form
\begin{align}
\label{eq:strategy:startingpoint}
\int \diff \mu_B (a,b) \mathcal{R}(a,b) \mathcal{O}(a,b),
\end{align}
where $\gls{O(a,b)} \coloneqq \prod_{k} a_{k} \prod_{l} b_{l}$ is a {\em local} observable, i.e.
a product of finitely many field factors $a_k$ and $b_l$.
We will show that, inserting absolute values inside \eqref{eq:strategy:startingpoint} 
leads to the following estimate
\begin{align*}
\left|\int\diff\mu_B(a,b)\mathcal{R}(a,b)\mathcal{O}(a,b)\right|
\leq \eto^{K\frac{|\Lambda|}{W^2}}\int \diff\mu_C(a,b) \left|\eto^{\Tr DB}\right| \
\left|\eto^{\mathcal{V}(a,b)}\right| \ \left|\mathcal{O}(a,b)\right|,
\end{align*}
where $C$ is a real covariance (defined in \eqref{eq:preliminary:defC} below),
and $D,B$ and $\mathcal{V}$ were defined in Lemma \ref{lem:strategy:dualrep}.
Guided by the saddle point approach, we will partition the domain of integration into different 
regions, respectively near to and far from the saddle points, and estimate the integral 
in each region separately (cf. Lemma \ref{lem:preliminary:estimates}).
To obtain the finite volume estimate of \eqref{eq:strategy:semicircle},
an additional preliminary step of integration by parts is needed to improve the error estimates.
All this is done in Section \ref{sec:preliminary}.

These arguments work only in finite volume, since the factor $\exp(K|\Lambda|W^{-2})$ 
diverges as $|\Lambda| \to \infty$.
To deal with this problem, we will partition $\Lambda$ into cubes 
(of finite, but large volume).
Applying a suitable  cluster expansion, we can write \eqref{eq:strategy:startingpoint}
as a sum of the form
\begin{align*}
\sum_{Y}c_Y F_Y,
\end{align*}
where $Y$ are polymers, i.e. unions of cubes, and the constant $c_Y$ is an exponentially small
factor controlling the sum.
Finally,  $F_Y$ is a functional integral depending only on the fields inside $Y$,
and can be estimated by the same tools as in the finite volume case.
The precise definitions and details are given in Section \ref{sec:cluster}.

\section{Preliminary results}\label{sec:preliminary}
In this section, we start by collecting in Section \ref{subsec:cov} and \ref{subsec:estim}
some results and bounds we will need later. Finally in Section \ref{subsec:finitev} we prove an estimate for
the absolute value of the integral  \eqref{eq:strategy:startingpoint} in a large but finite volume. 
The proof uses a partition of the integration domain  into regions, selecting values of $(a,b)$ in the
vicinity or far from the saddles.
\subsection{Properties of the covariance}
\label{subsec:cov}
The Hessian $B^{-1} = - W^2 \Delta + (1 - \mathcal{E}^2)$ has a complex mass term
\begin{align*}
1 - \mathcal{E}^2 = 2 \left( 1 - \tfrac{E^2}{4} \right) + i E \sqrt{1 - \tfrac{E^2}{4}} 
=\vcentcolon m_r^2 + i \sigma_{E} m_i^2,\quad \sigma_{E} \coloneqq \sgn (E).
\end{align*}
For $|E|<2$, $m_r^2>0$, 
hence the integrals \eqref{eq:strategy:dualrepspa1} and \eqref{eq:strategy:dualrepspa2} 
are finite.
We introduce the real covariance $C$ defined by 
\begin{align}
\label{eq:preliminary:defC}
\gls{C}  \coloneqq [\re (B^{-1})]^{-1}= (- W^2 \Delta + m_r^2)^{-1}.
\end{align}
Note that $B^{-1} = C^{-1} + i \sigma_E m_i^2$ and $C>0$ both as a quadratic form and pointwise.
The decay of $C_{ij}$ depends on the space dimension $d$.
For $d=2$, we have
\begin{align}
\label{eq:preliminary:expdecayC}
0 < C_{ij} \leq 
\begin{cases}
\frac{K}{W^{2}} \ln\left(\frac{W}{m_r(1+|i-j|)}\right) & \text{if } |i-j|\leq \frac{W}{m_r},\\
\frac{K}{|i-j|^{1/2}W^{3/2}} \eto^{-\frac{m_r}{W}|i-j|} & \text{if } |i-j| > \frac{W}{m_r}.
\end{cases}
\end{align}
Morover $|B_{ij}|$ has the same decay as $C_{ij}$.
A proof is given in App. \ref{sec:appcov}.

\begin{rem}
For $d=3$, the decay is easier:
\begin{align}
\label{eq:preliminary:expdecayC3d}
C_{ij}\leq \tfrac{K}{W^2(1+|i-j|)} \eto^{-\frac{m_r}{W}|i-j|}\quad \forall i,j\in \Lambda.
\end{align}
Because of the $\log$-behavior for small distances, 
estimating the error terms \eqref{eq:model:semicircle}-\eqref{eq:model:derivative} 
in $d=2$ is more difficult than in $d=3$ (cf. \cite[eq. (2.6)-(2.7)]{disertori-pinson-spencer}).
\end{rem}

\subsection{Some useful estimates}
\label{subsec:estim}
We frequently use the following statement to estimate determinants.

\begin{lemma}
For any complex matrix $A$ with $\Tr A^* A < \infty$, we have 
\begin{align}
\label{eq:preliminary:det1+A}
\left|\det[1 + A]\right| \leq \left| \eto^{\Tr A} \right| \eto^{\frac{1}{2} \Tr A^* A}.
\end{align}
\end{lemma}

\begin{proof}
Consider the matrix $M = A + A^* + A^* A$, which is self-adjoint and diagonalizable 
with real eigenvalues $\lambda_i$.
Then
\begin{align*}
|\det[1 + A]|^2 = \det[1 + M] 
= \prod_i (1 + \lambda_i) \leq \eto^{\sum_i \lambda_i} 
= \eto^{\Tr M} \leq \left|\eto^{2 \Tr A}\right| \eto^{\Tr A^* A},
\end{align*}
where we apply $1+\lambda_i \leq \eto^ {\lambda_i}$ for all $\lambda_i\in\mathbb{R}$.
\end{proof}

In the finite volume estimates, we will insert quadratic terms in $a$ and $b$ into the measure 
and change the covariance from $C$ to $C_f \coloneqq (C^{-1} - fm_r^2)^{-1}$. 
We estimate the change of the normalization factor $\det[C^{-1}/C_f^{-1}] $ as follows. 

\begin{lemma}
\label{lem:preliminary:normalization}
For $d=2,W\gg 1$ and  $0<f<1/2$, there exist some constants $K > 0$ (independent of $W$ and $f$) such that
\begin{align*}
\det\left[\frac{C^{-1}}{C_f^{-1}}\right]
\leq \frac{1}{1-f} \exp\left[\frac{Kf|\Lambda|}{W^2}\ln \left(\frac{W^2}{1-f}\right)\right].
\end{align*}
\end{lemma}

\begin{proof}
We use the explicite eigenvalues of $C^{-1}$ and $C^{-1}_f$ to write 
\begin{align*}
\det\left[\frac{C^{-1}}{C_f^{-1}}\right] &= 
\prod_k\frac{2\sum_{l=1}^d(1-\cos k_l)W^2+m_r^2}{2\sum_{l=1}^d(1-\cos k_l)W^2+m_r^2(1-f)}\\
&\leq \frac{1}{1-f} 
\exp\left[{m_r^2f\sum_{k\neq 0}\left(2\sum_{l=1}^d(1-\cos k_l)W^2+m_r^2(1-f)\right)^{-1}}\right],
\end{align*}
where we extract the zero's mode and apply $1+\lambda_i \leq \eto^ {\lambda_i}$.
Approximating the sum in the exponential yields the above result.
\end{proof}

Finally, we give the Brascamp-Lieb inequality \cite{brascamp-lieb}, 
which is used in the estimates near the dominant saddle point.
\begin{theo}[Brascamp-Lieb inequality]
\label{theo:preliminary:brascamplieb}
Let $\mathcal{H}(x)$ be a positive Hamiltonian, symmetric under $x\mapsto -x$
and let $\diff\mu_\mathcal{H}(x)$ be a Gibbs measure given by
\begin{align*}
\diff  \mu_\mathcal{H}(x) 
\coloneqq \diff  x_1\cdots \diff  x_N \frac{1}{Z(\mathcal{H})}\eto^{-\frac{1}{2}\mathcal{H}(x)},
\end{align*} 
where  $Z(\mathcal{H}) \coloneqq \int \diff x_1\cdots \diff x_N 
\exp\left(- \mathcal{H}(x)/2\right)$  is the partition function.
If $\mathcal{H}''\geq C^{-1} >0$, the following inequalities hold:
\begin{align}
\label{eq:preliminary:brascamplieb}
\int \diff  \mu_\mathcal{H}(x)|x_i|^n &\leq \int \diff  \mu_C(x) |x_i|^n, 
\quad \forall n > 0 \text{ and }\\
\nonumber 
\int \diff  \mu_\mathcal{H}(x)\eto^{(v,x)}&\leq \int \diff  \mu_C(x) \eto^{(v,x)},
\end{align}
where $\diff  \mu_C(x)$ is the free Gaussian measure and
$v\in\mathbb{R}^N$ and $(v,x)=\sum_{i=1}^N v_i x_i$.
\end{theo}
\begin{rem}
A direct consequence of Brascamp-Lieb inequality is the following estimate,
which holds under the same assumptions as above:
\begin{align}
\label{eq:preliminary:brascamplieb2}
\begin{split}
\int \diff  \mu_\mathcal{H}(x)\prod_{i}|x_i|^{n_{i}} \eto^{(v,x)} & =
\int\diff\mu_{\mathcal{H}_v}(x)\prod_{i}|x_i|^{n_{i}} \int\diff\mu_{\mathcal{H}}(x) \eto^{(v,x)} \\
&\leq \prod_i\left[\int\diff\mu_{\mathcal{H}_v}(x)\prod_{i}|x_i|^n  \right]^{n_i/n}\int\diff\mu_{\mathcal{H}}(x) \eto^{(v,x)}\\
&\leq\sqrt{(2n-1)!!}\prod_{i}  C_{ii}^{n_{i}/2}
\eto^{\frac{1}{2}(v,Cv)} , 
\quad \forall n_{i} \geq  0, 
\end{split}
\end{align}  
where $n=\sum_i n_i$ and we changed in the first line the measure to $\mathcal{H}_v= \mathcal{H}-(v,\cdot)$ with $\mathcal{H}_v''=\mathcal{H}''$.
In the second line we applied a generalized H\"older estimate. In the last line, the Gaussian integrals are computed exactly after applying Brascamp Lieb and Cauchy Schwarz.
\end{rem}

\subsection{Finite volume estimates}
\label{subsec:finitev}
In the following we prove Theorem \ref{theo:strategy:reducedresult} in finite volume
by partitioning the domain of integration and estimating the functional integrals in each region
separately. 

\subsubsection{Inserting absolute values}
To control the infinite volume limit, we will need to estimate integrals of the
form \eqref{eq:strategy:startingpoint} with
$\mathcal{O} (a,b) =\mathcal{O}_{m,n} (a,b) = 
\prod_{k=1}^p|a_{j_k}|^{m_k} \prod_{l=1}^q |b_{j_l}|^{n_l}$, 
with $p,q \in \mathbb{N}$, $m_k,n_l\in\mathbb{N}$ and $j_k,j_l\in\Lambda$ for all 
$k\leq p, l\leq q$ and  $m \coloneqq \sum_{k=1}^p m_k$ and $n \coloneqq \sum_{l=1}^q n_l$.
Following  \cite{disertori-pinson-spencer}, we put the absolute values 
inside the integral \eqref{eq:strategy:startingpoint} and replace the complex covariance $B$ 
\eqref{eq:strategy:defB} by the real one $C$  \eqref{eq:preliminary:defC}.
The next two lemmas  are the analogs in $d=2$ of \cite[Lemma 3 and 4]{disertori-pinson-spencer}. 

\begin{lemma}
\label{lem:preliminary:dmuB}
The absolute value of the complex measure $\diff  \mu_B$ is bounded by 
\begin{align*}
|\diff  \mu_B(a, b)| \leq \eto^{K\frac{ |\Lambda|}{W^2}} \diff  \mu_C(a, b).
\end{align*}
\end{lemma}

\begin{proof}
The measure $\diff  \mu_B(a,b)$ can be written as
\begin{align*}
|\diff  \mu_B(a, b)| 
= \left| \tfrac{\det B^{-1}}{\det C^{-1}} \right| \diff \mu_C(a, b) 
= |\det [1 + i\sigma_E m_i^2 C]| d\mu_C(a, b).
\end{align*}
Applying \eqref{eq:preliminary:det1+A} with $A = i \sigma_E m_i^2 C$,  $\Tr A$ is
purely imaginary and,  using \eqref{eq:preliminary:expdecayC},   
\begin{align*}
\Tr A^* A = m_i^4 \Tr C^* C 
\leq  \sum_{i,j\in\Lambda}\tfrac{K}{W^3 |i-j|}\eto^{-\frac{m_r}{W}|i-j|}
+ \hspace{-0,2cm} \sum_{|i-j|\leq \frac{W}{m_{r}}}\tfrac{K}{W^{4}} 
\ln^{2}\left( \tfrac{W}{m_{r} (1+|i-j|)} \right)
\leq \sum_{i\in\Lambda}\tfrac{K}{W^2}.
\end{align*}
\end{proof}

\begin{lemma}
\label{lem:preliminary:determinant1+DB}
The determinant in the remainder \eqref{eq:strategy:remainder} can be bounded by
\begin{align*}
|\det[1 + D B]| \leq \eto^{K\frac{|\Lambda|}{W^2}} \left| \eto^{\Tr D B} \right|.
\end{align*}
\end{lemma}

\begin{proof}
Applying again (\ref{eq:preliminary:det1+A}), we need to bound
$\Tr (DB)^*(DB) = \sum_{i,j\in\Lambda} |D_j|^2 |B_{ij}|^2$.
We estimate $D$ by its supremum norm, $\sup_{j\in\Lambda}\sup_{a_j,b_j\in\mathbb{R}}|D_j(a_j, b_j)| \leq K$.
Finally we bound $\Tr B^*B$ by $K|\Lambda|W^{-2}$ as we did above for $C$.
\end{proof}

Applying the two lemmas in \eqref{eq:strategy:startingpoint}, we have
\begin{align*}
\left|\int\diff\mu_B(a,b)\mathcal{R}(a,b)\mathcal{O}_{m,n}(a,b)\right|
\leq \eto^{ K\frac{|\Lambda|}{W^2}}F^{m,n} ,
\end{align*}
where $F^{m,n} \coloneqq \int \diff\mu_C(a,b) \left|\eto^{\Tr DB}\right| \
\left|\eto^{\mathcal{V}(a,b)}\right| \ \left|\mathcal{O}_{m,n}(a,b)\right|$.

\subsubsection{Partition of the integration domain}
Guided by the saddle point picture, we partition, as in  \cite{disertori-pinson-spencer},  
the domain of integration into regions near and far from the saddle points: 
$1 = \sum_{k=1}^5 \chi[I^k]$ with
\begin{equation*}
\begin{alignedat}{2}
I^1 & \coloneqq \{ a, b : |a_j|, |b_j - b_{j'}| \leq \delta \, \forall  j, j' \in \Lambda 
&&\text{ and } |b_0| \leq 2 \delta \}, \\
I^2 & \coloneqq \{ a, b : |a_j|, |b_j - b_{j'}| \leq \delta \, \forall  j, j' \in \Lambda 
&&\text{ and } |b_0 - 2 \mathcal{E}_i| \leq 2 \delta \}, \\
I^3 & \coloneqq \{ a, b : b_j \in \mathbb{R} \, \forall j \in \Lambda 
&&\text{ and } \exists  j_0 \in \Lambda: |a_{j_0}| > \delta \}, \\
I^4 & \coloneqq \{ a, b : |a_j| \leq \delta \, \forall j \in \Lambda 
&&\text{ and } \exists  j_0, j_0' \in \Lambda :  |b_{j_0} - b_{j_0'}| > \delta\}, \\
I^5 & \coloneqq \{ a, b : |a_j|, |b_j - b_{j'}| \leq \delta \, \forall  j, j' \in \Lambda 
&&\text{ and } |b_0|,| b_0 - 2 \mathcal{E}_i| > 2 \delta \},
\end{alignedat}
\end{equation*}
for $\delta =\delta (W)> 0 $ small to be fixed later.
Hence, we can write $F^{m,n}= \sum_{s=1}^{5}F^{m,n}_{s}$, where
\begin{align*}
\gls{Fmns} \coloneqq \int \diff\mu_C(a,b) \ \chi[\gls{Is}] \ \eto^{\re \Tr DB} \
\left|\eto^{\mathcal{V}(a,b)}\right| \
|\mathcal{O}_{m,n}(a,b)|.
\end{align*}
In the ``small field'' regions $I^1$ and $I^2$,  all $a$ variables are near the saddle,  
and the $b$ variables are all near the first saddle at $0$ in $I^1$, 
or near the second one at $2 \mathcal{E}_i$ in $I^2$.
The main contribution to $F^{m,n}$ comes from the  region $I^{1}$, while $I^{2}$
is suppressed by a small factor from the determinant.
In the ``large field'' regions $I^s$, $s=3,4,5$, at least one variable is far away from
the saddle points. 
Their contribution is exponentially suppressed by the 
corresponding probabilities $\int  \diff\mu_{C}\ \chi[I^{s}]$.   

The following lemma gives the precise estimates on $F_s^{m,n}$.
Since we proceed analog to \cite[Section 5]{disertori-pinson-spencer}, only the main ideas and the crucial steps
are given in the proof.

\begin{lemma}
\label{lem:preliminary:estimates}
Let  $\delta=W^{-\nu}$, for some $0<\nu<1$ and $W\gg 1$. Then for any $|\Lambda|$  we have
\begin{align}
F_1^{m,n} &\leq  K^{m+n+1} \left(\tfrac{\ln W}{W^2}\right)^{(m+n)/2} \sqrt{(2m)!!(2n)!!}\ 
\eto^{K|\Lambda|W^{-3}(\ln W)^{3/2}},\nonumber\\
\label{eq:preliminary:F2}
F_2^{m,n} &\leq K^{ n+m+1} \eto^{-c|\Lambda|W^{-2}\ln W} ,
\end{align}
where in the second line $c>0$ is independent of $W$ and $|\Lambda|$.
Moreover, there exists $W_{0}(\nu)\gg 1$ such that for any $W\geq W_{0}(\nu)$ and 
 $W^{3\nu }\leq |\Lambda|\leq (C_{jj})^{-2}\delta^{2}\leq K W^{4-2\nu}(\ln W)^{-2}$, we have
\begin{align*}
F_s^{m,n} &\leq K^{n+m+1}W^{n} \prod_{k=1}^p \sqrt{m_k!}\prod_{l=1}^q \sqrt{n_l!}
\ \eto^{-K\delta^2W^2(\ln W)^{-1}} \qquad\text {for }s=3,4,\\
F_5^{m,n} &\leq K^{n+m+1}W^{n} \prod_{k=1}^p \sqrt{m_k!}\prod_{l=1}^q \sqrt{n_l!}
\eto^{-K\delta^2|\Lambda|},
\end{align*}
\end{lemma}

\begin{rem}
In the following, we want to fix the volume of our cube $\Lambda$ to an appropriate finite size. The natural 
choice would be $W^2$. This would ensure the global prefactor $\eto^{K|\Lambda |W^{-2}}$
from Lemma \ref{lem:preliminary:dmuB} is bounded by a constant independent of $W$.
On the other hand the contribution
of the second saddle would be suppressed only
by some $W^{-c}$ (cf. \eqref{eq:preliminary:F2}) which is not enough to compensate various $W$ factors 
arising in the cluster expansion.
Extending the volume to $W^2 (\ln W)^\alpha$ for fixed $\alpha \in (0,1)$
reinforces the decay to $\eto^{-c(\ln W)^{1+\alpha}}$ which bounds an arbitrary factor $W^n$ for
 $\alpha >0$. The price to pay is a worse estimate on the global prefactor
 $\eto^{K|\Lambda |W^{-2}}\leq \eto^{K(\ln W)^{\alpha }}.$ For $\alpha<1$ this can be compensated by the
 observable, which is of order $O (W^{-2})$ after extracting the leading contribution (cf. \eqref{eq:model:derivative}).
\end{rem}

\begin{proof}
Following \cite{disertori-pinson-spencer},
we first perform some (region dependent) estimates on the exponential terms 
$\re \Tr DB+\re \mathcal{V} (a,b)$ and insert the results in the measure.
In region $I^1$ the resulting measure is no longer Gaussian, 
hence we apply a Brascamp-Lieb inequality. 
In the other regions the measure remains Gaussian. 
The decay comes from $\re \Tr DB$  in  $I^2$
and from a small probability argument in the large field regions.
New features of $d=2$ appear in the choice of the volume of the cube $|\Lambda|$ and
in the bounds of $B$ and $C$, for example we have $|B_{jj}|\leq K W^{-2} \ln W$.
\begin{description}[leftmargin=0cm]
\item [Region $\mathbf{I^1}$]
In the first region, all variables $a_j$ and $b_j$ are small and we bound
\begin{align}
\label{eq:preliminary:region1}
\re V(x) \leq K |x|^3 \quad \text{and} \quad |D_j(a,b)| \leq K (|a_j|+ |b_j|)
\quad \text{for }|x|,|a_j|,|b_j|< \delta.
\end{align}
Then $\re \Tr D B \leq \sum_{j\in\Lambda} |D_j| |B_{jj}|
\leq \sum_{j\in\Lambda} K (|a_j|+ |b_j|)W^{-2}\ln W$ and
\begin{align*}
F_1^{m,n} &\leq \int\diff \mu_C(a,b) \chi[I^1]
\eto^{K\sum_{j\in\Lambda} (|a_j|+|b_j|) W^{-2}\ln W + |a_j| ^3 + |b_j|^3 }
\prod_{k=1}^p|a_{j_k}|^{m_k} \prod_{l=1}^q |b_{j_l}|^{n_l}.
\end{align*}
We define the Hamiltonian of the Gibbs measure by
\begin{align*}
\mathcal{H}(x) \coloneqq x^t C^{-1} x - K \sum_{j\in\Lambda} |x_j| W ^{-2} \ln W + |x_j|^3
\end{align*}
and $Z(\mathcal{H}) \coloneqq\int\prod_{j\in\Lambda} \diff x_j \exp(-\mathcal{H}(x)/2)\chi[I^1]$.
Then we can write
\begin{align*}
F_1^{m,n} &\leq  \left(\tfrac{Z (\mathcal{H})}{Z_0} \right)^{2} \int\diff \mu_{\mathcal{H}}(a,b) 
\prod_{k=1}^p|a_{j_k}|^{m_k} \prod_{l=1}^q |b_{j_l}|^{n_l},
\end{align*}
where $Z_0 = \det[C^{-1}/2\pi]^{1/2}$. 
Repeating the proof of \cite[Lemma 5]{disertori-pinson-spencer} in $d=2$, 
\begin{align}
\label{eq:preliminary:normalizationbl}
Z(\mathcal{H}) \leq
\eto^{ \sum_{j} (C_{f})_{jj}^{3/2}+ ( C_{f})_{jj}^{1/2}W^{-2} } Z_{0}\leq 
 \eto^{K|\Lambda|W^{-3}(\ln W)^{3/2}} Z_0. 
\end{align}
where $C_f^{-1} \coloneqq C^{-1} - fm_r^2 \leq \mathcal{H}''$ for $f = \mathcal{O}(\delta)$,
and we used $\delta<1$. 
When $m> 0$, 
\begin{align*}
\prod_{k=1}^p|a_{j_k}|^{m_k} \leq \frac{1}{m} \sum_{k=1}^p m_k |a_{j_k}|^m.
\end{align*}
The same holds for $b$.  
Applying the Brascamp-Lieb inequality \eqref{eq:preliminary:brascamplieb} 
and a Cauchy-Schwarz estimate, we obtain a factor 
$\sqrt{(2m)!!} (C_f)_{j_kj_k}^{m/2} \leq \sqrt{(2m)!!} (KW^{-2}\ln W)^{m/2}$ for each  
$|a_{j_k}|^m$ and an analog factor for $|b_{j_k}|^n$.

\item[Region $\mathbf{I^2}$]
As in \cite{disertori-pinson-spencer}, we can bound the factors $a_j^{m_{j}}$ and $b_j^{n_{j}}$ 
by constants and the potential by
\begin{align}
\label{eq:preliminary:region2a}
\re V(a_j) \leq \tfrac{m_r^2}{2} f_a a_j^2, \qquad 
\re V(ib_j) \leq \tfrac{m_r^2}{2} f_b b_j^2 + (1- f_b)2\mathcal{E}_i m_r^2(b_j-\mathcal{E}_i), 
\end{align}
with $f_a = f_b = \mathcal{O}(\delta)$.
Analog to \cite[Lemma 6]{disertori-pinson-spencer} the trace can be estimated as
\begin{align}
\label{eq:preliminary:region2b}
\re D_j B_{jj} \leq -2c W^{-2}\ln W,
\end{align}
where $c>0$ is  independent of $W$ and $\Lambda$.
Combining these estimates and using Lemma \ref{lem:preliminary:normalization} in the second step,
we obtain
\begin{align*}
F_2^{m,n} 
&\leq K^{ m+n} \eto^{-2c|\Lambda|W^{-2}\ln W}\int \diff\mu_C(a,b)
\eto^{\frac{m_r^2}{2}\sum_{j\in\Lambda}\left(f_a a_j^2+ f_b b_j^2\right)}
\eto^{(1- f_b)2\mathcal{E}_i m_r^2(b_j-\mathcal{E}_i) }\\
&\leq K^{m+n} \eto^{-2(c-K(f_a+f_b))|\Lambda|W^{-2}\ln W} \tfrac{K}{\sqrt{(1-f_a)(1-f_b)}}
\int \diff\mu_{C_{f_b}}(b)\eto^{(1- f_b)2\mathcal{E}_i m_r^2(b_j-\mathcal{E}_i) }\\
&\leq K^{m+n+1}\eto^{-  2 (c-K\delta)|\Lambda|W^{-2}\ln W}
\leq K^{m+n+1}\eto^{-c|\Lambda|W^{-2}\ln W},
\end{align*}
where the remaining integral in the second line is $1$ and we used 
$\delta < c/2K$ for $W$ large enough.

\item[Regions $\mathbf{I^3}$ and $\mathbf{I^4}$]
As in \cite{disertori-pinson-spencer}, for arbitrary $a_j$ and $b_j$ in $\mathbb{R}$ we can bound
\begin{align}\label{eq:preliminary:region34}
\re V(a_j)  \leq \tfrac{m_r^2}{2} f_a a_j^2, \quad 
\re V(ib_j) \leq \tfrac{m_r^2}{2} f_b b_j^2 + \mathcal{O}(1-f_b), \quad 
\re D_j B_{jj}\leq \tfrac{K}{W^{2}}\ln W,
\end{align}
where $f_{a},f_{b}\in (1/2,1)$.  
Inserting the quadratic terms into the measure, and 
using a small fraction of the remaining mass, we bound
\begin{align*}
|a_j|^{m_{j}} \leq \left(\tfrac{K}{\sqrt{\varepsilon (1-f_{a})}}\right)^{m_{j}} \sqrt{m_{j}!} 
\ \eto^{\frac{1}{2}\varepsilon m_r^2 (1-f_{a}) a_j^2}
\end{align*}
and the same for $b_j$, where $0<\varepsilon \ll 1$ is small but fixed and independent of $W$.
Using Lemma \ref{lem:preliminary:normalization},
and $\ln W^{2}(1-f)^{-1}\leq K_{m}\ln W$ for all $f\in[0,1-W^{-m}], m\in \mathbb{N}$, we obtain
\begin{align*}
F_s^{m,n}
&\leq \tfrac{ \mathcal{K}_{m,n} K^{m+n}}{(1-f_{a})^{m/2} (1-f_{b})^{n/2}}
\eto^{K|\Lambda|(W^{-2}\ln W+(1-f_b))}\int \diff\mu_C(a,b) 
\eto^{\frac{m_r^2}{2}\sum_{j\in\Lambda}(\tilde{f}_a a_j^2+ \tilde{f}_b b_j^2)} 
\chi[I^{s}]\\
&\leq  \mathcal{K}_{m,n} K^{m+n} W^{n+1}\eto^{K|\Lambda|W^{-2}\ln W}
\int\diff\mu_{C_{\tilde{f}_a}}(a)\diff\mu_{C_{\tilde{f}_b}}(b)\ \chi[I^{s}],
\end{align*}
where $\mathcal{K}_{m,n}=\prod_{k=1}^p \sqrt{m_k!}\prod_{l=1}^q \sqrt{n_l!} $, and
$\tilde{f}_{a}=f_{a}+\varepsilon (1-f_{a})$ (same for $\tilde{f}_{b}$). 
In the second line we take $f_{a}\in (1/2,3/4)$ and $f_{b}=1-W^{-2}$, 
to ensure that all error terms in the exponent are not larger than the first one, i.e. 
$|\Lambda| W^{-2}\ln W$.
Applying \cite[Lemma 8]{disertori-pinson-spencer}, we bound the remaining  integral by:
\begin{align*}
\int\diff\mu_{C_{\tilde{f}_a}}(a)\chi[I^3]&
\leq \eto^{-x\delta}\sum_{j\in\Lambda}\eto^{\frac{1}{2}x^2 (C_{\tilde{f}_a})_{jj}}
\leq |\Lambda| \eto^{-K\delta^2W^2(\ln W)^{-1}} ,\\
\int\diff\mu_{C_{\tilde{f}_b}}(b)\chi[I^4]&
\leq \eto^{-x\delta}\sum_{j,j'\in\Lambda}\eto^{\frac{1}{2}x^2 [(C_{\tilde{f}_b})_{jj}+
(C_{\tilde{f}_b})_{j'j'}-2(C_{\tilde{f}_b})_{jj'}]}
\leq |\Lambda|^2 \eto^{-K\delta^2W^2(\ln W)^{-1}},
\end{align*}
where we set  $x=K \delta W^{-2}\ln W$ and in the first line we used 
$(C_{\tilde{f}_a})_{jj}\simeq W^{-2}\ln W$.
In the second line, $(C_{\tilde{f}_b})_{jj}\simeq W^{-2}\ln W+ W^{2}|\Lambda|^{-1}$, 
since $(1-\tilde{f}_{b})=O (W^{-2})$. The additional term is canceled by the
sum $(C_{\tilde{f}_b})_{jj}+(C_{\tilde{f}_b})_{j'j'}-2(C_{\tilde{f}_b})_{jj'}$.
Now, inserting the constraints we assumed on $|\Lambda|$ and $\delta$ we obtain the result.

\item[Region $\mathbf{I^5}$]
The proof in region $I^5$ is similar to the one in the other large field region, with the 
difference that the exponential decay comes from the bound of the potential in $b$ 
that can be improved to 
\begin{align}\label{eq:preliminary:region5}
\re V(ib_j) \leq \tfrac{m_r^2}{2} f_b b_j^2 + \mathcal{O}(1-f_b) -c\delta^2
\qquad \text{with } f_b = 1 - W^{-2}.
\end{align}
By the same arguments as above, we obtain  a factor
$\exp(-c|\Lambda|\delta^2)$ from the last term
which gives the main behaviour of the integral since  $\delta >W^{-1}(\ln W)^{1/2}$
for $W$ large enough. 
\qedhere
\end{description}
\end{proof}

\subsubsection{Improved estimates}
\label{sec:integrationbyparts}
Let us now fix $|\Lambda|=W^{2} (\ln W)^{\alpha },$ with $\alpha \in (0,1)$ as discussed is the
remark below Lemma \ref{lem:preliminary:estimates}. We want to  
apply Lemma \ref{lem:preliminary:estimates} to \eqref{eq:strategy:dualrepspa1}
and \eqref{eq:strategy:dualrepspa3} to prove \eqref{eq:model:semicircle} and
\eqref{eq:model:derivative}. 
For the correction to the semicircle law in \eqref{eq:model:semicircle}  we obtain 
\begin{align*}
\left|\int \diff\mu_B(a,b)\mathcal{R}(a,b)a_0 \right|
\leq& \eto^{K(\ln W)^\alpha} \left[\tfrac{(\ln W)^{1/2}}{W}+ \eto^{-c(\ln W)^{1+\alpha}}
+\eto^{-K\delta^2W^2((\ln W)^{-1}+(\ln W)^{\alpha})}\right]\\
\leq& \eto^{K(\ln W)^\alpha} \tfrac{(\ln W)^{1/2}}{W}
\end{align*}
which is not the desired  estimate.
To estimate  the derivatives in  \eqref{eq:strategy:dualrepspa3},
we need  to extract enough $W$ factors 
to control the sum over the indices $j_k$.
If we apply Lemma \ref{lem:preliminary:estimates} naively, we obtain
\begin{align*}
|\eqref{eq:strategy:dualrepspa3}|\leq&K_{n}
\sum_{j_1,\dots,j_n}\eto^{K(\ln W)^{\alpha}} \left(\tfrac{(\ln W)^{1/2}}{W}\right)^{n+1}
=K_{n}\eto^{K(\ln W)^{\alpha}} \tfrac{(\ln W)^{1/2}}{W} \left(W (\ln W)^{\alpha+1/2}\right)^n,
\end{align*}
which grows in $W$ algebraically for $n>0$.
To improve these bounds similar to \cite{disertori-pinson-spencer},
we apply  a few preliminary  steps of integration by parts. This is
done in the next lemma.

\begin{lemma}
\label{lem:preliminary:integrationbyparts}
 For general $\Lambda\subset \mathbb{Z}^{d},$  
 the integrals \eqref{eq:strategy:dualrepspa1} and \eqref{eq:strategy:dualrepspa3}
can be written as
\begin{align*}
\lim_{\varepsilon\to 0 }\frac{1}{|\Lambda|}\mathbb{E}[\Tr G^+_\Lambda(E_\varepsilon)] -a_s^+
&= \sum_{l_0\in\Lambda} B_{0l_0} \int\diff\mu_B(a,b) \partial_{a_{l_0}} \mathcal{R}(a,b),\\
\lim_{\varepsilon\to 0 }\partial_E\frac{1}{|\Lambda|}\mathbb{E}[\Tr G^+_\Lambda(E_\varepsilon)]
&= \sum_{j_1\in\Lambda}\sum_{l_0,l_1\in\Lambda} B_{0l_0}B_{j_1l_1}
\int\diff\mu_B(a,b)\partial_{x_{l_1}} \partial_{a_{l_0}} \mathcal{R}(a,b) +\delta_{j_1,l_1}\\
\lim_{\varepsilon\to 0 }\partial_E^n\frac{1}{|\Lambda|}\mathbb{E}[\Tr G^+_\Lambda(E_\varepsilon)]
&= \sum\limits_{\substack{j_1,\dots,j_n\in\Lambda\\l_0,\dots,l_n\in\Lambda}} B_{0l_0}\prod_{m=1}^nB_{j_ml_m}
\int \diff\mu_B(a,b) \prod_{m=1}^n\partial_{x_{l_m}} \partial_{a_{l_0}} \mathcal{R}(a,b) ,
\end{align*}
with $\partial_{x_l} = \partial_{a_l}+i\partial_{b_l}$.
\end{lemma}

\begin{proof} We use integration by parts.
For the first equation we only need to apply one step  of integration by parts. 
For the derivatives, the case $n=1$ is special. We calculate
\begin{align*}
&\sum_{j\in \Lambda}\int \diff\mu_B(a,b) \mathcal{R}(a,b) a_0 (a_{j}-ib_{j})\\
=& \sum_{j\in\Lambda}\int \diff\mu_B(a,b) \sum_{l_0\in\Lambda} B_{0l_0}
\left((a_{j}-ib_{j})\partial_{a_{l_0}}\mathcal{R}(a,b)-\delta_{jl_0}\mathcal{R}(a,b)\right)\\
=&\sum_{j\in\Lambda}\sum_{l_0,l_1\in\Lambda} B_{0l_0}B_{jl_1}   \int \diff\mu_B(a,b) 
\partial_{x_{l_1}}\partial_{a_{l_0}}\mathcal{R}(a,b)-\delta_{jl_0},
\end{align*}
where we used in the last step that $\int\diff\mu_B(a,b)\mathcal{R}(a,b)=1$.

For $n\geq 2$, we apply several steps of integration by parts. Writing $x_j =  a_j -ib_j$,
\begin{align*}
&\sum_{j_1,\dots,j_n\in \Lambda}\int \diff\mu_B(a,b) \mathcal{R}(a,b) a_0 \prod_{m=1}^n x_{j_m}\\
=& \sum_{j_1,\dots,j_n\in\Lambda}\int \diff\mu_B(a,b) \sum_{l_0\in\Lambda} B_{0l_0}
\left(\prod_{m=1}^n x_{j_m}\partial_{a_{l_0}}\mathcal{R}(a,b)+\mathcal{R}(a,b)\partial_{a_{l_0}}
\prod_{m=1}^n x_{j_m}\right)\\
=&\sum_{j_1,\dots,j_n\in\Lambda}\sum_{l_0,\dots,l_n\in\Lambda} B_{0l_0}\prod_{m=1}^nB_{j_ml_m}
\int \diff\mu_B(a,b) \prod_{m=1}^n\partial_{x_{l_m}} \partial_{a_{l_0}} \mathcal{R}(a,b) 
\end{align*}
where the last term in the second line  corresponds to the derivative of a constant,
hence equals zero. In the third line  we used  $\partial_{x_i}x_j = 0$ for all $i,j$. 
\end{proof}

These representations give the stated decay in finite volume $|\Lambda| = W^2 (\ln W)^\alpha$:
\begin{lemma}
For fixed $|\Lambda| = W^2 (\ln W)^\alpha$ we have 
\begin{align*}
|\lim_{\varepsilon\to 0 }\frac{1}{|\Lambda|}\mathbb{E}[\Tr G^+_\Lambda(E_\varepsilon)] -a_s^+|
&\leq \eto^{K(\ln W)^\alpha} W^{-2}\ln W,\\
|\lim_{\varepsilon\to 0 }\partial_E^n\frac{1}{|\Lambda|}\mathbb{E}[\Tr G^+_\Lambda(E_\varepsilon)]|
\leq C_n.
\end{align*}
\end{lemma}

\begin{proof}
We apply Lemma \ref{lem:preliminary:estimates} on the representations of the previous lemma.
Deriving the remainder $\mathcal{R}(a,b)$, we obtain 
\begin{align*}
\partial_{a_l}\left(\mathcal{R}(a,b) \right)
= \left(\det[1+DB] \partial_{a_l}V(a_l)+ 
\det B\, {\det}_{\{l\},\{l\}}[B^{-1}+D]\,\partial_{a_l} D_l\right) \eto^{\mathcal{V}(a,b)}.
\end{align*}
In the first summand, we can bound $|\partial_{a_l} V(a_l)| \leq K |a_l|^2$. 
In the second one, we bound $|\partial_{a_l}D_l|\leq K$. 
In region $I^1$ the matrix $B^{-1}+D$ is invertible and 
\begin{align*}
\left|\det B\ {\det}_{\{l\},\{l\}}[B^{-1}+D]\right| = \left|\det[1+DB](B^{-1}+D)^{-1}_{ll}\right|
\leq |\det[1+DB]| \tfrac{\ln W}{W^2}.
\end{align*}
In the other regions, it suffices to write the expression above as $|\det (1+M) |$ 
(for a certain matrix $M$) and bound it similar to Lemma \ref{lem:preliminary:determinant1+DB}.
Using Lemma \ref{lem:preliminary:estimates}, the integral is bounded by
\begin{align*}
\sum_{l\in\Lambda}|B_{0l}|\eto^{K(\ln W)^\alpha}W^{-2}\ln W\leq\eto^{K(\ln W)^\alpha} W^{-2}\ln W
\end{align*}
since $\sum_{l\in\Lambda}|B_{0l}|\leq K$. 
This proves the first part.

For $n=1$ the first integral yields a factor $(W^{-1}(\ln W)^{1/2})^3$
which controls the sum over $j$ easily.
In the second term, the sum over $j$ disappears.
In both cases  the sum over $l$ is performed by  $B_{0l}$ and is bounded by a constant.

For $n\geq2$, note that each factor $B_{ij}$ controls a sum over $i\in\Lambda$ or $j\in\Lambda$.
The largest contribution appears when all $l_{m}$ are different. In this case the expression 
above is bounded by $W^{-2} (\ln W)^{1+ (2+\alpha )n}\ll 1$  for all $n\leq n_{0} (W)$.
When $l_{m}=l_{m'}$ the factor $W^{-2}\ln W$ comes from $B_{j_{m'}l_{m'}}$, since no
sum over $l_{m'}$ is needed.
\end{proof}

\begin{rem}
Note that the representation for the first derivative is special, but the additional term is
easy to handle, since it directly give control over the sums over $\Lambda$, hence we neglect it in the following.
\end{rem}

\subsubsection{Large volume}
We can easily extend the above result from one cube of volume $W^2(\ln W)^\alpha$ to a finite 
union of cubes of this volume. 
The procedure is independent of the space dimension and follows \cite[Corollary 1]{disertori-pinson-spencer}.
The idea is to decouple the cubes by replacing the periodic Laplacian in the covariance by one 
with Neumann boundary conditions.
\begin{theo}
\label{theo:preliminary:largevolume}
For $d=2$
and each fixed $\alpha\in(0,1)$, there exists a value $W_0(\alpha)$ such that
for $W\geq W_0(\alpha)$ and $\Lambda$ a union of $N$ cubes of volume $W^2(\ln W)^\alpha$
we have for all $E\in\mathcal{I}$
\begin{align*}
|\bar\rho_\Lambda(E) - \rho_{SC}(E) | &\leq W^{-2} \eto^{NK(\ln W)^\alpha}\\
|\partial_E^n \bar\rho_\Lambda(E)|&\leq C_{n}N^{n}\eto^{NK(\ln W)^{\alpha}} ,
\end{align*}
where $\rho_{SC}$ is Wigner's semicircle law \eqref{eq:intro:semicirclelaw}
and $C_{n}$ depends only on $n$.
\end{theo}

\begin{rem}
Note that here $n$ can take any value independent of $W$.
\end{rem}

\begin{proof}
We consider again the dual representations \eqref{eq:strategy:dualrepspa1} and 
\eqref{eq:strategy:dualrepspa2},  apply the steps of integration by parts described above and 
pull the sums in front of the integral. 
The measure is again bounded by Lemma \ref{lem:preliminary:dmuB}, and  the determinant $\det 1+DB$
by Lemma \ref{lem:preliminary:determinant1+DB}.
When  derivatives fall on the determinant, this is replaced by terms of the form 
 $\det B \det_{\mathcal{J}\mathcal{J}}(B^{-1}+ D)$ for an index set $\mathcal{J}$, which
can be bounded in the same way.
Collecting a factor $\exp(K|\Lambda|W^{-2}) = \exp(NK(\ln W)^\alpha)$, we need to estimate an 
integral of the form \eqref{eq:strategy:startingpoint}.
Note that all terms in the integral factorize over the cubes except the measure $\diff\mu_C$.
Before applying Lemma \ref{lem:preliminary:estimates}, we insert the partition of the
integration domain
in each cube separately:  $1=\prod_{\triangle}\sum_{s_{\triangle}=1}^{5}\chi[I^{s_{\triangle}}_{\triangle}].$
We estimate the terms $\re \Tr DB$ and $\re\mathcal{V}$ in each cube depending on the region 
(as in \eqref{eq:preliminary:region1}, \eqref{eq:preliminary:region2a},
\eqref{eq:preliminary:region2b}, \eqref{eq:preliminary:region34}, \eqref{eq:preliminary:region5}).
Inserting the quadratic contributions into the measure and extracting the normalizing factor, 
we obtain
\begin{align*}
\sum_{\{s_{\triangle}\}_{\triangle\in\Lambda}} \sqrt{ \det  \tfrac{C^{-1}}{C_{f_{a}}^{-1}} 
\det \tfrac{C^{-1}}{C_{f_{b}}^{-1}}}
\int \diff\mu_{C_{f_{a}}} (a) \diff\mu_{C_{f_{b}}} (b) |\mathcal{O}(a,b)|
\eto^{\sum_{\triangle}h_{s_{\triangle}} (a,b)}\prod_{\triangle}\chi[I^{s_{\triangle}}_{\triangle}],
\end{align*}
where we collect all non-quadratic (cubic, linear and constant) terms  in $h_{s_{\triangle}} (a,b)$, and
$C_{f}^{-1}=C^{-1}-\hat{f}m_r^2$, and 
$\hat{f}=\sum_{\triangle,s_{\triangle}>1}f_{\triangle}\mathbf{1}_{\triangle}$ is a block diagonal matrix.
Note that, for the moment, only the mass in regions with $s_{\triangle}>1$
has been modified. Now we can bound the normalization factor in each cube as usual since
\begin{align}\label{eq:preliminary:prooflargevolume1}
\tfrac{\det C^{-1}}{\det C_f^{-1}}=\det[1+\hat{f}m_r^2 C_f] \leq \det[ 1+ \hat{f}m_r^2 C_f^N ] 
= \prod_{\triangle: s_{\triangle}>1} \det[1_\triangle+f_\triangle m_r^2 C_{f_\triangle}^{\triangle,N}]
\end{align}
where $C_f^N= (- W^2\Delta_{N} (1-\hat{f})m_r^2)^{-1},$  
$-\Delta_N$ is the Laplacian  
with Neumann boundary conditions on the cube boundaries, and $C_{f_\triangle}^{\triangle,N}$
is this covariance restricted to $\triangle$. To prove the inequality above, we use 
 $C_f \leq C^N_f$ and  $\hat{f}\geq 0$ as quadratic forms, and the minmax-principle to compare the 
corresponding eigenvalues.

As in \cite[Lemma 8]{disertori-pinson-spencer}, we estimate the characteristic function  $\chi[I^3_{\triangle}]$ 
by
$K\exp(\pm\sum_{j\in \triangle}xa_j).$
A similar bound holds for $\chi [I^4_{\triangle}].$ 
Now, apart from the cubic contributions in the first region, all terms depending on $a$ or $b$
in the integral  are of the form $|a|^n$ or $\exp((a,v)),$ for some vector $v.$ The 
same holds for $b.$ We are then reduced to estimate an integral of the following form 
\begin{align*}
\int   \diff\mu_{C_{f_{a}}} (a) \diff\mu_{C_{f_{b}}} (b)\prod_{\triangle: s_{\triangle}=1} 
\eto^{\mathcal{F}_{\triangle} (a)+\mathcal{F}_{\triangle} (b)}\prod_{j\in \triangle} |a_{j}|^{m_{j}} |b_{j}|^{n_{j}}
 \prod_{\triangle: s_{\triangle}>1} 
\eto^{(a,v_{\triangle})+ (b,w_{\triangle})},
\end{align*}
where $\mathcal{F}_{\triangle} (a)=K\sum_{j\in \triangle}\ |a_{j}|^{3}+|a_{j}|W^{-2},$ $K>0$ is some
constant, $n_{j},m_{j}\geq 0,$ $v_{ \triangle}, w_{ \triangle}$ are some vectors. 
Defining $\mathcal{H}(a)=(a,C_{f_{a}}^{-1}a)/2-\sum_{\triangle: s_{\triangle}=1}  \mathcal{F}_{\triangle} (a)$ 
(same for $b$), we can apply Brascamp-Lieb \eqref{eq:preliminary:brascamplieb2} 
and \eqref{eq:preliminary:normalizationbl}. As a result the integral above is bounded by
\begin{align}\label{eq:preliminary:prooflargevolume2}
K^{N} \eto^{\frac{1}{2} (v,C_{f_{a}}v)+\frac{1}{2} (w,C_{f_{b}}w)  }  
\prod_{\triangle: s_{\triangle}=1} \prod_{j\in \triangle} {\scriptstyle \sqrt{(2n_{j}-1)!!}
\sqrt{(2m_{j}-1)!!}}  \ (C_{f_{a}})_{jj}^{m_{j}/2} 
(C_{f_{b}})_{jj}^{n_{j}/2} 
\end{align}
where now $\hat{f}=\sum_{\triangle}f_{\triangle}\mathbf{1}_{\triangle}$, and $f_{\triangle}>0$
for all cubes. Note that, to avoid heavy notations, we write $C_{f}$ in this case, too.
Now we replace $C_f$ by $C_f^N$ in the exponent  and hence obtain factorized estimates
over each cube. Since $C_f^N$ decays in the same way as $C_f,$ the bounds now work
as before. 

Finally, when estimating $n$ derivatives in $E$, we collect a factor $N^{n}$  from the sums over the $j_k$'s.
\end{proof}

This result is not sufficient to deal with the case of very large (or infinite) volume.
To handle this case, we introduce in the next section a cluster expansion.

\section{Proof of  Theorem \ref{theo:strategy:reducedresult}}
\label{sec:cluster}

Following  \cite{disertori-pinson-spencer} we will apply a cluster expansion
which is a variation of the rooted Brydges-Kennedy Taylor forest formula
(cf. \cite{brydgesleshouche},\cite{AR1994}) to decouple
an appropriate finite region containing the observable from the remaining volume.
In contrast to \cite{disertori-pinson-spencer}, 
we perform first the steps of integration by parts described in Section \ref{sec:integrationbyparts}.
This preliminary procedure simplifies the extraction of the correct decay later.
The cluster expansion and the preliminary steps of integrations by parts
are more easily implemented by going back to the original representation of the
integrals \eqref{eq:strategy:dualrepsusy1},\eqref{eq:strategy:dualrepsusy2} in terms of Bosonic and Fermionic variables,
as in App. \ref{sec:appproofsusy}.  This is done in  Section \ref{subsec:cluster1} below. 
In Section \ref{subsec:cluster} a cluster expansion is applied to the supersymmetric representation
obtained in Section \ref{subsec:cluster1}. 
The following Sections \ref{subsec:decayGB}-\ref{subsec:summing} bound the different terms in the cluster expansion.
More precisely, in Section \ref{subsec:decayGB} we give an estimate on the propagators and in  Section \ref{subsec:funtint} we
bound the functional integral on a finite set of cubes.
Section \ref{subsec:summing} is devoted to combining all bounds and performing the sum over vertex and cube positions
as well as the tree structure.
Finally in Section \ref{subsec:cluster:derivatives}, we sketch the procedure for the derivatives.

\subsection{Supersymmetric representation}
\label{subsec:cluster1}
To  modify the dual representation introduced in Lemma \ref{lem:strategy:susy}
we  introduce a family $\gls{rhoj}$ of Grassmann variables (cf. App. \ref{sec:appsusy}).
For each $j\in \Lambda $ we denote by $M_{j}$ the supermatrix
$\gls{M}_j = \begin{pmatrix}a_j & \bar\rho_j\\\rho_j & ib_j\end{pmatrix}.$ Note that
the trace is replaced by  $\Str M_j = a_j - ib_j$ (cf. \eqref{eq:appsusy:strace}).
With these notations we can state the new representaion.
\begin{lemma}
\label{lemma:cluster1:susyrep}
The integrals in \eqref{eq:strategy:dualrepspa1} and \eqref{eq:strategy:dualrepspa3} can be reorganized to yield
\begin{align}
\label{eq:cluster:startingpoint1}
\lim_{\varepsilon\to 0 }\frac{1}{|\Lambda|}\mathbb{E}[\Tr G^+_\Lambda(E_\varepsilon)] -a_s^+
&= \int \diff \mu_B(M) \eto^{\mathcal{V} (M)} a_0,\\
\label{eq:cluster:startingpoint2}
\lim_{\varepsilon\to 0 }\partial_E^{n}\frac{1}{|\Lambda|}\mathbb{E}[\Tr G^+_\Lambda(E_\varepsilon)]
&= \int \diff \mu_B(M) \eto^{\mathcal{V} (M)} a_0 \prod_{k=0}^n\Str M_{j_k},
\end{align}
where the supersymmetric gaussian measure  is defined by 
\begin{align}\label{eq:cluster:susygausmeasure}
\gls{muM}\coloneqq\diff M \eto^{-\frac{1}{2}\Str (M,B^{-1}M)}=
\diff \mu_B(a,b)\diff\mu_B(\bar\rho,\rho)
\end{align}
with product measure $\diff M = \prod_{j\in\Lambda}\diff M_j = \prod_{j\in\Lambda} \diff a_j 
\diff b_j \diff\bar\rho_j \diff\rho_j$, 
\begin{align*}
\diff\mu_B (\bar\rho, \rho) \coloneqq \prod_{j \in \Lambda} \diff  \bar\rho_j  \diff  \rho_j
\det\left[ \tfrac{2\pi}{B^{-1}}\right] \eto^{-(\bar\rho, B^{-1} \rho)},
\end{align*}
and $\Str (M,B^{-1}M) = \sum_{i,j\in\Lambda} B^{-1}_{ij} \Str(M_i M_j)$,
Finally, all non Gaussian terms in the integral are collected in the exponent
$\mathcal{V}(M) = \sum_{j\in\Lambda} \mathcal{V}(M_j),$ defined by
\begin{align}
\label{eq:cluster:defV(M)}
\begin{split}
 \mathcal{V}(M_j)&\coloneqq
-\ln \Sdet [\bar{\mathcal{E}}-M_j] - \mathcal{E}\Str M_j - \tfrac{\mathcal{E}^2}{2}\Str M_j^2\\
&= \int_0^1 (1-t)^2 \Str \frac{M^3}{(\bar{\mathcal{E}}-tM)^3}\diff t 
=\mathcal{V}_j(a,b)+\bar\rho_j\rho_j D_j,
\end{split}
\end{align}
where $\mathcal{V}_j(a,b)$ is the potential introduced in \eqref{eq:strategy:defV}.
Here we abuse notation by using the same letter for the potential $ \mathcal{V}(M_j)$ and $\mathcal{V}(a,b),$
since the two expressions are closely related.
\end{lemma}

\begin{rem}
This representation simplifies the cluster expansion since the covariance appears only 
in the Gaussian measure.
Note that the normalization constants for the real and Fermionic variables above cancel each other.
\end{rem}
\begin{proof}
We replace the determinant in $ \mathcal{R}(a,b)$ by  a Fermionic integral using \eqref{eq:appsusy:gaus}
and collect all remaining terms into the exponent $\gls{V(M)}.$
\end{proof}
The following result is the analog of  Lemma \ref{lem:preliminary:integrationbyparts} in this
new formalism.
\begin{lemma}
The expressions \eqref{eq:cluster:startingpoint1}  and \eqref{eq:cluster:startingpoint2}
can be reorganized as follows
\begin{align}
\label{eq:cluster:afterip1}
\lim_{\varepsilon\to 0 }\frac{1}{|\Lambda|}\mathbb{E}[\Tr G^+_\Lambda(E_\varepsilon)] -a_s^+&=
\sum_{l_0\in\Lambda} B_{0l_0}\int\diff\mu_B(M) \partial_{a_{l_0}} \eto^{\mathcal{V}(M)}
\eqqcolon \sum_{l_0\in\Lambda}B_{0l_0} F_\Lambda^{(l_0)},\\
\lim_{\varepsilon\to 0 }\partial_E^{n}\frac{1}{|\Lambda|}\mathbb{E}[\Tr G^+_\Lambda(E_\varepsilon)]
&=
\sum_{l_0,\dots,l_n} B_{0l_0}\prod_{m=1}^n B_{j_m l_m}\int\diff\mu_B(M)\prod_{m=1}^n 
\Str\partial_{M_{l_m}} \partial_{a_{l_0}} \eto^{\mathcal{V}(M)}\nonumber\\
&\eqqcolon \sum_{l_0,\dots,l_n}B_{0l_0}\prod_{m=1}^n B_{j_m l_m}F_\Lambda^{(l_0,\dots,l_n)},
\label{eq:cluster:afterip2}
\end{align}
where $\partial_{M_j}$ is defined in \eqref{eq:appproofsusy:partialM}. 
\end{lemma}  
\begin{proof}
We apply integration by parts as in Lemma \ref{lem:preliminary:integrationbyparts}.
 and use the following relations:
$(\Str \partial_{M_i}) \Str M_j^n = n \delta_{ij} \Str M_j^{n-1} $
and
$[\Str \partial_{M_i},\Str \partial_{M_j}]= 0.$
\end{proof}
Note that  $\partial_{a_{l_0}}$ moves the local observable $a_{0}$ in $0$ to the local observable
$\partial_{a_{l_0}}\mathcal{V}(M_{l_0})$  at position  $l_0.$  
Moreover, the $B$-factors enable summation over $j_1,\dots,j_n$ and $l_0$.

\begin{rem}
In the remaining we will show that applying a cluster expansion to 
$F_\Lambda^{(l_0)}$ and $F_\Lambda^{(l_0,\dots,l_n)}$ defined in
\eqref{eq:cluster:afterip1} and \eqref{eq:cluster:afterip2} 
yields the stated estimate for the 
semicircle law \eqref{eq:strategy:semicircle}, but not for the derivatives  \eqref{eq:strategy:derivative},
since in this last case one may not be able to extract enough fine structure to sum  over the indices $l_1,\dots,l_n$.
This happens when two or more of the $l_k$ coincide and we obtain linear terms in
$M$ from the derivatives $\prod_{m=1}^n \Str\partial_{M_{l_m}} \partial_{a_{l_0}} \exp(\mathcal{V}(M))$.
In this case, we need to apply again integration by parts on the resulting field factors
before performing the cluster expansion.
Nevertheless, for clarity,  we first prove the cluster expansion only for \eqref{eq:cluster:afterip2}.
It is easy to see that the same approach works for the (more involved) expression we obtain after some steps 
of integration by parts. A detailed description of the procedure can be found
in Section \ref{subsec:cluster:derivatives} below.
\end{rem}

The following lemma will simplify the cluster expansion, since the integrals over 
regions without observable contributions turn out to be trivial.

\begin{lemma}
\label{lem:cluster:complementY}
If we restrict the functional integrals $F_\Lambda^{(l_0)}$ and $F_\Lambda^{(l_0,\dots,l_n)}$ defined in
\eqref{eq:cluster:afterip1} and \eqref{eq:cluster:afterip2}  to a set $Y^C= \Lambda\backslash Y$ not containing $l_0$, 
we have for  $m>0$ and  indices $l_j\in Y^C$ for $j=1,\dots,m$ that
\begin{align*}
F_{Y^C}&=\int\diff\mu_{B_{Y^C}}(M) \ \eto^{\sum_{j\in Y^C}\mathcal{V}(M_j)}=1,\\
F^{(l_1,\dots,l_m)}_{Y^C}&=\int\diff\mu_{B_{Y^C}}(M)\prod_{j=1}^m
\Str\partial_{ M_{l_j}}\eto^{\sum_{j\in Y^C}\mathcal{V}(M_j)}=0,
\end{align*}
where $B_{Y^C}$ is the covariance restricted to the volume $Y^C$.
\end{lemma}

\begin{proof}
Using the definition of $\mathcal{V}(M)$, we can write
\begin{align*}
F_{Y^C}
 =& \int \diff M  
 \eto^{-\frac{1}{2} \Str( M, \tilde {J}^{-1}M)} 
 \eto^{(\mathcal{E}, \Str M)} 
 \int \diff \bar\Phi  \diff \Phi  \eto^{i( \bar\Phi, (\bar{\mathcal{E}} - M) \Phi)},\\
 F^{(l_1,\dots,l_m)}_{Y^C}
 =& \int \diff M  
 \eto^{-\frac{1}{2} \Str( M, \tilde {J}^{-1}M)} 
 \eto^{(\mathcal{E}, \Str M)} \\
&\times \left(\sum_{P_1,P_2} \prod_{j_1\in P_1}(-\mathcal{E}^2\Str M_{j_1}) \prod_{j_2\in P_2}\Str \partial_{M_{j_2}}\right)
 \int \diff \bar\Phi  \diff \Phi  \eto^{i( \bar\Phi, (\bar{\mathcal{E}} - M) \Phi)},
 \end{align*}
where $\tilde {J}^{-1} = B_{Y^C}^{-1}+ \mathcal{E}^2$ and we insert a superintegral  with measure 
$\diff \bar\Phi \diff \Phi = \prod_{j\in\Lambda}\diff \bar z_j \diff z_j \diff\bar\chi_j\diff\chi_j$
for the superdeterminant. In the second line we sum over all partitions $P_1 \cup P_2 = \{1,\dots, m\}$ with
$P_1\cap P_2 = \emptyset$. Note that we can rewrite both $\Str \partial_{M_j}$ and
$\Str M_j$ using integration by parts (in $\Phi$ for the first, in $M$ and then in  $\Phi$ for the 
second) as 
$\partial_{\bar{\mathcal{E}}}$ and $\sum_k \mathcal{E}^2\tilde{J}^{-1}_{jk}\partial_{\bar{\mathcal{E}}}$.
For a general $\Lambda'\subset\Lambda$, the restriction of $B^{-1}$ to $\Lambda'$ 
does not have the form $-W^2\Delta + (1-\mathcal{E}^2)$, but $\re\tilde {J}^{-1}\geq 1$ still 
holds (cf. Lemma \ref{lem:appcov:mass}). Hence we can interchange the measures and perform
integration over $M$ by completing the square. Inserting \eqref{eq:appproofsusy:averageH}, we obtain
\begin{align*}
F_{Y^C}
& =\mathbb{E}_{\tilde J}\left[ \int \diff \bar\Phi  \diff \Phi  \eto^{i (\bar\Phi, (\mathcal{E} +A-H) \Phi)}\right]
= \mathbb{E}_{\tilde J}[1] = 1,\\
F^{(l_1,\dots,l_m)}_{Y^C}
&=  \left(\sum_{P_1,P_2} \prod_{j_1\in P_1}\sum_{k_{j_1}} \mathcal{E}^2\tilde{J}^{-1}_{j_1 k_{j_{1}}}
\right)\partial_{\bar{\mathcal{E}}}^{m}
\mathbb{E}_{\tilde J}[1] = 0
\end{align*}
where $A$ is a diagonal matrix with $A_j = \mathcal{E}\sum_{k} \tilde{J}_{jk}$.
Note that $\im (\mathcal{E} +A-H)>0,$ since $\re \mathcal{E}$ and $ \im \tilde{J}$ have
the same sign. Hence the integrals above are well-defined.
\end{proof}

\subsection{Cluster expansion}
\label{subsec:cluster}
In the following, we prove a cluster expansion for the integrals 
$F_\Lambda^{(l_0)}$ and $F_\Lambda^{(l_0,\dots,l_n)}$ defined in \eqref{eq:cluster:afterip1} and \eqref{eq:cluster:afterip2}.
We partition a large but finite volume $\Lambda$ into disjoint cubes $\triangle$ of fixed volume 
$W^2 (\ln W)^\alpha$.
By interpolating the covariance, the functional integral over $\Lambda$ can be rewritten as a sum
of local integrals over unions of these cubes called polymers.
Here, we use a non-standard cluster expansion interpolating in the real covariance $C$ instead 
of $B$ and setting $B(s)=(C(s)^{-1}+i\sigma_Em_i^2)^{-1}$ (cf. also the remark below). 
This is done in an inductive procedure.
Because of the interpolation in $C$, we extract a ``multi-link'' consisting of three edges instead of a single edge in each step.

Before stating the result, we give a few notations:
The volume $\Lambda$ is divided into cubes $\gls{triangle}$ of size $W^2 (\ln W)^\alpha$. Denote by $\gls{triangle0}$ the root cube containing $l_0$.
In each step we extract a \emph{generalized cube} $\gls{tildetriangle}=(\triangle,\triangle',\triangle'')$ connected via a multi-link $( \gls{ijkk'} )$,
where $k'\in\triangle,i\in\triangle', j\in\triangle''$ and $k$ is in the volume already extracted.
The links $(i,k)$ and $(k',j)$ are ``weak'' while  $(k,k')$ is ``strong'' in the sense that it prescribes the tree structure. The collection of $\triangle_0$ and the 
extracted generalized cube is called the \emph{generalized polymer} $\gls{Ytilde} =(\triangle_0,\tilde\triangle_1,\dots,\tilde\triangle_r)$.

\begin{figure}[h]
\centering
\begin{tikzpicture}
\foreach \x/\y/\labl/\p/\q in {
0/0/$\triangle_0 = \triangle_1'$/0/-0.8,
2/0/$\triangle_1=\triangle_1''$/1/-0.8,
5/0/$\triangle_0$/-0.3/0,
8/0/$\triangle_1=\triangle_1''$/1.3/-0.8,
7/-1/$\triangle_1'$/1.3/-0.5,
11/0/$\triangle_0$/-0.3/0,
13/0/$\triangle_1$/-0.3/0,
11/-2/$\triangle_1'$/-0.3/0,
13/-3/$\triangle_1''$/-0.3/0,
2/-2/$\triangle_0$/-0.3/0,
4/-3/$\triangle_1'=\triangle_1''$/-0.8/0,
6/-3/$\triangle_1$/1.3/0
}
\draw (\x,\y)  node [label={[shift={(\p,\q)}]\small{\labl}}]{} rectangle ++ (1,1);

\foreach \a/\b/\c/\d/\e/\f/\g/\h/\l/\m/\n/\o/\p in {
0.25/0.25/0.75/0.75/2.25/0.25/2.75/0.75/right/left/right/left/0,
7.25/-0.75/5.75/0.25/8.25/0.75/8.75/0.25/right/left/right/left/0,
11.75/-1.75/11.25/0.25/13.75/0.25/13.25/-2.25/left/above/above/below/0,
4.25/-2.25/2.75/-1.25/6.25/-2.25/4.75/-2.25/below/left/right/below/0
}
{
\draw [-](\a,\b) -- (\c,\d) -- (\e,\f) -- (\g,\h);
\draw[fill=white] (\a,\b) node [\l]{\tiny{$i_1$}}circle [radius=0.05];
\draw[fill=black] (\c,\d) node [\m]{\tiny{$k_1$}}circle [radius=0.05];
\draw[fill=black] (\e,\f) node [\n]{\tiny{$k_1'$}}circle [radius=0.05];
\draw[fill=white] (\g,\h) node [\o]{\tiny{$j_1$}}circle [radius=0.05];
}
\end{tikzpicture}
\caption{Some examples for the first generalized cube
$\tilde{ \triangle}_1 = (\triangle_1,\triangle_1',\triangle_1'')$
extracted by the first link $l_1=(i_1,k_1,k_1',j_1)$ with
$k_1\in\triangle_0$, $k_1'\in\triangle_1$, $i_1\in\triangle_1'$ and $j_1\in\triangle_1''$.
The cubes may coincide with the unique constraint $\triangle_0\neq\triangle_1$.}
\label{fig:generalizedcubes}
\end{figure}
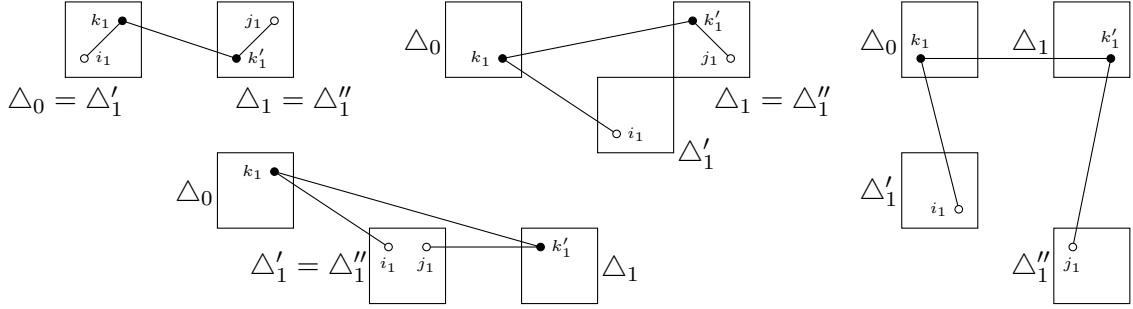

\begin{lemma} \label{lem:cluster:expansion}
For $(\ast)$ equals to the set of fixed indices $(l_0)$ and $(l_0,\dots,l_n)$, respectively, 
we can write $F_\Lambda^{(l_0)}$ and $F_\Lambda^{(l_0,\dots,l_n)}$ defined in
\eqref{eq:cluster:afterip1} and \eqref{eq:cluster:afterip2}  as
\begin{align*}
\begin{split}
F_\Lambda^{(\ast)} =& \sum\limits_{\substack{\tilde Y: \\(\ast)\in \tilde Y}}
\sum\limits_{\substack{\gls{T} \text{ on }\tilde Y,\\ |T|=r}} 
\int_{[0,1]^{r}} \diff s_r\cdots \diff s_1 M_T(s)
\\  &\times
\sum\limits_{\substack{(i_q,j_q)\\(k_q,k_q')}}
\prod_{q=1}^r G_q(s)_{i_qk_q}C_{k_q k_q'}G_q(s)_{k_q'j_q}
\ F_{{T}}^{(\ast)}[s](\{i_q,j_q\}_{q=1}^r) ,
\end{split}
\end{align*}
where $\tilde Y =(\triangle_0,\tilde\triangle_1,\dots,\tilde\triangle_r)$ is a generalized polymer
consisting of the root cube $\triangle_0$ (containing $l_0$)
and $r$ ordered generalized cubes $\tilde \triangle_q$, $q=1,\dots,r$.
Each $\tilde\triangle_q = (\triangle_q,\triangle_q',\triangle_q'')$ is a collection of three cubes 
not necessarily disjoint with the unique constraint 
$\triangle_q \cap (\triangle_0 \cup \bigcup_{p=1}^{q-1}\tilde\triangle_p)=\emptyset$.
For $(\ast)=(l_0,\dots,l_n)$, the generalized polymer $\tilde Y$ needs to contain all
the indices $l_0,\dots,l_n$.
We sum over all ordered trees $T$ on the generalized polymer, such that the $q$-th tree link connects
$\tilde\triangle_q$ with $\triangle_0 \cup \bigcup_{p=1}^{q-1}\tilde\triangle_p$.

Each tree link  consists of three lines $(i_q,k_q)$, $(k_q,k_q')$ and $(k_q',j_q)$, where
the $k_q-k_q'$ connection forms the tree structure. Precisely,
$k_q\in\triangle_0 \cup \bigcup_{p=1}^{q-1}\tilde\triangle_p$ is in the generalized polymer
 up to index $q-1$,
$k_q'\in\triangle_q$, $i_q\in\triangle_q'$ and $j_q\in\triangle_q''$.
Note that the position of $\triangle'$ and $\triangle''$ is arbitrary and they can coincide with 
each other or an already extracted cube.
For each link, we have an interpolation parameter $0\leq s_q\leq 1$.
The functional integrals $F_{{T}}^{(l_0)}$ and $ F_{T}^{(l_0,\dots,l_n)}$ are defined by
\begin{align}
\label{eq:cluster:functionalintegral1}
F_{{T}}^{(l_0)}[s](\{i_q,j_q\})
\coloneqq & \int \diff \mu_{B(s)}(M) \prod_{q=1}^r\Str(\partial_{M_{i_q}}\partial_{ M_{j_q}}) 
\left[\partial_{a_{l_0}}\eto^{\mathcal{V}(M)}\right],\\
\nonumber
F_{{T}}^{(l_0,\dots,l_n)}[s](\{i_q,j_q\})
\coloneqq & \int \diff \mu_{B(s)}(M) \prod_{q=1}^r\Str(\partial_{M_{i_q}}\partial_{ M_{j_q}})
\left[\prod_{m=1}^n\Str \partial_{M_{l_m}}\partial_{a_{l_0}}\eto^{\mathcal{V}(M)}\right],
\end{align}
where $ \gls{B(s)} \coloneqq (C(s)^{-1}+i\sigma_Em_i^2)^{-1}$ and
$\gls{C(s)}_{ij} \coloneqq s_{ij}C_{ij}$ with 
\begin{align*}
s_{ij} \coloneqq 
\begin{cases}
1 & \text{if } \exists \ q: i,j\in\tilde\triangle_q,\\
\prod_{p=q'}^{q-1} \gls{sp}& \text{if } \exists \ q'<q: i\in\tilde{\triangle}_q \text{ and }
j\in\tilde{\triangle}_{q'} \text{ or vice versa},\\
0 & \text{otherwise},
\end{cases}
\end{align*}
$M_T(s)$ is a product of $s$ factors extracted by the derivative $\partial_{s_q} B(s)$.
The propagator $\gls{Gq}^{-1} \coloneqq (1+i\sigma_Em_i^2 C(s))_{|s_p=1 \forall p > q}$
depends only on the first $q$ interpolation parameters $s_1,\dots,s_q$.
\end{lemma}

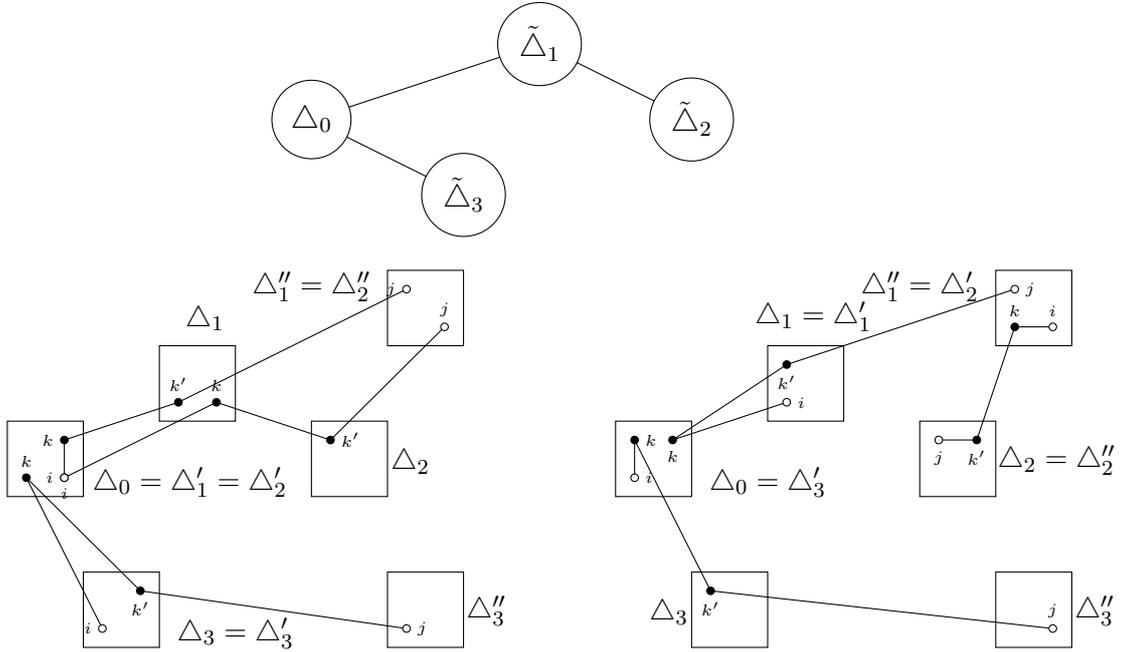
\begin{figure}[h]
\centering
\begin{tikzpicture}
\draw[-] (9,5) -- (7,6) -- (4,5) -- (6,4);
\foreach \x/\y/\s in {
4/5/$\triangle_0$,
7/6/$\tilde\triangle_1$,
9/5/$\tilde\triangle_2$,
6/4/$\tilde\triangle_3$
}
\node[draw, circle, fill=white] at (\x,\y) {\s};

\foreach \x/\y/\labl/\p/\q in {
0/0/$\triangle_0=\triangle_1'=\triangle_2'$/2.4/-0.3,
2/1/$\triangle_1$/0.6/0.9,
5/2/$\triangle_1''=\triangle_2''$/-1/0.3,
4/0/$\triangle_2$/1.3/0,
1/-2/$\triangle_3=\triangle_3'$/2/-0.3,
5/-2/$\triangle_3''$/1.3/0,
8/0/$\triangle_0=\triangle_3'$/2/-0.3,
10/1/$\triangle_1=\triangle_1'$/0.6/0.9,
13/2/$\triangle_1''=\triangle_2'$/-1/0.3,
12/0/$\triangle_2=\triangle_2''$/1.8/0,
9/-2/$\triangle_3$/-0.3/0,
13/-2/$\triangle_3''$/1.3/0
}
\draw (\x,\y)  node [label={[shift={(\p,\q)}]\small{\labl}}]{} rectangle ++ (1,1);

\foreach \a/\b/\c/\d/\e/\f/\g/\h/\l/\m/\n/\o/\p in {
0.75/0.25/0.75/0.75/2.25/1.25/5.25/2.75/left/left/above/left/0,
0.75/0.25/2.75/1.25/4.25/0.75/5.75/2.25/below/above/right/above/0,
1.25/-1.75/0.25/0.25/1.75/-1.25/5.25/-1.75/left/above/below/right/0,
10.25/1.25/8.75/0.75/10.25/1.75/13.25/2.75/right/below/below/right/0,
13.75/2.25/13.25/2.25/12.75/0.75/12.25/0.75/above/above/below/below/0,
8.25/0.25/8.25/0.75/9.25/-1.25/13.75/-1.75/right/right/below/above/0
}
{
\draw [-](\a,\b) -- (\c,\d) -- (\e,\f) -- (\g,\h);
\draw[fill=white] (\a,\b) node [\l]{\tiny{$i$}}circle [radius=0.05];
\draw[fill=black] (\c,\d) node [\m]{\tiny{$k$}}circle [radius=0.05];
\draw[fill=black] (\e,\f) node [\n]{\tiny{$k'$}}circle [radius=0.05];
\draw[fill=white] (\g,\h) node [\o]{\tiny{$j$}}circle [radius=0.05];
}
\end{tikzpicture}
\caption{Two examples of links and underlying generalized polymere for a fixed tree structure
on $\tilde Y$.
Note that cubes and even the notes can coincide as long as the conditions 
$k_q\in\bigcup_{p=1}^{q-1}\tilde{\triangle}_p\cup\triangle_0$
and $k_q'\in\tilde{\triangle}_q$ are fulfilled.}
\label{fig:tree}
\end{figure}

\begin{rem}
The most standard way to do a cluster expansion would be to interpolate directly in the 
total complex covariance $B$. The propagator $G(s)CG(s)$ would then be replaced by $B$.
Nevertheless, in order to bound our expressions, we need   $(\re B(s)^{-1})^{-1}$ to behave
similar to $C,$ and it is not easy to compare these two operators. 

One may also use a standard cluster expansion (e.g. a Brydges-Kennedy Taylor forest formula 
or Erice type cluster expansion \cite{abdesselam-rivasseau,brydgesleshouche}) in the real covariance $C$. 
Since the propagator $G(s)C G(s)$ has an $s$-dependence, derivatives in $s$ could also 
fall on it, which complicates the algebra involved in factorizing the contributions from
different connected components.

Therefore, we use  the same ``inductive'' interpolation scheme as in \cite{disertori-pinson-spencer}
(analog to older versions of cluster expansions, cf. \cite[Chapter III.1]{rivasseau}).
\end{rem}

\begin{proof}
We construct the cluster expansion by an inductive argument.
The large volume is divided into cubes of size $W^2 (\ln W)^\alpha$. 
In the following, we want to extract the set of cubes interacting with the observable.
Therefore, we test if there exists a connection between the root cube $\triangle_0$ 
and some other cube $\triangle\subset \Lambda$.

We introduce an interpolating covariance $B(s_1)$ with $0\leq s_1\leq 1 $,
which satisfies $B(1)=B$ while $B(0)$ decouples the root cube $\triangle_0$ from the rest of the volume. 
We define $B(s_1)^{-1} \coloneqq C(s_1)^{-1} + i\sigma_E m_i^2$, where
\begin{align*}
C(s_1)_{ij} \coloneqq \left\{
\begin{array}{l l}
s_1  C_{ij} & \text{if } i \in \triangle_0 \text{ and } j \in \triangle\neq \triangle_0 \text{ or vice versa} ,\\ 
C_{ij} & \text{otherwise}.
\end{array} \right.
\end{align*}
This is equivalent to $C(s_1) = s_1 C + (1-s_1) (C_{\triangle_0\triangle_0}+C_{\triangle_0^C\triangle_0^C})$,
where $C_{\triangle\triangle}$ is the covariance $C$ restricted to the set $\triangle$. 
By this definition, $C(s_1)$ is still a positive operator because it is a convex combination of positive operators.
Define
\begin{align*}
F^{(l_0)}_\Lambda[s_1] \coloneqq\int \diff \mu_{B(s_1)}(M) \partial_{a_{l_0}} \eto^{\mathcal{V}(M)}
\end{align*}
and $F^{(l_0,\dots,l_n)}_\Lambda[s_1]$ similarly.
Note that for $s_1=1$ we have $F^{(\ast)}_\Lambda[s_1]_{|s_1=1}=F^{(\ast)}_\Lambda$.
By the fundamental theorem of calculus
\begin{align*}
F^{(\ast)}_\Lambda[s_1]_{|s_1=1} = F^{(\ast)}_\Lambda[s_1]_{|s_1=0} 
+\int_0^1 \diff s_1 \ \partial_{s_1}F^{(\ast)}_\Lambda[s_1].
\end{align*}
$F^{(\ast)}_\Lambda[s_1]_{|s_1=0}$ corresponds to 
decoupling $\triangle_0$ from the remaining volume. 
By Lemma \ref{lem:cluster:complementY}, the integral over  $\triangle_0^C$ 
yields one in the case $\ast = l_0$ or if all indices $l_0,\dots,l_n$ are in $\triangle_0$,
and zero otherwise.
The derivative is written by integration by parts as
\begin{align*}
\int \partial_{s_1} \diff \mu_{B(s_1)}(M) [\cdot]
&= \int \diff \mu_{B(s_1)}(M) \sum_{i_1j_1}\partial_{s_1}B(s_1)_{i_1j_1}
\tfrac{1}{2} \Str \partial_{M_{i_1}}\partial_{M_{j_1}}[\cdot],
\end{align*}
Moreover, the propagator $\partial_{s_1}B(s_1)_{i_1 j_1}$ gives three connections
\begin{align*}
\partial_{s_1}B(s_1)_{i_1 j_1} 
&=\sum_{\triangle_1\neq \triangle_0}\sum\limits_{\substack{k_1\in\triangle_0 \\ k_1'\in\triangle_1}} 
G(s_1)_{i_1 k_1} C_{k_1 k_1'} G(s_1)_{k_1'j_1}+ G(s_1)_{i_1 k_1'} C_{k_1' k_1} G(s_1)_{k_1j_1},
\end{align*}
where $G(s_1) = (1+i\sigma_Em_i^2 C(s_1))^{-1}$. 
Since the matrices $C$ and $G(s_1)$ are symmetric, 
one can rewrite the second summand as $G(s_1)_{j_1 k_1} C_{k_1 k_1'} G(s_1)_{k_1' i_1}$. 
To sum the two terms, note that the supertrace is invariant under changing $i_1$ and $j_1$.
Therefore we obtain
\begin{align*}
\partial_{s_1}F^{(\ast)}_\Lambda[s_1] =
\sum\limits_{\substack{(i_1,j_1)\\(k_1,k_1')}}
\left(G(s_1)_{i_1 k_1} C_{k_1 k_1'} G(s_1)_{k_1'j_1}\right)F^{(\ast)}_\Lambda[s_1]((i_1,j_1)),
\end{align*}
where
\begin{align*}
F^{(l_0)}_\Lambda[s_1]((i_1,j_1)) = 
\int\diff\mu_{B(s_1)}(M) \Str \partial_{M_{i_1}}\partial_{M_{j_1}} \partial_{a_{l_0}} \eto^{\mathcal{V}(M)}
\end{align*}
and $F^{(l_0,\dots,l_n)}_\Lambda[s_1]((i_1,j_1))$ is defined similarly.

For each pair $(k_1,k_1')$ with $k_1\in \triangle_0$ and $k_1'\in\triangle_1$, 
there is a strong connection between $\triangle_0$ and $\triangle_1$, 
but there is no corresponding derivative in the functional integral as for $i_1$ and $j_1$.
If $i_1$ or $j_1$ belong to some cube $\triangle \nsubseteq \triangle_0\cup\triangle_1$, 
they give some additional connections.
Therefore, the first step of induction extracts a link consisting of three connection 
between the four points $i_1,j_1,k_1 $ and $k_1'$. 
This link connects $\triangle_0$ to a set of one, two or three new cubes, 
which we call the generalized cube $\tilde\triangle_1$ (cf. Figure \ref{fig:generalizedcubes}).

Now, we fix $(i_1,j_1),(k_1,k_1')$ corresponding to a connection 
between $\triangle_0$ and $\tilde\triangle_1$. 
We test if there is a connection between $\tilde\triangle_{0,1} = \triangle_0\cup\tilde\triangle_1$ 
and any other cube $\triangle'$. 
For this, we define for $0\leq s_2\leq 1$ the real interpolating covariance as
\begin{align*}
C(s_1,s_2)_{ij} = \left\{
\begin{array}{l l}
s_2  C(s_1)_{ij} & 
\text{if } i \in \tilde\triangle_{0,1} \text{ and } j \notin \tilde\triangle_{0,1}
\text{ or vice versa,} \\ 
C(s_1)_{ij} & \text{otherwise}.
\end{array} \right.
\end{align*}
We can write $C(s_1,s_2)$ again as a convex combination of positive operators  
\begin{align*}
C(s_1,s_2) = s_2 C(s_1) +(1-s_2) (C_{\tilde\triangle_{0,1}\tilde\triangle_{0,1}}(s_1)
+C_{\tilde\triangle_{0,1}^C\tilde\triangle_{0,1}^C}(s_1)),
\end{align*}
thus $C(s_1,s_2)$ is still positive. 
Now, $F^{(\ast)}_\Lambda[s_1]((i_1,j_1))=F^{(\ast)}_\Lambda[s_1,s_2]((i_1,j_1))_{|s_2=1}$. 
By the fundamental theorem of calculus
\begin{align*}
F^{(\ast)}_\Lambda[s_1,s_2]((i_1,j_1))_{|s_2=1} = 
F^{(\ast)}_\Lambda[s_1,s_2]((i_1,j_1))_{|s_2=0} +
\int_0^1 \diff s_2 \ \partial_{s_2}F^{(\ast)}_\Lambda[s_1,s_2]((i_1,j_1)).
\end{align*}
As before $F^{(\ast)}_\Lambda[s_1,s_2]((i_1,j_1))_{|s_2=0}$ corresponds to 
the functional integral restricted to $\tilde\triangle_{0,1}$ 
(if all indices $l_0,\dots,l_n$ are in $\tilde{\triangle}_{0,1}$, otherwise it is zero). 
The derivative in $s_2$ of $F^{(\ast)}_\Lambda[s_1,s_2]((i_1,j_1))$ gives
\begin{align*}
\sum\limits_{\substack{(i_2,j_2),(k_2,k_2')\\k_2\in\tilde\triangle_{0,1},k_2'\notin\tilde\triangle_{0,1}}}
\left[G(s_1,s_2)_{i_2 k_2} C(s_1)_{k_2 k_2'} G(s_1,s_2)_{k_2'j_2}\right]
F^{(\ast)}_\Lambda[s_1,s_2]((i_1,j_1),(i_2,j_2)).
\end{align*}
Note that $i_2$ and $j_2$ are arbitrary but $k_2$ needs to be in $\tilde\triangle_{0,1}$ and $k_2'$ in a new cube. 

We repeat this argument until we construct all possible connected components containing the root cube. 
Note that in the second case, only generalized polymers containing all indices $l_0,\dots,l_n$ give
a non-zero contribution.
This is a finite sum for $\Lambda$ fixed. 
The $k_r-k_r'$ connections build a tree structure on the generalized cubes,
while the positions of $i_r$ and $j_r$ are arbitrary (cf. Figure \ref{fig:tree}).
\end{proof}

\subsection{Decay of \texorpdfstring{$\mathbf{G_q(s)}$}{Gp(s)} and \texorpdfstring{$\mathbf{B(s)}$}{B(s)}}
\label{subsec:decayGB}

First we determine the decay of the propagator $G_q(s)$ and the interpolated complex covariance $B(s)$.
Note that, for $C(s)$, we can simply use that $C(s)_{ij} \leq C_{ij}$. 
\begin{lemma}
\label{lem:cluster:decayBG}
The decays of $B(s)$ and 
$G_q(s) $, respectively, are given by 
\begin{align*}
|B(s)_{ij}|&\leq |C_{ij}|+\tfrac{K}{W^2}\eto^{-f\frac{m_r}{W}|i-j|} \\
|G_p(s)_{ij}|&\leq \delta_{ij} +|C_{ij}|+\tfrac{K}{W^2}\eto^{-f\frac{m_r}{W}|i-j|}
\end{align*}
where  $f= \inf [1/2,g]$ with a constant $g<1$ independent of $W$. 
\end{lemma}

\begin{rem}
Note that the decay of $B(s)$ and $G(s)$ is bounded by 
\begin{align*}
|C_{ij}|+\tfrac{K}{W^2}\eto^{-f\frac{m_r}{W}|i-j|} \leq 
\begin{cases}
\frac{K}{W^{2}} \ln\left(\frac{W}{m_r(1+|i-j|)}\right) & \text{if } |i-j|\leq \frac{W}{m_r},\\
\frac{K}{W^{2}} \eto^{-f\frac{m_r}{W}|i-j|} & \text{if } |i-j| > \frac{W}{m_r}.
\end{cases}
\end{align*}
\end{rem}

\begin{proof}
The proof works exactly like the one in \cite[Lemma 15]{disertori-pinson-spencer}, 
replacing the three dimensional decay \eqref{eq:preliminary:expdecayC3d} with the two dimensional decay
given in \eqref{eq:preliminary:expdecayC}, the key relation being
$\sum_{k\in\Lambda}(C(s)_{jk} \exp(\mu|k-j|))^2\leq KW^{-2}$
for  $\mu <  m_r/(2W)$. 
\end{proof}

\subsection{Bounding the functional integrals}
\label{subsec:funtint}
To estimate \eqref{eq:cluster:functionalintegral1}, 
we fix a generalized polymer $\tilde Y$ and indices $\{i_{q},j_{q}\}$, and we define 
$\mathcal{J} = \{i_{q},j_{q}:q=1,\dots,r\}\subseteq Y$ 
as the set of all derived indices.
Then the corresponding integrand can be written as
\begin{align*}
\prod_{p=1}^r \Str \partial_{M_{i_p}}\partial_{M_{j_{p}}}
\left[ \partial_{a_{l_{0}}}\eto^{\mathcal{V}(M)}\right]
=  \partial_{a_{l_{0}}}\prod_{j\in Y/\mathcal{J}}\eto^{\mathcal{V}(M_j)}
\sum_{d\in \mathcal{D}}\prod_{j\in\mathcal{J}}
\partial_{M_j}^{d_j}\eto^{\mathcal{V}(M_j)}
\end{align*}
where $\partial_{M_j}^{d_j} \coloneqq \partial_{a_j}^{d_j(a)}\partial_{b_j}^{d_j(b)}
\partial_{\bar\rho_j}^{d_j(\bar\rho)}\partial_{\rho_j}^{d_j(\rho)},$ and  
$\mathcal{D} = \{d = \{d_j\}_{j\in\mathcal{J}}\}$ 
is a set of multi-indices with $d_j = 
(d_j(a),d_j(b),d_j(\bar\rho),d_j(\rho))$.
Note that $d_j(\bar\rho),d_j(\rho)\in\{0,1\}$ and $|d_j| \coloneqq d_j(a)+d_j(b)+d_j(\bar\rho)+ d_j(\rho)$
equals the multiplicity of $j$ in $\mathcal{J}$.
For the case $j=l_0$, we have an additional derivative in $a_{l_0}$ 
which needs to be treated separately.
Computing the derivatives for each $j\in \mathcal{J}\cup \{l_{0} \}$ and each multi-index $d_j$,
\begin{align*}
\partial_{a_{l_0}}^{\delta_{jl_0}}\partial_{M_j}^{d_j}\eto^{\mathcal{V}(M_j)}
= \sum_{r_j} M_j^{r_j} C_{d_{j},r_{j}}(a_j,b_j)
\eto^{\mathcal{V}_{j}(a,b)}\eto^{
- \bar\rho_{j}\rho_{j} D_{j} (a,b)[1-d_j(\bar\rho)d_j(\rho)] },\  
\end{align*}
where $M_j^{r_j} \coloneqq a_j^{r_j(a)}b_j^{r_j(b)}\bar\rho_j^{r_j(\bar\rho)}\rho_j^{r_j(\rho)},$
$r_j = (r_j(a),r_j(b),r_j(\bar\rho),r_j(\rho))$ are the remaining powers of the variables in $M,$
and $C_{d_{j},r_{j}} (a_{j},b_{j})$ is a bounded function remaining after derivatives have 
been taken. Note that we use the notation $d_{j}=0$ for $j\notin \mathcal{J}$, and the same for $r_{j}$
for $j\notin\mathcal{J}\cup\{l_0\}$. Using the definitions \eqref{eq:cluster:defV(M)}, \eqref{eq:strategy:defD}, 
\eqref{eq:strategy:defV}, and the relation 
$\partial_{a}^{d}\exp ( V(a) ) =\sum_{k=1}^{d} \binom{d}{k} \partial_{a}^{d-k}\left[ (\partial_{a}V (a))^{k} \right]
 \exp ( V(a)), $
 one can see that 
\begin{align*}
|C_{d_{j},r_{j}}(a_j,b_j)| \leq K^{d_j(a)+d_j(b)}d_j(a)!d_j(b)!
\end{align*}
independent of $r_{j}$ for all $(a_{j},b_{j})$ configurations.
Note that $n_j \coloneqq |r_j| + |d_j|\geq 3$ and  $|r_j|\leq 3 |d_j|$ 
for all $j\in\mathcal{J}\backslash\{l_0\}$.
If $l_0\notin \mathcal{J}$, we have $|r_{l_0}| =2$ and if $l_0\in\mathcal{J}$, we have at least $n_{l_0} \geq 2$.

\begin{lemma}\label{lem:cluster:boundfunctint}
The functional integral \eqref{eq:cluster:functionalintegral1} is bounded by
\begin{align}
\label{eq:cluster:boundfunctint}
|F_{{T}}^{(l_0)}[s](\{i_q,j_q\})| \leq
K_1^{|Y|(\ln W)^\alpha} \sum_{d\in\mathcal{D}}\sum_{\{r_j\}_{j\in\mathcal{J}}}\prod_{\triangle\in Y}
\left[K_2^{n_\triangle}r_\triangle!
\prod_{j\in \mathcal{J}\cap \triangle}d_j! \left(\tfrac{\ln W}{W^2}\right)^{\frac{r_j}{2}}\right],
\end{align}
where $K_1$ and $K_2$ are constants, and $|Y|$ denotes the number of cubes in $Y$.
\end{lemma}

\begin{proof}
We first compute the Fermionic integral and estimate the resulting determinant (see the two lemmas below).
We obtain
\begin{align}
\label{eq:cluster:fermint}
&\left|\int \diff\mu_{B(s)}(\bar\rho,\rho) 
\eto^{-\sum_{j\in J_4}\bar\rho_j\rho_j D_j}
\left(\prod_{j\in J_1}\rho_j\bar\rho_j\right)
\left(\prod_{j\in J_2}\rho_j\right)
\left(\prod_{j\in J_3}\bar\rho_j\right)
\right|\\
\leq& \prod_{\triangle\in Y} 
r_\triangle(\rho)^{r_\triangle(\rho)}r_\triangle(\bar\rho)^{r_\triangle(\bar\rho)}
\left(\tfrac{\sqrt{\ln W}}{W}\right)^{r_\triangle(\rho) + r_\triangle(\bar\rho)}
\eto^{\re \Tr (B(s)D)_{J_3\cup J_4,J_3\cup J_4} } \eto^{K(|Y|(\ln W)^\alpha+|\mathcal{J}|)},\nonumber
\end{align}
where  $Y = \bigcup_{j = 1}^5 J_j$ is a partition given by 
\begin{align*}
J_1 &= \{j\in Y : r_j(\rho) = r_j(\bar\rho) = 1\}, \\
J_2 &= \{j\in Y : r_j(\rho) = 1, r_j(\bar\rho) = 0\},  \\
J_3 &= \{j\in Y : r_j(\rho) = 0, r_j(\bar\rho) = 1\} ,\\
J_4 &= \{j\in Y : r_j(\rho) = r_j(\bar\rho) = 0, d_j(\rho)=d_j(\bar\rho)=0\}, \\
J_5 &= \{j\in Y : r_j(\rho) = r_j(\bar\rho) = 0, d_j(\rho)=d_j(\bar\rho)=1\}.
\end{align*} 

The remaining Bosonic functional integral is bounded by generalizing the results of the finite volume case.
We first insert the partition into the different domains of integration for each cube separately.
In $I^1$ we need the factors $|a_j|^{r_j(a)}|b_j|^{r_j(b)}$ of $M_j^{r_j}$
since they give additional small factors of order $W^{-1}(\ln W)^{1/2}$.
In the other regions these factors can be bounded by $\exp(\mathcal{V}_j(a,b))$ and
we include them into $C(a_j)C(b_j)$.

Summarizing the above procedure,
we estimate $|F^{(l_0)}_{T}[s](\{i_{q},j_{q}\})|$ by
\begin{align*}
&\eto^{K(|Y|(\ln W)^\alpha+|\mathcal{J}|)} 
\int\left|\diff\mu_{B(s)}(a,b)\right|
\eto^{\re \Tr (B(s)D)_{J_3J_3} + \re \Tr (B(s)D)_{J_4J_4}}
\prod_{j\in Y} \left|\eto^{\mathcal{V}_j(a,b)} \right|
\\&\times
\sum_{d\in\mathcal{D}}\sum_{\{r_j\}_{j\in \mathcal{J}}}
\prod_{\triangle\in Y}\frac{r_\triangle(\rho)^{r_\triangle(\rho)} r_\triangle(\bar\rho)^{r_\triangle(\bar\rho)}}{W^{r_\triangle(\rho)+r_\triangle(\bar\rho)}}
\\&\times
\prod_{j\in\mathcal{J}\cap \triangle} \left(
\chi[I_\triangle^1] |a_j|^{r_j(a)} |b_j|^{r_j(b)} + \chi[(I_\triangle^1)^C] \right)
  K^{d_j(a)+d_j(b)} d_j(a)! d_j(b)!
\end{align*}
We first apply Lemma \ref{lem:preliminary:dmuB} which also holds for $B(s)$ and $C(s)$.
Proceeding as in the proof of Theorem \ref{theo:preliminary:largevolume} 
we insert the bounds of Lemma \ref{lem:preliminary:estimates}. As a
result we obtain \eqref{eq:preliminary:prooflargevolume1} and 
\eqref{eq:preliminary:prooflargevolume2} 
with $C$ replaced by $C(s)$ and $C_{f}$ replaced by $C_f(s)= (C(s)^{-1}-fm_r^2)^{-1}>0.$ 
Now we have $C(s) < C_N = (-W^2\Delta_N+m_r^2)^{-1}$ since $C(s)$ can be represented as 
a quadratic form of block diagonal pieces of $C$ and each of these is smaller (as a quadratic form)
than $C_N$ by the arguments of Lemma \ref{lem:appcov:mass}.
This decouples the different cubes.
In $I^1$ we obtain for each $j\in\mathcal{J}$ a factor
\begin{align*}
r_j(a)!r_j(b)!\left(\tfrac{K(\ln W)^{1/2}}{W}\right)^{r_j(a)+r_j(b)}. 
\end{align*}
In the other regions we extract this factor from the exponential decay. 
The factorials in $d_j(a)$ and $d_j(b)$ are bounded by $d_j!$ and 
the factors in $r_\triangle$ are bounded by $K^{r_\triangle}r_\triangle!$.
Finally, we end up with \eqref{eq:cluster:boundfunctint}.
Note that we obtain an additional factor $W^{-2}\ln W$ in the case, where $l_0\notin\mathcal{J}$.
In the other case, we extract the precision later.
\end{proof}

\begin{lemma}\label{lem:cluster:detM}
The Fermionic integral (\ref{eq:cluster:fermint}) can be written as 
\begin{align}\label{eq:cluster:fermintdet}
\prod_{\triangle\in Y} 
r_\triangle(\rho)^{r_\triangle(\rho)}r_\triangle(\bar\rho)^{r_\triangle(\bar\rho)}
\left(\tfrac{\sqrt{\ln W}}{W}\right)^{r_\triangle(\rho) + r_\triangle(\bar\rho)}
\sigma \det M,
\end{align}
where $\sigma$ is a sign and the matrix $M$ is given by $M = (M_{J_i J_j})_{i,j = 1}^5$ with blocks
\begin{align*}
(M_{J_i J_j})_{\alpha\beta} &= 
\left(\tfrac{W^2}{r_{\triangle}(\rho_\alpha) r_{\triangle}(\bar\rho_\beta)\ln W} B(s)_{\alpha\beta}\right)_{\substack{\alpha\in J_i, \beta\in J_{j'}}}
&& \text{for } i,j \in \{1,2\},\\
(M_{J_i J_j})_{\alpha\beta} &= 
\left(\tfrac{W}{r_{\triangle}(\rho_\alpha) \sqrt{\ln W}} (DB(s))_{\alpha\beta}\right)_{\substack{\alpha\in J_i, \beta\in J_j}} 
&& \text{for } i\in\{1,2\} \text{ and } j\in\{3,4\},\\
(M_{J_i J_j})_{\alpha\beta} &=
\left(\tfrac{W}{r_{\triangle}(\bar\rho_\beta) \sqrt{\ln W}} B(s)_{\alpha\beta}\right)_{\substack{\alpha\in J_i, \beta\in J_{j'}}} 
&& \text{for } i\in\{3,4,5\} \text{ and } j\in\{1,2\},\\
(M_{J_i J_j})_{\alpha\beta} &= 
\left((1+DB(s))_{\alpha\beta}\right)_{\substack{\alpha\in J_i, \beta\in J_j}} 
&& \text{for } i \in \{3,4,5\}, j\in \{3,4\},\\
(M_{J_i J_j})_{\alpha\beta} &= 0 && \text{for } i \in \{1,2,3,4\}, j=5,\\
(M_{J_i J_j})_{\alpha\beta} &= \delta_{\alpha\beta} && \text{for } i=j=5,
\end{align*}
where $j'= 1$ for $j = 1$ and $j' =3 $ for $j = 2$.
\end{lemma}

\begin{proof}
Note that the integral is zero unless $|J_2| = |J_3|$ 
because of the symmetry of the Fermionic Gaussian integral. 
Computing the Fermionic integral (\ref{eq:cluster:fermint}) by \eqref{eq:appsusy:intminor}, it is equal to
\begin{align*}
\sigma \det B(s) {\det}_{J_1\cup J_3, J_1\cup J_2}(B(s)^{-1} + \tilde D),
\end{align*}
where $\sigma$ is a sign, $\tilde D$ is a diagonal matrix with $\tilde D_j = D_j [1-d_j(\bar\rho)d_j(\rho)]$
and $\det_{IJ} A $ is the determinant of the minor of $A$, 
where the rows with indices in $I$ and the columns with indices in $J$ are crossed out. 
Since we estimate the absolute value in the next step, we do not need the precise sign. 
We assume without loss of generality that the indices $j\in\Lambda$ are ordered such that
$B(s)$ is a block matrix of the form $B(s)= (B(s)_{J_iJ_j})_{i,j=1}^5$ and 
$B(s)_{J_i,J_j} = (B(s)_{\alpha\beta})_{\alpha\in J_i, \beta\in J_j}$.

To simplify this expression the minor is extended to a $|\Lambda|\times|\Lambda|$ matrix 
without changing the determinant up to a sign in the following way:
\begin{align*}
M_1 = \left[
\begin{matrix}
A &
\ast &
\ast &
\ast &
\ast \\
0 &
0 &
(B(s)^{-1})_{J_2J_3} &
(B(s)^{-1})_{J_2J_4} &
(B(s)^{-1})_{J_2J_5} \\
0 &
A' &
\ast &
\ast &
\ast \\
0 &
0 &
(B(s)^{-1})_{J_4J_3} &
(B(s)^{-1}+D)_{J_4J_4} &
(B(s)^{-1})_{J_4J_5} \\
0 &
0 &
(B(s)^{-1})_{J_5J_3} &
(B(s)^{-1})_{J_5J_4} &
(B(s)^{-1})_{J_5J_5} \\
\end{matrix}
\right] ,
\end{align*}
where the blocks $A$ and $A'$ have determinant one. The blocks $\ast$ can be chosen arbitrarily. 
We choose $A$ and $A'$ as the identity, $(M_1)_{J_1 J_2} = 0$ and 
the other freely selectable blocks $(M_1)_{J_i J_j}$ as $ (B(s)^{-1} + D)_{J_i J_j}$. 
By multiplying with $B(s)$ from the left, we obtain
\begin{align*}
\tilde M = B(s) M_1 = \left[
\begin{matrix}
B(s)_{J_1 J_1} &
B(s)_{J_1 J_3} &
(B(s)D)_{J_1J_3} &
(B(s)D)_{J_1J_4} &
0\\
B(s)_{J_2 J_1} &
B(s)_{J_2 J_3} &
(B(s)D)_{J_2J_3} &
(B(s)D)_{J_2J_4} &
0\\
B(s)_{J_3 J_1} &
B(s)_{J_3 J_3} &
(1+B(s)D)_{J_3J_3} &
(B(s)D)_{J_3J_4} &
0\\
B(s)_{J_4 J_1} &
B(s)_{J_4 J_3} &
(B(s)D)_{J_4J_3} &
(1+B(s)D)_{J_4J_4} &
0\\
B(s)_{J_5 J_1} &
B(s)_{J_5 J_3} &
(B(s)D)_{J_5J_3} &
(B(s)D)_{J_5J_4} &
(1)_{J_5J_5}\\
\end{matrix}\right]
\end{align*}
Extracting a factor $r_{\triangle_j}(\rho)\sqrt{\ln W}/W$ for each $j \in J_1\cup J_2$ from lines $J_1$ and $J_2$
and a factor $r_{\triangle_j}(\bar\rho)\sqrt{\ln W}/W$ for each $j \in J_1\cup J_3$ from columns $J_1$ and $J_2$, 
we obtain (\ref{eq:cluster:fermintdet}). 
Note that columns of $B(s)$ with indices in $J_3$ become columns with indices in $J_2$ in $\tilde M$ 
such that we need to extract the factors from column $J_2$. 
\end{proof}

\begin{lemma}
\label{lem:cluster:detMbound}
The determinant of the matrix $M$ can be bounded by 
\begin{align*}
|\det M| \leq K
\eto^{\re \Tr (B(s)D)_{J_3\cup J_4,J_3\cup J_4} } \eto^{K(|Y|(\ln W)^\alpha+|\mathcal{J}|)}.
\end{align*}
\end{lemma}

\begin{proof}
We use the usual bound for determinants (\ref{eq:preliminary:det1+A}) with $A= M-1$. Since 
\begin{align*}
\Tr A^* A = \sum_{i,j= 1}^5  \Tr (A_{J_i J_j})^* A_{J_i J_j} \text{ and } 
\Tr (A_{J_i J_j})^* A_{J_i J_j} = \sum_{\alpha\in J_i, \beta\in J_j} \bar{A}_{\alpha \beta} A_{\alpha\beta},
\end{align*}
each block can be bounded separately. 
For $j=5$ and all $i$ the trace above is zero.
For $i= 3,4,5$ and $j=3,4$, we have $A_{J_i J_j} = (DB(s))_{J_i J_j}$ and we estimate
\begin{align*}
\Tr (A_{J_i J_j})^* A_{J_i J_j} 
= \sum_{\alpha\in J_i,\beta\in J_j}  |B(s)_{\alpha\beta}|^2|D_\beta|^2
\leq K\sum_{\alpha\in Y}  \frac{1}{W^2} 
\leq K |Y|(\ln W)^{\alpha},
\end{align*}
where we used $|D_\beta|\leq K$ and the decay of $B(s)$ (Lemma \ref{lem:cluster:decayBG}).
For $i= 1,2$ and $ j= 3,4$, we only have off-diagonal terms and 
$A_{J_i J_j} = \left(\frac{W}{r_{\triangle}(\rho)\sqrt{\ln W}} B(s)D\right)_{J_i J_j}$. 
Using $r_\triangle \geq 1$ and $|D_\beta| \leq K$, we estimate
\begin{align*}
\Tr (A_{J_i J_j})^* A_{J_i J_j} 
\leq K\tfrac{W^2}{\ln W}\sum_{\alpha\in J_i,\beta\in J_j} |B(s)_{\alpha\beta}|^2 
\leq K\tfrac{W^2}{\ln W} \sum_{\alpha\in J_i} \tfrac{1}{W^2} \leq K \tfrac{|\mathcal{J}|}{\ln W},
\end{align*}
where the sum over $\beta\in J_j$ is extended to $\beta \in \Lambda$.
We bound the trace similarly for the case $i= 3,4,5$ and $ j= 1,2$,
where $A_{J_i J_j} = \left(\frac{W}{r_\triangle(\bar{\rho})\sqrt{\ln W}}B(s)\right)_{J_i J_j'}$. 
Extending the sum over $\alpha\in J_i$ to $\alpha \in \Lambda$, we end up with
\begin{align*}
\Tr (A_{J_i J_j})^* A_{J_i J_j} \leq
\tfrac{W^2}{\ln W}\sum_{\alpha\in J_i,\beta\in J_j} |B(s)_{\alpha\beta}|^2 
\leq  K\tfrac{W^2}{\ln W}\sum_{\beta J_j} \tfrac{1}{W^2} \leq K \tfrac{|\mathcal{J}|}{\ln W},
\end{align*}
For $i,j = 1,2$, we have 
$A_{J_i J_j}=\left( \frac{W^2}{r_\triangle(\rho)r_\triangle(\bar{\rho})\ln W}B(s)-1\right)_{J_i J_j'}$. 
Summing the trace of the quadratic term and the corresponding block of $\re \Tr A$, 
terms linear in $M$ cancel and we end up with
\begin{align*}
&\tfrac{1}{2} \Tr ((M_{J_iJ_j})^* -1 )(M_{J_i J_j}-1) + \re \Tr (M_{J_i J_j} -1)
\leq\tfrac{1}{2}\Tr (M_{J_i J_j})^* M_{J_i J_j}.
\end{align*} 
Now, the term $\Tr  (M_{J_i J_j})^* M_{J_i J_j}$ is bounded by using the factor $r_\triangle$ explicitly.
Rewriting the sum over $\alpha\in J_i$ and $\beta\in J_j$ into a sum over cubes, we obtain 
\begin{align*}
\Tr (M_{J_i J_j})^* M_{J_i J_j} &
\leq \sum_{\alpha\in J_i,\beta\in J_j} 
\frac{W^4}{r_{\triangle_\alpha}(\rho)^2 r_{\triangle_\beta}(\bar\rho)^2(\ln W)^2}|B(s)_{\alpha\beta}|^2\\
& \leq \sum\limits_{\substack{\triangle: \triangle\cap J_i \neq \emptyset \\ \triangle': \triangle'\cap J_j \neq \emptyset}}
\sum\limits_{\substack{\alpha\in\triangle\cap J_i \\ \beta\in\triangle'\cap J_j}}
\frac{W^4}{r_{\triangle}(\rho)^2 r_{\triangle'}(\bar\rho)^2(\ln W)^2}|B(s)_{\alpha\beta}|^2.
\end{align*}
The number of summands in the second sum is bounded by $r_\triangle(\rho) r_{\triangle'}(\bar\rho)$.
Applying the estimate $|\alpha-\beta| \geq W(\dist(\triangle_\alpha,\triangle_\beta)-1)$ 
and the decay of $B(s)$, we end up with
\begin{align*}
\begin{aligned}
\Tr (M_{J_i J_j})^* M_{J_i J_j} 
\leq &\sum\limits_{\substack{\triangle: \triangle\cap J_i \neq \emptyset \\ \triangle': \triangle'\cap J_j \neq \emptyset}}
\tfrac{K}{r_{\triangle}(\rho) r_{\triangle'}(\bar\rho)(\ln W)^2}  
\left( \delta_{\dist(\triangle,\triangle')-1<\frac{1}{m_r}}
\ln^2\left[\tfrac{W}{m_r}\tfrac{1}{W(\dist(\triangle,\triangle')-1)+1}\right]  \right.\\
&+ \left.\delta_{\dist(\triangle,\triangle')-1\geq\frac{1}{m_r}}
\eto^{-m_r(\dist(\triangle,\triangle')-1)}\right),
\end{aligned}
\end{align*}
where the factor $W^4$ cancels. The sum over $\triangle'$ is bounded by a constant independent of $W$ 
because of the exponential decay.
Therefore,
\begin{align*}
\Tr (M_{J_i J_j})^* M_{J_i J_j}  
\leq K\sum_{\triangle: \triangle\cap J_i \neq \emptyset} \frac{1}{r_\triangle(\rho)} \leq K |\mathcal{J}|.
\end{align*}
Combining these estimates, we end up with the result.
\end{proof}

\subsection{Summing up}
\label{subsec:summing}
In this section we will put together the estimates above to complete the proof.
The large factorials and combinatoric factors arising from  the bound of the 
functional integral and the sum over the cube positions will be controlled 
by  fractions of the exponential decay of  $G_q(s)CG_q(s),$ while
the  non-exponential part 
will allow to sum  over the vertex positions $i,j,k,k'$ inside 
each fixed cube.  Finally the sum over the tree structure
will be achieved by a standard argument.

\paragraph*{Reorganizing $\mathbf{W}$ factors.}
Before performing the estimates, we extract additional $W$ factors from  $G$
as follows:  
\begin{align}\label{eq:cluster:Wfactors}
&\prod_{q=1}^r |G_q(s)_{i_qk_q}| \ |C_{k_qk_q'}| \ |G_q(s)_{k_q'j_q}| 
\prod_{\triangle\in Y}\left(\tfrac{(\ln W)^{1/2}}{W}\right)^{r_\triangle}\\
=&
\prod_{q=1}^r \left|\tfrac{W^2}{\ln W}G_q(s)_{i_qk_q}\right| \ 
\left|\tfrac{\ln W}{W^2}C_{k_qk_q'}\right| \ 
\left|\tfrac{W^2}{\ln W}G_q(s)_{k_q'j_q}\right|
\prod_{j\in \mathcal{J}\cup\{l_{0} \} }\left(\tfrac{(\ln W)^{1/2}}{W}\right)^{n_j},\nonumber
\end{align}
where we remember that $d_{l_0}= 0$ and $r_{l_0}=2$ if $l_0\notin \mathcal{J}$.

\paragraph*{Factorials.}
We extract a small fraction of the exponential to control finite powers of factorials $d_\triangle!^p$.
The number $d_\triangle$ counts the number of multi-link starting points  $i_{q}$ and endpoints $j_{q}$
inside $\triangle.$ Denoting by $q_{0}$ the first (smallest) index in this family, one can see 
that  for all $q>q_{0},$ the cubes containing the vertex $k'_{q}$ are pairwise disjoint 
and different from $\triangle.$  For $d_\triangle$ large, more than half of these cubes
have distance of order $W(\ln W)^{\alpha/2}d_\triangle^{1/2}$ 
from $\triangle$ since we are in a finite dimensional space. 
Therefore we gain a factor of order $\exp(-(\ln W)^{\alpha/2 } d_\triangle^{3/2})$ from the exponential decay of $GCG$.
A small fraction from this beats finite powers of factorials in $d_\triangle$ and $n_\triangle$:
\begin{align}
\label{eq:cluster:factorials}
\prod_{q=1}^r \eto^{-\varepsilon|i_q-k_q|/W}\eto^{-\varepsilon|k_q-k_q'|/W}\eto^{-\varepsilon|k_q'-j_q|/W}
\leq \prod_\triangle\tfrac{K}{d_\triangle!^p}\leq \prod_\triangle\tfrac{K^{n_\triangle}}{n_\triangle!^{p/4}},
\end{align}
where we used $n_{\triangle}\leq 4 d_{\triangle}$ in the last step.
\vspace{0,2cm}

Applying Lemma \ref{lem:cluster:boundfunctint}, the bound of the factorials \eqref{eq:cluster:factorials}
and the reorganization  of the $W$ factors in \eqref{eq:cluster:Wfactors}, we have
\begin{align}\label{eq:cluster:puttingtogether}
&\prod_{q=1}^r\left(|G_q(s)_{i_qk_q}| |C_{k_qk_q'}| |G_q(s)_{k_q'j_q}| \right)|F_T[s](\{i_q,j_q\})|\\
 \leq&  \tfrac{\ln W}{W^2}K^{(\ln W)^{\alpha}}
\left[\prod_{q=1}^r \eto^{-f'd(\triangle'_{q},\triangle_{q})/W}
\eto^{-f'd(\triangle_{q},\triangle_{\mathcal{A}(q)})/W}\eto^{-f'd(\triangle_{\mathcal{A}(q)},\triangle''_{q})/W}\right]\nonumber\\
& \times
\sum_{d\in\mathcal{D}}\sum_{\{r_j\}_{j\in\mathcal{J}}}
\left(\tfrac{K^{(\ln W)^\alpha}(\ln W)^{1/2}}{W}\right)^{n_{l_0}-2}
\prod_{j\in\mathcal{J}\backslash\{l_0\}} \left(\tfrac{K^{(\ln W)^\alpha}(\ln W)^{1/2}}{W}\right)^{n_j} 
\prod_{q=1}^r \tilde{G}_{i_qk_q} \tilde{C}_{k_qk_q'} \tilde{G}_{k_q'j_q}\nonumber
\end{align}
where $f'=fm_r-\varepsilon$ is the remaining mass, $d(\triangle,\triangle')$ is the distance
between the centers of the cubes $\triangle$ and $\triangle'$, and $\tilde G$ and $\tilde C$ are
the prefactors of the exponential decay of $G$ and $C$ given by
\begin{align*}
\tilde{C}_{ij} =& 
\delta_{|i-j|\leq\tfrac{W}{m_r}} \tfrac{K\ln W}{W^4} \ln \left(\tfrac{W}{m_r}\tfrac{1}{|i-j|+1}\right)
 + \delta_{|i-j|>\tfrac{W}{m_r}} \tfrac{K\ln W}{W^{7/2}|i-j|^{1/2}}\\
\tilde{G}_{ij} =& \delta_{ij} \tfrac{W^2}{\ln W} 
+ \delta_{|i-j|\leq\tfrac{W}{m_r}} \tfrac{K}{\ln W} \ln \left(\tfrac{W}{m_r}\tfrac{1}{|i-j|+1}\right)
+ \delta_{|i-j|>\tfrac{W}{m_r}} \tfrac{1}{\ln W} .
\end{align*}

\paragraph*{Sum over the vertex position inside each cube.}
Remember that $\tilde{\triangle}_{q}=(\triangle_{q},\triangle'_{q},\triangle''_{q}),$ $q=0,\dots r.$ 
For each $\tilde{\triangle}_{q}$ we call $\tilde{\triangle}_{\mathcal{A} (q)}$ the ancestor
of  $\tilde{\triangle}_{q}$ in the tree, and  $\triangle_{\mathcal{A} (q)}$ the cube in 
 $\tilde{\triangle}_{\mathcal{A} (q)}$ containing $k_{q}.$
Let us now fix the tree structure $T,$ the position of the above cubes, and the multiplicities 
$d\in \mathcal{D}.$
\begin{lemma} 
The sum over the vertex positions $i_{q},j_{q},k_{q},k'_{q}$  compatible with the above 
constraints is bounded by
\begin{equation}\label{eq:cluster:sumpropagator}
\sum^{(T,d)}\limits_{\substack{i_q\in\triangle'_{q},j_q\in\triangle''_{q}\\
k_q\in\triangle_{q},k_q'\in\triangle_{\mathcal{A}(q)}}}
\prod_{q=1}^r \tilde{G}_{i_qk_q} \tilde{C}_{k_qk_q'} \tilde{G}_{k_q'j_q}
\leq K^{d_{l_0}}(\ln W)^{\alpha d_{l_0}}\prod_{j\in\mathcal{J}\backslash\{l_0\}} K^{d_j} W^{2} (\ln W)^{5d_{j}/2} 
\end{equation}
\end{lemma}

\begin{proof}
Each multi-link consists of four vertices $i_q,j_q,k_q,k_q'$, where $k_q$ and $k_q'$ must 
belong to different cubes while $i_q$ and $j_q$ are arbitrary. For $j=i_q$ or $j=j_q$, we say
\begin{itemize}[leftmargin=0.5cm]
\item $j$ is new in step $q$ if the $q$th multi-link extracts $j$ and $j$ was never extracted before.
\item $j$ is old in step $q$ if the $q$th multi-link extracts $j$ and $j$ was already extracted.
\end{itemize}
Since the multi-indices $d$ and the tree structure  are fixed,  the fact that
 $j$ is old or new is preserved when summing over its position inside the cubes.
We consider the different cases. Note that we only sum over $i_q$ and $j_q$ if they are new. 
If both $i_q$ and $j_q$ are of the same type (old or new), we distribute the resulting factor to both indices.
If one is old and one is new, the resulting factor counts only for the new index. 

\textit{a) $i_q \neq j_q$ and both $i_q$ and $j_q$ are new.} We sum over $i_q$ and $j_q$:
\begin{align*}
\sum_{k_q\in\triangle_{k_q}}\sum_{i_q\in\triangle_{i_q}} \tilde{G}_{i_qk_q}\sum_{k_q'\in\triangle_{k_q'}}
\tilde{C}_{k_qk_q'} \sum_{j_q\in\triangle_{j_q}} \tilde{G}_{k_q'j_q}
\leq K W^4 (\ln W)^{9\alpha/2-1} 
\end{align*}
Therefore, we pay a factor $W^2 (\ln W)^{9\alpha/4-1/2}$ for $i_q$ and the same factor for $j_q$.

\textit{b) $i_q\neq j_q$ and $i_q$ is new and $j_q$ is old. }
The same estimate holds for $i_q$ old and $j_q$ new. Then we sum only over $i_q$. 
\begin{align*}
\sum_{k_q'\in\triangle_{k_q'}}\tilde{G}_{k_q'j_q}\sum_{k_q\in\triangle_{k_q}} \tilde{C}_{k_qk_q'}
\sum_{i_q\in\triangle_{i_q}} \tilde{G}_{i_qk_q}\leq K W^2 (\ln W)^{7\alpha/2-1}.
\end{align*}
Hence, we need to bound a factor $W^2 (\ln W)^{7\alpha/2-1}$ for $i_q$ and no factor for $j_q$.

\textit{c) $i_q \neq j_q$ and $i_q$ and $j_q$ are old.} Then, $i_q$ and $j_q$ are both fixed and
\begin{align*}
\sum_{k_q\in\triangle_{k_q}}\tilde{G}_{i_qk_q}W^{-4} \ln^2 W\sum_{k_q'\in\triangle_{k_q'}}
\tilde{G}_{k_q'j_q}\leq K (\ln W)^{2\alpha} ,
\end{align*}
where we bound $\tilde{C}_{k_qk_q'}\leq W^{-4}\ln^2 W$. 
For both $i_q$ and $j_q$ we collect a factor $K (\ln W)^\alpha$.

\textit{d) $i_q = j_q$ and $i_q$ is new.} Then, we sum over $i_q$
\begin{align*}
\sum_{k_q\in\triangle_{k_q}}\sum_{k_q'\in\triangle_{k_q'}}
\tilde{C}_{k_qk_q'}
\sum_{i_q\in\triangle_{i_q}} \tilde{G}_{i_qk_q}\tilde{G}_{k_q'i_q}\leq K W^2(\ln W)^{5\alpha/2}
\end{align*}
and obtain a factor $W^2(\ln W)^{5\alpha/2}$ for $i_q$ and no factor for $j_{q}.$ 

\textit{e) $i_q= j_q$ and $i_q$ is old.} Then, $i_q$ is fixed and
\begin{align*}
\sum_{k_q\in\triangle_{k_q}}\tilde{G}_{i_qk_q}W^{-4}\ln^2 W\sum_{k_q'\in\triangle_{k_q'}}
\tilde{G}_{k_q'j_q}\leq K (\ln W)^{2\alpha}.
\end{align*}
We gain a factor $(\ln W)^\alpha$ for $i_q$ and $j_q$. Note that for $l_0 \in \mathcal{J}$, $l_0$
is always old. 
\end{proof}
Combining the products over $j\in \mathcal{J}$ in  \eqref{eq:cluster:puttingtogether} and
 \eqref{eq:cluster:sumpropagator} we obtain
\begin{align}\label{eq:cluster:bound1}
\prod_{j\in\mathcal{J}\backslash \{ l_{0}\}}
\hspace{-0,3cm} 
\left(\tfrac{K^{(\ln W)^{\alpha}}(\ln W)^{1/2}}{W}\right)^{n_j}K^{d_j}W^{2}(\ln W)^{5d_j/2}\leq
\hspace{-0,3cm}
\prod_{j\in\mathcal{J}\backslash \{ l_{0}\}} 
\left(\tfrac{K^{(\ln W)^\alpha}(\ln W)^{3}}{W^{1/3}}\right)^{n_j},
\end{align}
where we used  $n_j\geq 3$ for all $j\in\mathcal{J}, j\neq l_{0}.$ 
The point $j=l_0$ is special since $n_{l_0} \geq 2$. 
But since the position $l_0$ is fixed, $l_0$ is always ``old'' and we obtain 
\begin{align}\label{eq:cluster:bound2}
\left(K^{(\ln W)^\alpha}\tfrac{(\ln W)^{1/2}}{W}\right)^{n_{l_0}-2}(K(\ln W)^\alpha)^{d_{l_0}}.
\end{align}
Finally, we perform the sum over the multi-indices $d$ and $r,$ compatibles with the fixed
tree structure.  The sum over $r_j= (r_j(a),r_j(b),r_j(\bar\rho),r_j(\rho))$ gives a factor $K^r$ since
$|r_j|\leq 3|d_j|$ and $\sum_{j\in\mathcal{J}}|d_j| = 2r$.
The sum over  $d_j = (d_j(a),d_j(b),d_j(\bar\rho),d_j(\rho))$ can be estimated by an integral over a simplex of  length
$r$, giving an additional  factor $K^r.$ 

Combining these factors with \eqref{eq:cluster:bound1} and \eqref{eq:cluster:bound2}   
we obtain the bound  $ g^{r-1},$ where $g=K^{(\ln W)^\alpha}W^{-1/3+\varepsilon}$ with
$0<\varepsilon\ll 1/3,$ hence $g\ll 1$ for $W$ large.

\paragraph*{Sum over the cube position and the tree structure.}
For a fixed tree structure we use the remaining exponential decay of $GCG$ to sum over the positions of the
all cubes inside $\tilde{\triangle}_q$ for all $1\leq q \leq r$,
starting from the leafs (i.e. vertices with degree 1) and going towards the root $\triangle_0$. 
For each multi-link connecting a generalized cube $\tilde{\triangle}_{q}$ to its ancestor 
$\tilde{\triangle}_{\mathcal{A}(q)}$ the position of   $\triangle_{q}, \triangle'_{q}$ and $\triangle_{q}''$
is summed over using  the exponential decay  of $GCG.$  This costs only a constant factor for each cube.
Finally we pay a factor $3$ to choose the position of the ancestor in $\tilde{\triangle}_{\mathcal{A}(q)}$.
We end up with
\begin{align*}
\left|F^{(l_0)}_\Lambda\right| = \tfrac{\ln W}{W^2}
\eto^{K(\ln W)^\alpha} \left[1+ \sum_{r\geq 1} \sum_{T \text{ unordered}}
\sum_{\text{orders}}
\int_{[0,1]^r} \prod_{q=1}^r \diff s_q |M_T(s)|  \ g^{r-1}\right].
\end{align*}
Integrating over the interpolating factors $s$ cancels the last sum over the orders of the trees 
(cf. \cite[Lemma III.1.1]{rivasseau}):
$\sum_{\text{orders}}\int_{[0,1]^r} \prod_{q=1}^r \diff s_q |M_T(s)| = 1.$
The remaining sum is written as
\begin{align*}
\sum_{r\geq 1}\sum_{T} g^{r-1} \leq 1 +\sum_{r\geq 2}\sum_{T} \sqrt{g^{r}}= 
\sum_{\deg_{\tilde\triangle_0}\geq 1} \prod_{i_0=1}^{\deg_{\tilde\triangle_0}}\sqrt{g}
\left[\sum_{\deg_{\tilde\triangle_{i_0}}\geq 1} \prod_{i_1=1}^{\deg_{\tilde\triangle_{i_0}}-1} \sqrt{g} \sum \dots \right],
\end{align*}
where  $\deg_{\tilde\triangle}$ denotes the degree of the generalized cube $\tilde\triangle$ in the tree $T$. 
Since $g\ll  1$, we can sum from the leaves  towards the root
using a standard procedure (cf. \cite[Section 6.3.4]{disertori-pinson-spencer})
and bound the sum above by a constant.
It suffices to assume $\sqrt{g}< 1/4$ to make this procedure work. Hence $W_0(\alpha)$ need 
to be chosen large enough that $K^{(\ln W)^\alpha}W^{-1/3+\varepsilon}<1/16$ for all $W\geq W_0(\alpha)$. 
Finally we estimate the sum over $l_0$ using the exponential decay of $|B_{0l_0}|$. As a result
\begin{align*}
\left|\int \diff \mu_B(M) \eto^{\mathcal{V} (M)} a_0\right|\leq
\sum_{l_0\in\Lambda} |B_{0l_0}|
\left|F^{(l_0)}_\Lambda\right| \leq \tfrac{\ln W}{W^2}\eto^{K(\ln W)^\alpha}.
\end{align*}
This proves the first part of Theorem \ref{theo:strategy:reducedresult}.

\subsection{Derivatives}
\label{subsec:cluster:derivatives}
Bounding the derivative is similar to the procedure above. Our starting point is
\begin{align*}
\sum_{j_1,\dots,j_n}\sum_{l_0,\dots,l_n} B_{0 l_0}\prod_{m=1}^n B_{j_m l_m}
F_\Lambda^{(l_0,\dots,l_n)}.
\end{align*}
Since the $B$ factors control the sums over $l_0$ and over the $j_m$'s, we have only $n$ remaining sums
of $l_1,\dots,l_n$ over the volume $\Lambda$.
We observe that a cluster expansion of $F_\Lambda^{(l_0,\dots,l_n)}$ extracts only trees such that  all 
indices $l_0,\dots,l_n$ are in the connected cluster.  
We can extract a fraction of the exponential decay of $GCG$ 
to sum over the 'coarse' position of the $l_0,\dots,l_n$, i.e. the position
of the cubes containing the indices. Finally, to sum over the index position inside each cube, we 
need to extract at least a factor $(W^2 (\ln W)^\alpha)^{-1}$ for each
$l_1,\dots,l_n$.

As mentioned above, applying the cluster expansion directly to $F_\Lambda^{(l_0,\dots,l_n)}$ is not enough
to  extract this fine structure. Problems arise when two or more of the $l_j$'s coincide and
we have contributions from $\prod_{m=1}^n \Str\partial_{M_{l_m}} \partial_{a_{l_0}} \exp(\mathcal{V}(M))$
of the form $(\Str \partial_{M_l})^n\mathcal{V}(M_l)$ with $n\geq 2$. 
Since the lowest order contribution of $\mathcal{V}(M)$ is cubic, we obtain linear or constant terms in $ M$.
Note that constant terms vanish, since $(\Str \partial_{M_i})^n\Str M_j^n = \delta_{ij} n! \Str 1 = 0$.
Linear terms may display a problem if a derivative of the cluster expansion falls on them.
In this case, we have no field factor left and we gain only
only $W^{-1} (\ln W)^{1/2}$ from the derivative (cf. eq. \eqref{eq:cluster:Wfactors})
of the cluster expansion, which is not enough for the fine structure estimates.

In the special case $l_0$, problems arise for terms of the form
 $(\Str \partial_{M_{l_0}})^n\partial_{a_{l_0}}\mathcal{V}(M_{l_0})$,
with $n\geq 1$, since we obtain again linear or constant terms in $a_{l_0}$.
For linear terms we have the same problem as above.
Note that also in this case the constant term vanishes since the whole integral 
corresponds to the derivative of a constant (cf. proof of Lemma \ref{lem:strategy:dualrep}
and Lemma \ref{lem:cluster:complementY}) except in the special case  when $n\geq 2$ and all $l_k$ coincide.
Indeed, in this case the integral coincides with \eqref{eq:strategy:avarage1} and hence yields one, but 
this is no problem since we can 
sum over the remaining indices using the $B$ factors.

To solve these problems, we apply integration by parts on the linear contributions of the form $\Str M_{l_j}$
with $l_j\neq l_0$ as in \eqref{eq:appproofsusy:intbyparts}
\emph{before} performing the cluster expansion. Each new $B$ factor that we obtain ensures summation over at least one old index,
while a new index to be summed, coupled with an $\Str\partial_{M}$, appears. 
Again, the derivative may fall on the exponential $\exp(\mathcal{V}(M))$ (extracting a new term $\Str M^2$ at lowest order)
or a prefactor $\Str M^n$ for $n\geq 1$. If $n=1$ the integral vanishes by the same arguments as above.
For $n=2$ we obtain a new linear term, where we need to perform integration by parts.
In all other cases we obtain enough fine structure.
Note that the procedure ends after at most $2n$ steps.
 
Again a derivative falling on another linear
contribution vanishes since by the same arguments as above.
Therefore we end up with functional integrals of the form
\begin{align*}
\sum_{k\in\Lambda \forall k\in\mathcal{K}}
\int \diff\mu_B(M) a_{k_0}^{m_{k_0}}\prod_{k\in\mathcal{K}}\Str M_k^{m_k}  \eto^{\mathcal{V}(M)}
\end{align*}
where $m_{k_0}\geq 1$, $m_k\geq 2$ and $|\mathcal{K}|\leq n$.
Note that again the index $k_0$ is special and a constant term i.e. $m_{k_0}=0$ means the 
integral corresponds to the derivative of a constant.

Applying here the cluster expansion yields a connected tree containing all indices
$k_0,\dots,k_n$. We obtain a functional integral of the form
\begin{align*}
F_T^{(k_0,\dots,k_n)}[s](\{i_q,j_q\}) =\int \diff\mu_{B(s)}(M) 
 \prod_{q=1}^r\Str(\partial_{M_{i_q}}\partial_{ M_{j_q}})
\left[a_{k_0}^{m_{k_0}}\prod_{k\in\mathcal{K}}\Str M_k^{m_k}  \eto^{\mathcal{V}(M)}\right]
\end{align*}
and bound it similar to Lemma \ref{lem:cluster:boundfunctint}.
Note that the indices $k_0,\dots,k_n$ need to be treated 
separately as $l_0$ before.
We obtain $n_{k_0}\geq 1$ and $n_{k}\geq 2$ for $k\in\mathcal{K}$.
Since we sum later over these indices, they are 'old' and 
hence we obtain at least a total factor $W^{-(2n+1)}(\ln W)^n.$
Collecting all $W$ contributions we get
\begin{align*}
W^{2n} (\ln W)^{n\alpha } W^{-(2n+1)} (\ln W)^n \eto^{K  (\ln W)^{\alpha }} =W^{-1}  (\ln W)^{n (\alpha+1) } \eto^{K  (\ln W)^{\alpha }}\leq 1
\end{align*}
for $W$ large enough (depending on $n$). 
Note that the first factor comes from the sum over the index position inside each cube,
and the last from the contribution of the root cube (see end of Section \ref{subsec:summing}).

\begin{appendix}
\section{Supersymmetric Formalism}
\label{sec:appsusy}
We will summarize the main ideas of the supersymmetric formalism 
(see \cite{efetov} for an easy-to-read introduction and
\cite{berezin} for a detailed description).
\begin{defi}[Grassmann algebra]
Let $N\in \mathbb{N}$ and let $V$ be a vector space over a field $\mathbb{K}$
with basis $(\alpha_1,...,\alpha_N)$ and denote the antisymmetric tensor product by 
\begin{align*}
\wedge : V \times V &\to V\otimes_{as} V,\\
(v,w) &\mapsto v \wedge w = vw = -wv.
\end{align*}
The corresponding \emph{Grassmann algebra} is defined by 
\begin{align*}
\mathcal{A} \coloneqq \bigoplus_{k\geq 0} V^k,
\end{align*}
where $V^0 = \mathbb{K}$, $V^1 = V$ and $V^k = V^{k-1} \otimes_{as} V$ for $k\geq 2$. 
This is an associative algebra with unit.
We distinguish between the subsets of even elements 
$\mathcal{A}^0\coloneqq \bigoplus_{k\geq 0}V^{2k}$ and odd elements 
$\mathcal{A}^1 \coloneqq \bigoplus_{k\geq 0} V^{2k+1}$.
While the even elements form an algebra again, this is not true for $\mathcal{A}^1$.
Even elements commute with all elements in the Grassmann algebra
and are called \emph{Bosonic variables}.
On the other hand, two odd elements anticommute and are called \emph{Fermionic (or
Grassmann) variables}. 
\end{defi}
The generators $(\alpha_1,...,\alpha_N)$ of $\mathcal{A}$
are  Grassmann variables, and are provided with the anticommutation property
$
\alpha_i \alpha_j = - \alpha_j \alpha_i 
$
for all $i,j = 1,...,N$. Note that this directly implies $\alpha_i^2 = 0$.
Hence, any element in $\mathcal{A}$ is a finite polynomial of the form
\begin{align*}
f(\alpha_1,...,\alpha_N) = 
f_0 + \sum_{k = 1}^N \sum_{i_1 < ... < i_k} f_{i_1,...,i_k} \alpha_{i_1} \cdots \alpha_{i_k},
\end{align*}
where $f_0,f_{i_1,...,i_k}\in\mathbb{K}$ and $f_0$ is called
\emph{spectrum} of $f(\alpha_1,...,\alpha_N)$. 

\begin{defi}[Grassmann integration]
\label{def:appsusy:grassmannintegral}
As a formal symbol, we define the integral over a Grassmann variable as
$\int \diff \alpha_i \ 1 = 0$ and $\int \diff \alpha_i \ \alpha_i = \tfrac{1}{\sqrt{2\pi}}.$
To define integration of multiple variables, we assume Fubini's theorem applies, 
but the differentials anticommute.
\end{defi}

\begin{notation}
To keep the notation as short as possible, we write for any two families $(\zeta_i)_{i\in I}$
and $(\xi_i)_{i\in I}$ of Bosonic and/or Fermionic variables,
the sum over the corresponding index set $I$ as $(\zeta,\xi) = \sum_{i\in I} \zeta_i \xi_i$.
\end{notation}
\paragraph*{Gaussian integral.}
We will often use the following Gaussian integral formulas. 
Let $x\in \mathbb{R}^n$ and $z\in \mathbb{C}^n$.
For $M\in \mathbb{C}^{n\times n}$ with positive definite Hermitian part, 
\begin{align*}
\int \diff x \eto^{-\frac{1}{2}(x,Mx)}\eto^{(x,y)} 
&=\tfrac{(2\pi)^{n/2}}{\sqrt{\det M}} \eto^{\frac{1}{2}(y,M^{-1}y)},\\
\int \diff \bar z\diff z \eto^{-(\bar z, Mz)}\eto^{(\bar v,z) + (\bar z, w)}
&= \tfrac{(2\pi)^n}{\det M}\eto^{(\bar v,M^{-1} w)},
\end{align*}
where $y,v,w\in\mathbb{C}^n$ and the measures are usual Lebesgue product measures, i.e.
$\diff x = \prod_{i=1}^n \diff x_i$,
$\diff \bar z \diff z =\prod_{i=1}^n \diff  \bar z_i  \diff  z_i$ and
$\diff\bar z_i \diff z_i = 2 \diff \re z_i \diff \im z_i$. 
Note that the formulas remain valid if we replace $y_i,v_i$ and $w_i$ by even elements of $\mathcal{A}$.
A direct consequence of the latter are the following identities
\begin{align}
\label{eq:appsusy:idea}
\int  \diff \bar z \diff z \ \eto^{(\bar z, Mz)} =\tfrac{(2\pi)^n}{\det M}\quad \text{and}\quad
\int \diff \bar z \diff z \ z_k  \ \bar z_l \ \eto^{- (\bar z, Mz)} = M^{-1}_{kl}\tfrac{(2\pi)^n}{\det M}.
\end{align}
Using Definition \ref{def:appsusy:grassmannintegral} above, we obtain similar Fermionic formulas.
Let $(\chi_i)_{i=1}^n$ and $(\bar\chi_i)_{i=1}^n\subset\mathcal{A}_1$ be two families
of the Grassmann variables, where the $\bar\chi_i$'s are independent of the $\chi_i$'s.
For an arbitrary $M\in \mathbb{C}^{n\times n}$, we have
\begin{align}
\label{eq:appsusy:gaus}
\int \diff\bar\chi\diff \chi \eto^{- (\bar\chi, M \chi)} &= (2\pi)^{-n}\det M,\\
\nonumber
\int \diff\bar\chi\diff\chi\eto^{-(\bar\chi,M\chi)}\eto^{(\bar\rho,\chi)+(\bar\chi,\rho)} 
&= (2\pi)^{-n}\det M\eto^{(\bar\rho,M^{-1}\rho)},\\
\label{eq:appsusy:intminor}
\int \diff\bar\chi\diff \chi \eto^{- (\bar\chi, M \chi)} \prod_{i\in I}\chi_i \prod_{j\in J} \bar\chi_j
&= \sigma_{IJ}\delta_{|I|=|J|}(2\pi)^{-n}{\det}_{JI} M,
\end{align}
where $\diff\bar\chi\diff \chi = \prod_{i=1}^n \diff\bar\chi_i\diff\chi_i$ and 
$(\rho_i)_{i=1}^n$ and $(\bar\rho_i)_{i=1}^n\subset\mathcal{A}_1$ are two families of
Grassmann variables. Moreover, $\sigma_{IJ}$ is a sign,
$I,J\subset\{1,\dots,n\}$ are two index sets and 
$\det_{JI} M$ is the determinant of the minor of $M$ where the rows with indices in $J$ and 
the columns with indices of $I$ are crossed out.

\paragraph*{Supervectors and Supermatrices}
To combine real or complex variables with Grassmann ones,
we introduce the notation of a \emph{supervector} $\Phi$ consisting of
$p$ Bosonic variable $X=(X_i)_{i=1}^p\in(\mathcal{A}^0)^p$ 
and $q$ Fermionic variable $\alpha= (\alpha_j)_{j=1}^q\in(\mathcal{A}^1)^q$ by
\begin{align*}
\Phi =
\begin{pmatrix}
X \\
\alpha 
\end{pmatrix}.
\end{align*}
A \emph{supermatrix} is a linear transformation between supervectors, i.e.
\begin{align}
\label{eq:appsusy:supermatrix}
\Phi' = \mathbf{M} \Phi, \quad \mathbf{M} = 
\begin{pmatrix}
a & \sigma \\
\rho & b\\
\end{pmatrix},
\end{align}
where $a,b$ are $p\times p$ and $q\times q$ matrices in $\mathcal{A}_0$ 
and $\sigma,\rho$ are $p\times q$ and $q \times p$ matrices in $ \mathcal{A}_1$.
We denote supermatrices by bold face capital letters.
For the supermatrix $\mathbf{M}$, we define the notation of a \emph{supertrace}
and a \emph{superdeterminant} as
\begin{align}\label{eq:appsusy:strace}
\Str \mathbf{M} \coloneqq \Tr a -  \Tr b 
\quad \text{and} \quad
\Sdet \mathbf{M} \coloneqq \det[a-\sigma b^{-1} \rho] \det [b^{-1}].
\end{align}
Finally, the inverse of the supermatrix $\mathbf{M}$ is given by
\begin{align*}
\mathbf{M}^{-1} =
\begin{pmatrix}
(a-\sigma b^{-1}\rho)^{-1} &
-(a-\sigma b^{-1}\rho)^{-1}\sigma b^{-1}\\
-b^{-1}\rho (a-\sigma b^{-1}\rho)^{-1} &
b^{-1} + b^{-1}\rho (a-\sigma b^{-1}\rho)^{-1}\sigma b^{-1}
\end{pmatrix}.
\end{align*}
Let $\mathbf{M}$ be a supermatrix of the form \eqref{eq:appsusy:supermatrix} and $\Phi$ a supervector 
and  $\Phi^*$ its adjoint
\begin{align}\label{eq:appsusy:supervector}
\Phi =
\begin{pmatrix}
z \\
\chi 
\end{pmatrix}
\quad\text{ and } \quad
\Phi^* = ( \bar z,\bar\chi),
\end{align}
where $z\in\mathbb{C}^p,$  $\chi=(\chi_j)_{j=1}^q$ and $\bar\chi=(\bar\chi_j)_{j=1}^q$
are again independent families of Grassmann variables.
We can write the superdeterminant as a Gaussian integral 
\begin{align}
\label{eq:appsusy:sdetintegral}
\int \diff\Phi^*\diff\Phi \eto^{-(\bar\Phi, \mathbf{M} \Phi)} = \Sdet \mathbf{M^{-1}},
\end{align}
where $\diff\Phi^* \diff\Phi = \diff \bar\chi \diff \chi \diff \bar z\diff z$.
Below, we consider only the special case $p=q=1$.

\section{Proof of Lemma \ref{lem:strategy:susy}}
\label{sec:appproofsusy}
We combine \eqref{eq:appsusy:idea} and \eqref{eq:appsusy:gaus}
to rewrite the Green's function as a Gaussian integral.
Let  $\chi = (\chi_i)_{i\in\Lambda}$ and
$\bar\chi = (\bar\chi_i)_{i\in\Lambda}$ be two families of Grassmann variables,
 $z = (z_j)_{j\in\Lambda}\in\mathbb{C}^\Lambda$ and $\Phi = (\Phi_j)_{j\in\Lambda}$ and
$\Phi^* = (\Phi^*_j)_{j\in\Lambda}$ two sets of supervectors defined as in \eqref{eq:appsusy:supervector}.
Using the fact that $(-i(E_\varepsilon-H))^{-1}$ has positive definite Hermitian part, we write
\begin{align}
\sum_{k\in\Lambda}G^+_\Lambda (E_\varepsilon)_{kk}&=
-i (2 \pi)^{-|\Lambda|}\det[-i(E_\varepsilon - H)]
\int \diff  \bar z \diff  z
\eto^{i ( \bar z, (E_\varepsilon - H) z)}
\sum_{k\in\Lambda}z_k \bar{z}_k
\nonumber\\
&=
-i \int \diff \Phi^* \diff  \Phi 
\eto^{i \sum_{i,j\in\Lambda} (\Phi_i,(\delta_{ij}\mathbf{E}_\varepsilon-\mathbf{H}_{ij}) \Phi_j)} 
\sum_{k\in\Lambda}z_k \bar{z}_k,\label{eq:appproofsusy:help1}
\end{align}
where the product measure 
is defined
as in \eqref{eq:appsusy:sdetintegral}. Note that the bold face printed $\mathbf{E}_\varepsilon$ and
$\mathbf{H}_{ij}$ are $2\times 2$ supermatrices with diagonal entries $E_\varepsilon$ and $H_{ij}$,
respectively, and vanishing off-diagonal entries.
Since the contribution of the random matrix $H$ appears only in the exponential,
using a Hubbard-Stratonovitch transformation as in \cite[Lemma 1]{disertori-pinson-spencer}, 
we can rewrite the average over $H$ as
\begin{align}
\label{eq:appproofsusy:averageH}
\mathbb{E}\left[ \eto^{-i \sum_{i,j\in\Lambda}(\bar\Phi_i, \mathbf{H}_{ij} \Phi_j)} \right] =& 
\eto^{-\frac{1}{2} \sum_{i,j\in\Lambda} J_{ij} (\Phi_i^{*} \Phi_j^{}) (\Phi_j^{*} \Phi_i^{})}=
\eto^{-\frac{1}{2} \sum_{i,j\in\Lambda} J_{ij} \Str (\Phi_i^{} \Phi_i^{*}) (\Phi_j^{}\Phi_j^{*})}\\
=& \int  \prod_{j\in\Lambda} \diff \mathbf{M}_j
\eto^{- \frac{1}{2}\sum_{i,j\in\Lambda} J_{ij}^{-1}\Str[\mathbf{M}_i\mathbf{M}_j]}
\eto^{- i \sum_{j\in\Lambda}(\bar\Phi_j, \mathbf{M}_j \Phi_j)},\label{eq:appproofsusy:averageH1}
\end{align}
where
\begin{align*}
\mathbf{M}_j = \begin{pmatrix} a_j & \bar{\rho}_j \\ \rho_j & i b_j \end{pmatrix} 
\quad \text{and} \quad 
\diff \mathbf{M}_j \coloneqq \diff a_j \diff b_j \diff\bar\rho_j \diff \rho_j,
\end{align*}
$(a_j)_{j\in\Lambda}$ and $(b_j)_{j\in\Lambda}$ are families of real variables and
$(\rho_j)_{j\in\Lambda}$ and $(\bar{\rho}_j)_{j\in\Lambda}$ are two families
of Grassmann variables.

The expression for the observable  $-i\sum_{k\in\Lambda} \bar z_k z_k$ can be written as the derivative
$\sum_{k\in\Lambda} \partial_{a_k}$ of the second exponential in \eqref{eq:appproofsusy:averageH1}. Hence, 
 applying integration by parts in the variables $a_j,$  the derivative falls on the first exponential
in \eqref{eq:appproofsusy:averageH1} which yields 
$\sum_{k\in\Lambda}\sum_{l\in\Lambda} J^{-1}_{kl}a_l=\sum_{l\in\Lambda}a_{l}$.
Since the integral expression is still translation invariant in  $\Lambda$, by relabeling the indices
we can now substitute the sum $|\Lambda|^{-1}\sum_{l\in\Lambda} a_l $ by $a_0$. 
This step simplifies the integral compared to \cite[eq.(3.1)]{disertori-pinson-spencer}.

By \eqref{eq:appsusy:sdetintegral}, the integral over the supervector yields
\begin{align*}
\frac{1}{|\Lambda|}\sum_{k\in\Lambda}\mathbb{E}[ G^+_\Lambda(E_\varepsilon)_{kk} ] = &
\int \prod_{j\in\Lambda} \diff \mathbf{M}_j 
\eto^{- \sum_{i,j\in\Lambda}J_{ij}^{-1}\Str[\mathbf{M}_i  \mathbf{M}_j]}
\prod_{j\in\Lambda}\Sdet[\mathbf{E}_\varepsilon - \mathbf{M}_j]^{-1} a_0.
\end{align*}
Finally, we insert the expressions for the supermatrix $\mathbf{M}$ and perform the integration
over the Grassmann variables applying \eqref{eq:appsusy:gaus}. This proves \eqref{eq:strategy:dualrepsusy1}.

For  \eqref{eq:strategy:dualrepsusy2}, note that each $E$-derivative 
of \eqref{eq:appproofsusy:help1} results into 
$i\sum_{k\in\Lambda}\bar z_k z_k + \bar\chi_k \chi_k,$ which can be replaced by
$\sum_{k\in\Lambda}\Str \mathbf{M}_k$. Now, using the definition \eqref{eq:cluster:susygausmeasure}
with $B$ replaced by $J$, integration by parts in $\Str \mathbf{M}_k$ yields
\begin{align}\label{eq:appproofsusy:intbyparts}
\int \diff\mu_{J}(\mathbf{M}) \Str \mathbf{M}_j \mathcal{F}(\mathbf{M}) =
\sum_{k} J_{kj} \int\diff\mu_{J}(\mathbf{M})\Str \partial_{\mathbf{M}_j} \mathcal{F}(\mathbf{M}), 
\end{align}
where
\begin{align}\label{eq:appproofsusy:partialM}
\partial_{\mathbf{M}_j} = \begin{pmatrix} \partial_{a_j} &-\partial_{\rho_j}\\ 
\partial_{\bar\rho_j} & i\partial_{b_j}\end{pmatrix}
\end{align}
and $\mathcal{F}(\mathbf{M})$ is any smooth function
such that the integral above exists.

\section{Estimates of the covariance}
\label{sec:appcov}
Let $d=2$,  $\Lambda\subset \mathbb{Z}^2$  a finite cube,
$-\Delta_\Lambda^P$ the discrete Laplacian on $\Lambda$ with periodic boundary conditions and 
$-\Delta$ the discrete Laplacian on $\mathbb{Z}^2$. 
We consider the two covariances $C_{m}^{\Lambda}:= (-\Delta_\Lambda^P + m^2)^{-1}$ 
for the finite cube 
$\Lambda$ and $C^\infty_{m} = (-\Delta + m^{2})^{-1}$ for $\mathbb{Z}^2$, 
 with $m>0$.
We will prove the following result.
\begin{lemma} \label{lem:appcov:covariance}
The finite volume covariance $C_{m}^{\Lambda}$ satisfies
\begin{align}
\label{eq:appcov:upperbound}
0< (C_{m}^{\Lambda})_{ij} \leq
\begin{cases}
K \ln\left(\frac{1}{m(1+|i-j|)}\right) & \text{if } |i-j|\leq \frac{1}{m}\\
\frac{K}{(|i-j|m)^{1/2}} \eto^{-m|i-j|} & \text{if } |i-j| > \frac{1}{m},
\end{cases}
\end{align}
provided the mass is small  $0<m\ll 1$ and $m|\Lambda|^{1/2}>1$. 
Moreover for all $m<1$ the diagonal part satisfies
\begin{align*}
 (C_{m}^{\Lambda})_{jj} \geq (C^\infty_{m})_{jj} \geq K_{1} \ln (m^{-1})+ K_{2}
\end{align*}
for some constants $K_{1},K_{2}>0$ uniformly in $\Lambda$. 
\end{lemma}
\begin{rem} The decay for $J$ and $C$ in  
\eqref{eq:preliminary:expdecayC} follow directly from this result.
The same holds for the complex covariance $B$ since  $|B_{ij}|\leq C_{ij}$
(see \eqref{eq:appcov:Bestimate} below).
\end{rem}

\begin{proof}
First we establish a series expansion and write $C_m^\Lambda$ as a series in $C_m^\infty$.
Using the ideas of Salmhofer \cite{salmhofer} in the continuous case, we prove the 
desired decay for $C_m^\infty$. Finally we conclude that $C_m^\Lambda$ has the same decay.
\begin{description}[leftmargin=0cm]
\item[Step 1: Series expansion]
To compare $C_{m}^{\Lambda }$ and $C^\infty_{m}$, we can write the two Laplacians as 
$-\Delta_\Lambda^P= 4 \mathds{1}_\Lambda - N_\Lambda^P$, and $-\Delta= 4\mathds{1}_{\mathbb{Z}^2} - N_{\mathbb{Z}^2},$
where $\mathds{1}$ is the identity matrix on $\Lambda$ and $\mathbb{Z}^2$, respectively,
and $N$ is the matrix with entries $N_{ij}=1$ if $|i-j|=1$ and $N_{ij} = 0$ otherwise. 
Note that one uses the periodic distance $|\cdot|_P$ in the torus $\Lambda$ 
in the case of periodic boundary conditions.
The covariances can then be written as a series
\begin{align}
\label{eq:appcov:seriesexpansion}
 C_{m}^{\Lambda } =(D-N_\Lambda^P)^{-1} = \sum_{k=0}^\infty D^{-1} (N_\Lambda^P D^{-1})^k,
\quad
  C^\infty_{m} 
  = \sum_{k=0}^\infty 
 D^{-1}_{\mathbb{Z}^2} (N_{\mathbb{Z}^2} D^{-1}_{\mathbb{Z}^2})^k  
\end{align}
where $D=(4+m^2)\mathds{1}_\Lambda$ and $D_{\mathbb{Z}^2}= (4+m^{2})\mathds{1}_{\mathbb{Z}^2}$ 
are diagonal matrices.
 This representation is obtained by iterating the identity
  \begin{align*}
  (A+B)^{-1} - A^{-1} = - A^{-1} B (A+B)^{-1},
 \end{align*}
for matrices $A$ and $B$ with $A$ and $A+B$ invertible.
 To prove convergence, we use the structure of $N_\Lambda^P$ 
 and rewrite the sum as a sum over paths
 \begin{align*}
  \sum_{k=0}^\infty \left (D^{-1} (N_\Lambda^P D^{-1})^k\right )_{ij} 
  &= \sum_{k=0}^\infty \sum_{\gamma\in \Gamma_{ij}^\Lambda, |\gamma| = k} \lambda^{k+1}
 \leq \sum_{k=0}^\infty 4^k \lambda ^{k+1}
  = \frac{\lambda}{1-4\lambda} = \frac{1}{m^2}<\infty,
 \end{align*}
 where $\lambda = (4 + m^2)^{-1}$ and $\Gamma_{ij}^\Lambda$ is 
 the set of all paths from $i$ to $j$ in the torus $\Lambda$. 
 In the second step, we bound the number of paths from $i$ to $j$ of length $k$
 by the number of all paths of length $k$ starting in $i$, i.e. $4^k$.
 Therefore, the sum converges and is well-defined.
  By the same arguments, we can write
\begin{align*}
  (C^\infty_{m})_{ij}  ={ \sum}_{\gamma \in \Gamma_{ij}^{\mathbb{Z}^2}} \ \lambda^{|\gamma|+1}<\infty,
 \end{align*}
 where $\Gamma_{ij}^{\mathbb{Z}^2}$ is the set of all paths from $i$ to $j$ in $\mathbb{Z}^2$.
The formulas above imply directly that $(C_m^\Lambda),(C_m^\infty)_{ij}>0$. 
Now, we can identify each point $x\in\mathbb{Z}^2$ with a point $\tilde x = x + n |\Lambda|^{1/2}$ 
in $\Lambda$, where $n\in\mathbb{Z}^2$. 
By identifying each path $\gamma\in\Gamma_{ij}^{\mathbb{Z}^2}$ in $\mathbb{Z}^2$
with the corresponding path in $\tilde\gamma\in\Gamma_{ij}^\Lambda$ in the torus, 
we easily obtain $(C_{m}^\infty)_{ij} \leq (C^\Lambda_m)_{ij},$ hence
the first part of Lemma \ref{lem:appcov:covariance}. In order to get the inverse relation
we need to characterize the paths in $\Gamma^\Lambda_{ij}$, which cannot be identified
with paths in $\Gamma_{ij}^{\mathbb{Z}^2}$.
These are exactly the paths in $\Gamma^\Lambda_{ij}$ such that their corresponding paths in $\mathbb{Z}^2$ 
do not end in $j$ but rather in $j_{n} = j + n |\Lambda|^{1/2}$, with $n\in\mathbb{Z}^2\backslash \{0\}$.
These paths cross the boundary of the cube $\Lambda$ in such a way
that for at least one of the $2$ space dimensions the differences $n_1$ and/or $n_2$ of the 
number of crossings in positive and negative direction, respectively, is non-vanishing. 
We write
\begin{align}
\label{eq:appcov:expansionC}
(C_m^\Lambda)_{ij} 
= \sum_{n\in\mathbb{Z}^2}{\sum}_{\tilde\gamma\in\Gamma_{ij_{n}}^{\mathbb{Z}^2}} \ \lambda^{|\gamma|+1}
= \sum_{n\in\mathbb{Z}^2} (C_m^\infty)_{i j_{n}},
\end{align}
where $ j_{n} = j + n |\Lambda|^{1/2}$ as before.

\item[Step 2: Decay of $\mathbf{C_m^\infty}$]
We prove that $C_m^\infty$ has the desired decay \eqref{eq:appcov:upperbound}
The proof is similar to the proof in the continuous case in \cite[Lemma 1.10]{salmhofer}, 
but the expressions become more complicated in the discrete case. 
We give a sketch of the main steps.
By its Fourier representations, the covariance for $\mathbb{Z}^2$ can be  written as 
\begin{align}\label{eq:appcov:fourierCinfty}
(C_m^\infty)_{ij}  
= \int_{[-\pi,\pi]^2}
\frac{\eto^{i(k,i-j)}}{2\sum_{l=1}^2(1-\cos k_l)+m^2}\diff^2k,
\end{align}
where  $k = (2\pi)/|\Lambda|^{1/d} n$.
First, by rescaling $k\to mk$, we obtain 
\begin{align*}
(C_m^\infty)_{ij}  
=  \int_{\left[-\pi/m,\pi/m\right]^2}
\frac{\eto^{im(k,i-j)}}{2m^{-2}\sum_{l=1}^2
(1-\cos(mk_l))+1}\diff  ^2k.
\end{align*}
We can assume that $i_1-j_1\geq |i_{2}-j_{2}|>0$. Considering the integrand as a function in $k_1$, there are two poles 
\begin{align*}
k_1^\pm = \pm i \tfrac{1}{m}\arcosh 
\left(\tfrac{m^2}{2}+2-\cos \left(mk_2\right)\right) 
= \pm i r(k_2)
\end{align*}
of order one in the complex plain. Closing the integration contour for $k_1$ to the rectangle
with vertices $- m^{-1}\pi, m^{-1}\pi, m^{-1}\pi+iy$ and $- m^{-1}\pi+iy$ such that 
the sign $\sgn y (i_1-j_1) = 1$ and send $|y|\to \infty$, we can apply the residue theorem and 
the following integral in $k_2$ remains
\begin{align}
\label{eq:appcov:salmhofer}
(C_m^\infty)_{ij}  
= 2\pi \int_{-\pi/m}^{\pi/m}
\frac{\eto^{imk_2 (i_2-j_2)}\eto^{-m (i_{1}-j_{1}) r(k_2)}}{\frac{2}{m}
\sinh\arcosh\left(\frac{m^2}{2}+2-\cos (mk_2)\right)}\diff  k_2.
\end{align}
Using $\sinh\arcosh z = \sqrt{z^2-1}$ for $z\geq 1$, 
the absolute value of the integral can be bounded by
\begin{align*}
|(C_m^\infty)_{ij}  |\leq2\pi \int_{0}^{\pi/m}
\frac{\eto^{-r(k_2)t}}{\frac{1}{m}\sqrt{ \left(\frac{m^2}{2}
+2-\cos(m k_2)\right)^2-1}} \diff  k_2,
\end{align*}
where $t= m|i-j|/2$. Note that we showed above that $(C_m^\infty)_{ij}\geq 0$.
Let us assume $t\geq 1$ first.
The residual $r(k_2)$ is monotone increasing for $k_2\in\left[0,m^{-1}\pi\right]$ 
and bounded by
\begin{align*}
r(k_2)\geq 
\begin{cases}
r(0) + ck_2^2= 1 + O (m)+ ck_2^2 & \text{if } k_2 \leq 1,\\
r(0) + ck_2 = 1 + O (m) + ck_2 &\text{if } k_2 > 1,
\end{cases}
\end{align*}
where $c$ is independent of $m$ and $k_2$.
One can bound the square root in the denominator for all $k_2\in\left[0,m^{-1}\pi\right]$ by
\begin{align*}
\sqrt{ \left(\tfrac{m^2}{2}+2-\cos\left(m k_2\right)\right)^2-1}
\geq m.
\end{align*}
Therefore, the integral is bounded by
\begin{align*}
(C_m^\infty)_{ij}\leq& 2\pi \eto^{-t} \left(\int_0^1 \eto^{-tck_2^2} \diff  k_2 +
\int_1^{\pi/m}\hspace{-0,3cm}\eto^{-tck_2} \diff  k_2\right)\\\leq& 
  2\pi \eto^{-t}\left(\frac{1}{\sqrt{t}}\int_0^{\sqrt{t}} \hspace{-0,3cm}\eto^{-ck_2^2}\diff  k_2 +
 \frac{1}{t}\int_{t}^{\pi t/m}\hspace{-0,3cm} \eto^{-ck_2} \diff  k_2\right)\\
\leq & \frac{2\pi }{ \sqrt{t}} \eto^{-t} \left(\int_0^\infty \eto^{-ck_2^2} \diff  k_2 + 
\int_1^{\infty} \eto^{-ck_2} \diff  k_2\right)\leq  \frac{K}{\sqrt{t}}\eto^{-t},
\end{align*}
where in the last line we used $t\geq 1$.
This proves the second part of \eqref{eq:appcov:upperbound}.
In the case  $0<t\leq 1$, we perform first in \eqref{eq:appcov:salmhofer}
the change of variables   
\begin{align*}
s=r(k_2)t\quad \equiv \quad k_2 (s) = \tfrac{1}{m}
\arccos \left(2+\tfrac{m^2}{2}-\cosh \left(\tfrac{ms}{t}\right)\right).
\end{align*}
Inserting the Jacobian  
\begin{align*}
\tfrac{\diff  k_2}{\diff  s} = 
\tfrac{\sinh\left(\frac{ms}{t}\right)}{t\sqrt{1-\left(2+\frac{m^2}{2}-
\cosh\left(\frac{ms}{t}\right)\right)^2}},
\end{align*}
and repeating the arguments after \eqref{eq:appcov:salmhofer}, we obtain
\begin{align*}
(C_m^\infty)_{ij}& \leq  K \int_{s_0}^{s_1} 
\frac{\eto^{-s}}{\frac{t}{m}\sqrt{1-\left(2+\frac{m^2}{2}-
\cosh\left(\frac{ms}{t}\right)\right)^2}}\diff  s 
\sim K\int_t^\infty \hspace{-3pt}\frac{\eto^{-s}}{\sqrt{s^2-t^2}}\diff  s 
\sim K \ln t^{-1},
\end{align*}
where $s_0 = r(0)t$ and $s_1 =  r(\pi/m)t$, and we used again $m\ll 1$. 
It remains to consider the case $i=j$. 
Using the Fourier integral representation one can see  that  $(C^\infty_m)_{jj} \leq K \ln m^{-1}$,
hence we can change the bound for small distances to 
\begin{align*}
(C^\infty_m)_{ij} &\leq K \ln\left(\tfrac{1}{m(1+|i-j|)}\right) 
& \text{if } |i-j|\leq \tfrac{1}{m}
\end{align*}

\item[Step 3: Conclusion]
In order to estimate \eqref{eq:appcov:expansionC}, we divide the sum into two pieces:
\begin{align*}
(C_m^\Lambda)_{ij} = 
\sum_{n\in\mathbb{Z}^2} (C_m^\infty)_{i j_{n}} = \sum_{n\in\mathbb{Z}^2: |n|<2} (C_m^\infty)_{ij_{n}} + 
\sum_{n\in\mathbb{Z}^2:|n|\geq 2} (C_m^\infty)_{i j_{n}}.
\end{align*}
For the first sum, note that
$(C_m^\Lambda)_{ij}$ depends only on the distance $|i-j|_P$ and we can assume that the periodic distance 
is reached inside the cube, i.e. $|i-j|_P = |i-j|$. Then we can estimate 
$|i- j_{n}| \geq |i-j|$ and therefore each covariance $(C_m^\infty)_{i j_{n}}\leq (C_m^\infty)_{ij}$. 
Since the sum contains finitely terms, the first sum decays as $C_m^\infty$ with a modified constant
$K$ in front.
To control the second sum note that
$|i- j_{n}| \geq \max_k (n_k -1) |\Lambda|^{1/2}\geq|\Lambda|^{1/2}\geq m^{-1}$ for $|n|\geq 2$.
Extracting the desired decay from each $(C_m^\infty)_{ij_n}$, a fraction of the exponential decay
\eqref{eq:appcov:upperbound} remains in the sum that allows to perform the sum and yields
a constant. 

Finally, to prove the second part of Lemma \ref{lem:appcov:covariance},
we partition the integration region of \eqref{eq:appcov:fourierCinfty} into $\|k\|\leq 1,$ and   $\|k\|> 1.$ 
The integral over the second region is bounded below by a constant, while the 
the integral over the first region generates the $\ln m^{-1}$ contribution.
\qedhere
\end{description}
\end{proof}

\begin{rem}
For the case of a complex mass as in $B$, note that we can apply the same series expansion as in 
\eqref{eq:appcov:seriesexpansion}  and estimate the absolute value by
\begin{align}
\label{eq:appcov:Bestimate}
|B_{ij}| \leq  W^{-2}\sum_{k=0}^ \infty | (\tilde D^{-1}(N_\Lambda^P \tilde D^{-1})^k)_{ij}| \leq
W^{-2}\sum_{k=0}^ \infty  ((D^{-1}(N_\Lambda^P  D^{-1})^k)_{ij} = C_{ij},
\end{align}
where $\tilde D$ is a diagonal matrix with entries $4+(m_r^2+i m_i ^2)/W^2$ and $\re \tilde D = D$.
\end{rem}

\begin{lemma}
\label{lem:appcov:mass}
Let $C\in\mathbb{R}^{N\times N}$ be a real symmetric matrix such that $C^{-1}\geq c\Id$ as a quadratic 
form, for some $c>0$. Let $B = (C^{-1}+im\Id)^{-1},$ with $m\in \mathbb{R}.$  
Then, the  restriction of $B$ to any subset  $Y\subset \{1,\dots,N\}$, satisfies
$\re (B_Y)^{-1}\geq c \Id_Y$ for any choice of $m$.
\end{lemma}

\begin{proof}
Using Schur's complement, we can write 
\begin{align*}
\re B_Y^{-1} 
= C^{-1}_{YY} - C^{-1}_{YY^C}C^{-1}_{Y^CY^C}((C^{-1}_{Y^CY^C})^2+m^2)^{-1}C^{-1}_{Y^CY}.
\end{align*}
By assumption we have for all $v\in\mathbb{R}^{Y}$ and $w\in\mathbb{R}^{Y^C}$
\begin{align*}
(v, C^{-1}_{YY}v) + (v, C^{-1}_{YY^C}w) + (w,C^{-1}_{Y^CY}v) + (w, C^{-1}_{Y^CY^C}w)\geq c((v,v)+ (w,w)).
\end{align*}
Choosing $w= -C^{-1}_{Y^CY^C}((C^{-1}_{Y^CY^C})^2+m^2)^{-1}C^{-1}_{Y^CY}v$,
we obtain an even better bound than the desired result for $\re (B_Y)^{-1}$.
\end{proof}

\end{appendix}

\section*{List of symbols}
\begin{longtable}{ll}
$\Lambda$ & $\subset \mathbb{Z}^2$, discrete cube.\\
$H$ & $:\Lambda\times\Lambda\to\mathbb{C}$ random band matrix.\\
$W$ & band width. \\ 
$\bar\rho_\Lambda(E) $ & averaged density of states in finite volume $\Lambda$.\\
$G^+_{\Lambda}(z)$ & Green's function, $z\in\mathbb{C}$. \\
$\rho_{SC}(E)$ & Wigner's semicircle law.\\
$E_\varepsilon$ & $=E+i\varepsilon$ energy with imaginary part. \\
$J$ & initial covariance. \\
$\mathcal{I}$ & energy interval. \\
$\alpha$ & $\in (0,1)$, parameter entering in the definition of the reference volume \\& in the cluster expansion. \\

$a,b$ & $\in \mathbb{R}^\Lambda$ integration variables.\\
$a_s^\pm,b_s^\pm$ & saddle points. \\
$\mathcal{E}$ & $=\mathcal{E}_r-\mathcal{E}_i = \tfrac{E}{2}-i\sqrt{1-\tfrac{E^2}{4}}$, value of saddle point $a_s^+$.\\ 
$B$ & new complex covariance, obtained after contour deformation.\\
$C$ & new real covariance.\\
$\diff\mu_J(a,b)$ & Gaussian measure with covariance $J$. \\
$\mathcal{R}(a,b)$ & remainder in the functional integral after contour deformation.\\
$D$ & diagonal matrix depending on $a,b$.\\
$\mathcal{V}(a,b)$ & effective potential after contour deformation.\\
$\ensuremath{V(x)}$ & cubic Taylor remainder.\\
$\mathcal{O}(a,b)$ & local observable, later $\mathcal{O}_{m,n}(a,b)$.\\
$m_r,m_i$ & real and imaginary part of complex mass term $1-\mathcal{E}^2$ of $C$.\\
$I^s$ & $\subset \mathbb{R}^\Lambda \times \mathbb{R}^\Lambda$, partition of integration domain, $s=1,...5$.\\
$F^{m,n}_s$ & functional integral with local observable $\mathcal{O}_{m.m}$ restricted to $I^s$.\\

$M$ & $=(M_j)_{j\in\Lambda}$ set of $2\times 2$ supermatrices.\\
$\bar\rho_j,\rho_j)_{j\in\Lambda}$ & set of Grassmann variables.\\
$\diff\mu_B(M)$ &  Gaussian measure in both complex and Grassmann variables.\\
$\mathcal{V}(M)$ & effective potential depending on the supermatrix $M$.\\

$s_p$ & inductively introduced interpolation parameters.\\
$C(s)$ & interpolated real covariance $C(s)_{ij}=s_{ij}C_{ij}$.\\
$B(s)$ & interpolated complex covariance $(C(s)^{-1}+i\sigma_Em_i^2)^{-1}$.\\
$G_q(s)$ & propagator depending only on $s_1,\dots,s_q$.\\
$\tilde Y$ & $(\triangle_0,\tilde\triangle_1,\dots,\tilde\triangle_r)$ generalized polymer.\\
$T$ & ordered tree on generalized polymer $\tilde Y$.\\
$\triangle$ & cube in $\mathbb{Z}^2$ of size $W^2(\ln W)^\alpha$.\\
$\triangle_0$ & root cube containing $0$.\\
$\tilde\triangle$ & $=(\triangle,\triangle',\triangle'')$ generalized cube.\\
$i,k,k',j$ & $\in\mathbb{Z}^2$ indices summed over  $i\in\triangle',j\in\triangle'',k'\in\triangle,k\in$ ``old'' cubes.
\end{longtable}


\begin{thebibliography}{CFGK87}

\bibitem[And58]{anderson-1958}
P.~W. Anderson.
\newblock {Absence of Diffusion in Certain Random Lattices}.
\newblock {\em Phys. Rev.}, 109:1492--1505, Mar 1958.

\bibitem[AR95a]{AR1994}
A.~Abdesselam and V.~Rivasseau.
\newblock Trees, forests and jungles: a botanical garden for cluster
  expansions.
\newblock In {\em Constructive physics ({P}alaiseau, 1994)}, volume 446 of {\em
  Lecture Notes in Phys.}, pages 7--36. Springer, Berlin, 1995.

\bibitem[AR95b]{abdesselam-rivasseau}
A.~Abdesselam and V.~Rivasseau.
\newblock {Trees, forests and jungles: a botanical garden for cluster
  expansions}.
\newblock In {\em {Constructive physics ({P}alaiseau, 1994)}}, volume 446 of
  {\em {Lecture Notes in Phys.}}, pages 7--36. Springer, Berlin, 1995.

\bibitem[AW15]{AizenmanWarzel2015}
M.~Aizenman and S.~Warzel.
\newblock {\em {Random operators}}, volume 168 of {\em {Graduate Studies in
  Mathematics}}.
\newblock American Mathematical Society, Providence, RI, 2015.
\newblock Disorder effects on quantum spectra and dynamics.

\bibitem[BE15]{bao-erdoes-2015}
Z.~{Bao} and L.~{Erd{\"o}s}.
\newblock {Delocalization for a class of random block band matrices}.
\newblock 2015.
\newblock arXiv:1503.07510.

\bibitem[Ber87]{berezin}
F.~A. Berezin.
\newblock {\em {Introduction to superanalysis}}, volume~9 of {\em {Mathematical
  Physics and Applied Mathematics}}.
\newblock D. Reidel Publishing Co., Dordrecht, 1987.

\bibitem[BGP14]{georges-peche}
F.~Benaych-Georges and S.~P{\'e}ch{\'e}.
\newblock {Largest eigenvalues and eigenvectors of band or sparse random
  matrices}.
\newblock {\em Electron. Commun. Probab.}, 19:no. 4, 1--9, 2014.

\bibitem[BL76]{brascamp-lieb}
H.~J. Brascamp and E.~H. Lieb.
\newblock {On extensions of the {B}runn-{M}inkowski and
  {P}r{\'e}kopa-{L}eindler theorems, including inequalities for log concave
  functions, and with an application to the diffusion equation}.
\newblock {\em J. Functional Analysis}, 22(4):366--389, 1976.

\bibitem[Bry86]{brydgesleshouche}
D.~C. Brydges.
\newblock {A short course on cluster expansions}.
\newblock 1986.
\newblock Les Houche, Session XLIII.

\bibitem[CCGI93]{CCGI-1993}
G.~Casati, B.~V. Chirikov, I.~Guarneri, and F.~M. Izrailev.
\newblock {Band-random-matrix model for quantum localization in conservative
  systems}.
\newblock {\em Phys. Rev. E}, 48:R1613--R1616, Sep 1993.

\bibitem[CFGK87]{CFGK87}
F.~Constantinescu, G.~Felder, K.~Gawedzki, and A.~Kupiainen.
\newblock {Analyticity of density of states in a gauge-invariant model for
  disordered electronic systems.}
\newblock {\em J. Stat. Phys.}, 48:365, 1987.

\bibitem[CMI90]{CMI-1990}
G.~Casati, L.~Molinari, and F.~Izrailev.
\newblock {Scaling properties of band random matrices}.
\newblock {\em Phys. Rev. Lett.}, 64:1851--1854, Apr 1990.

\bibitem[Dis04]{disertori2004}
M.~Disertori.
\newblock {Density of states for {GUE} through supersymmetric approach}.
\newblock {\em Rev. Math. Phys.}, 16(09):1191--1225, 2004.

\bibitem[DPS02]{disertori-pinson-spencer}
M.~Disertori, H.~Pinson, and T.~Spencer.
\newblock {Density of states for random band matrices}.
\newblock {\em Comm. Math. Phys.}, 232(1):83--124, 2002.

\bibitem[Efe83]{efetov-adv}
K.~Efetov.
\newblock {Supersymmetry and theory of disordered metals}.
\newblock {\em Adv. in Phys.}, 32(1):53--127, 1983.

\bibitem[Efe97]{efetov}
K.~Efetov.
\newblock {\em {Supersymmetry in disorder and chaos}}.
\newblock Cambridge University Press, Cambridge, 1997.

\bibitem[FM91]{Mirlin-Fyodorov}
Y.~V. Fyodorov and A.~D. Mirlin.
\newblock {Scaling properties of localization in random band matrices: a
  {$\sigma$}-model approach}.
\newblock {\em Phys. Rev. Lett.}, 67(18):2405--2409, 1991.

\bibitem[Haa10]{haake}
F.~Haake.
\newblock {\em {Quantum signatures of chaos}}.
\newblock {Springer Series in Synergetics}. Springer-Verlag, Berlin, enlarged
  edition, 2010.
\newblock With a foreword by H. Haken.

\bibitem[KK08]{kirsch-2008}
W.~Kirsch and F.~Klopp.
\newblock {An invitation to Random {S}chr{\"o}dinger operators}.
\newblock {\em Panorama Synth{\`e}se}, pages 1--119, 2008.

\bibitem[LSZ08]{littelmann-sommrs-zirnbauer}
P.~Littelmann, H.-J. Sommers, and M.~R. Zirnbauer.
\newblock {Superbosonization of invariant random matrix ensembles}.
\newblock {\em Comm. Math. Phys.}, 283(2):343--395, 2008.

\bibitem[Mir00]{mirlin1991}
A.~D. Mirlin.
\newblock {Statistics of energy levels and eigenfunctions in disordered
  systems}.
\newblock {\em Phys. Rep.}, 326(5-6):259--382, 2000.

\bibitem[MPK92]{MPK-199}
S.~A. Molchanov, L.~A. Pastur, and A.~M. Khorunzhii.
\newblock {Limiting eigenvalue distribution for band random matrices}.
\newblock {\em Theoretical and Mathematical Physics}, 90(2):108--118, 1992.

\bibitem[Pch15]{Pchelin2015}
V.~Pchelin.
\newblock {Poisson statistics for random deformed band matrices with power law
  band width}, 2015.

\bibitem[PF92]{PasturFigotin1992}
L.~Pastur and A.~Figotin.
\newblock {\em Spectra of random and almost-periodic operators}, volume 297 of
  {\em Grundlehren der Mathematischen Wissenschaften [Fundamental Principles of
  Mathematical Sciences]}.
\newblock Springer-Verlag, Berlin, 1992.

\bibitem[PSSS16]{SchenkerPeledShamisSodin2016}
R.~Peled, J.~Schenker, M.~Shamis, and S.~Sodin.
\newblock {On the Wegner orbital model}, 2016.

\bibitem[Riv91]{rivasseau}
V.~Rivasseau.
\newblock {\em {From perturbative to constructive renormalization}}.
\newblock {Princeton Series in Physics}. Princeton University Press, Princeton,
  NJ, 1991.

\bibitem[Sal99]{salmhofer}
M.~Salmhofer.
\newblock {\em {Renormalization}}.
\newblock {Texts and Monographs in Physics}. Springer-Verlag, Berlin, 1999.
\newblock An introduction.

\bibitem[Sch09]{Schenker2009}
J.~Schenker.
\newblock {Eigenvector Localization for Random Band Matrices with Power Law
  Band Width}.
\newblock {\em Communications in Mathematical Physics}, 290(3):1065--1097,
  2009.

\bibitem[Sha13]{shamis2013}
M.~Shamis.
\newblock {Density of states for {G}aussian unitary ensemble, {G}aussian
  orthogonal ensemble, and interpolating ensembles through supersymmetric
  approach}.
\newblock {\em J. Math. Phys.}, 54(11), 2013.

\bibitem[Shc14]{Shcherbina2014}
T.~Shcherbina.
\newblock {On the Second Mixed Moment of the Characteristic Polynomials of 1D
  Band Matrices}.
\newblock {\em Comm. Math. Phys.}, 328(1):45--82, 2014.

\bibitem[Sod10]{sodin-2010}
S.~Sodin.
\newblock {The spectral edge of some random band matrices}.
\newblock {\em Ann. of Math. (2)}, 172(3):2223--2251, 2010.

\bibitem[SW80]{schafer-wegner-1980}
L.~Sch{\"a}fer and F.~Wegner.
\newblock {Disordered system with n orbitals per site: Lagrange formulation,
  hyperbolic symmetry, and {G}oldstone modes.}
\newblock {\em Z. Phys. B}, 38:113--126, 1980.

\bibitem[Weg79]{wegner-1979}
F.~Wegner.
\newblock {The mobility edge problem: continuous symmetry and a conjecture.}
\newblock {\em Z. Phys. B}, 35:207--210, 1979.

\bibitem[Weg81]{Wegner1981}
F.~Wegner.
\newblock {Bounds on the density of states in disordered systems}.
\newblock {\em Zeitschrift f{\"u}r Physik B Condensed Matter}, 44(1):9--15,
  1981.

\end{thebibliography}
\end{document}